\documentclass[conference]{IEEEtran}

\usepackage[noadjust]{cite}
\usepackage{amsmath,amssymb,amsfonts}
\usepackage{algorithmic}
\usepackage{graphicx}
\usepackage{textcomp}
\usepackage{xcolor}
\usepackage{csquotes}

\usepackage[colorlinks,allcolors=blue!50!black]{hyperref}

\usepackage{tikz}
\usetikzlibrary{shapes.geometric}

\usepackage{mathtools}

\usepackage{amsthm}
\usepackage[capitalize]{cleveref}

\crefname{theorem}{Theorem}{Theorems}
\crefname{lemma}{Lemma}{Lemmas}
\crefname{proposition}{Proposition}{Propositions}
\crefname{claim}{Claim}{Claims}
\crefname{subclaim}{Claim}{Claims}
\crefname{corollary}{Corollary}{Corollaries}
\crefname{subcorollary}{Corollary}{Corollaries}
\crefname{definition}{Definition}{Definitions}
\crefname{remark}{Remark}{Remarks}
\crefname{example}{Example}{Examples}
\crefname{question}{Question}{Questions}

\numberwithin{equation}{section}
\theoremstyle{plain}
\newtheorem{theorem}[equation]{Theorem}
\newtheorem{lemma}[equation]{Lemma}
\newtheorem{proposition}[equation]{Proposition}

\newtheorem{corollary}[equation]{Corollary}

\newtheorem{question}[equation]{Question}
\newtheorem*{question*}{Question}
\theoremstyle{definition}
\newtheorem{definition}[equation]{Definition}
\theoremstyle{remark}
\newtheorem{remark}[equation]{Remark}
\newtheorem*{remark*}{Remark}
\newtheorem{example}[equation]{Example}
\newtheorem*{example*}{Example}

\newtheorem*{notation*}{Notation}

\newcommand{\ot}{\otimes}
\newcommand{\op}{\oplus}

\newcommand{\bbN}{\mathbb{N}}
\newcommand{\bbZ}{\mathbb{Z}}
\newcommand{\bbQ}{\mathbb{Q}}
\newcommand{\bbR}{\mathbb{R}}
\newcommand{\bbC}{\mathbb{C}}

\newcommand{\bbF}{\mathbb{F}}

\newcommand{\group}[1]{\mathsf{#1}}
\newcommand{\SL}{\group{SL}}
\newcommand{\GL}{\group{GL}}
\newcommand{\Og}{\group{O}}
\newcommand{\Ug}{\group{U}}
\newcommand{\Sp}{\group{Sp}}
\newcommand{\Symg}{\group{S}}
\newcommand{\Eg}{\group{E}}

\newcommand{\class}[1]{\textnormal{\textbf{#1}}}
\newcommand{\NP}{\class{NP}}
\newcommand{\coAM}{\mathrm{co}\class{AM}}
\newcommand{\GIclass}{\class{GI}}
\newcommand{\TIclass}{\class{TI}}
\newcommand{\OCIclass}{\class{TOCI}}

\newcommand{\problem}[1]{\textnormal{\texttt{#1}}}
\newcommand{\GI}{\problem{GI}}
\renewcommand{\OE}{\problem{OE}}
\newcommand{\OCI}{\problem{OCI}\!}

\newcommand{\Tr}{\operatorname{Tr}}
\newcommand{\Id}{\operatorname{Id}}
\newcommand{\id}{\operatorname{id}}

\newcommand{\rep}[2]{\left({#1}; {#2}\right)}

\newcommand{\tworow}[2]{({#1};\,{#2})}

\newcommand{\tspace}[3]{{\textstyle\bigotimes\nolimits^{#2}_{#3}}{#1}}

\newcommand{\vtuple}[1]{\mathbf{#1}}
\newcommand{\ituple}[1]{\mathbf{#1}}

\usepackage{scalerel}

\DeclareMathOperator{\Sym}{Sym}

\DeclarePairedDelimiter\abs{\lvert}{\rvert}

\begin{document}

\title{Complexity theory of orbit closure intersection\\ for tensors: reductions, completeness, and\\ graph isomorphism hardness}
\author{\IEEEauthorblockN{Vladimir Lysikov}\IEEEauthorblockA{\textit{Faculty of Computer Science}\\\textit{Ruhr University Bochum}\\Bochum, Germany\\vladimir.lysikov@rub.de} \and \IEEEauthorblockN{Michael Walter}\IEEEauthorblockA{\textit{Faculty of Computer Science}\\\textit{Ruhr University Bochum}\\Bochum, Germany} \smallskip \IEEEauthorblockA{\textit{Faculty of Physics}\\\textit{LMU Munich}\\Munich, Germany\\michael.walter@lmu.de}}

\maketitle

\begin{abstract}
A wide range of natural computational problems in computer science, mathematics, physics, and other sciences amounts to deciding if two objects are equivalent.
Very often this equivalence is defined in terms of group actions.
A natural question is to ask \emph{when two objects can be distinguished by polynomial functions that are invariant under the group action}.
For finite groups, this is just the usual notion of equivalence, but for continuous groups such as the general linear groups it gives rise to a new notion, called \emph{orbit closure intersection}.
This new notion has recently seen substantial interest in the community, as it captures, among others, the graph isomorphism problem, noncommutative polynomial identity testing, null cone problems in invariant theory, equivalence problems for tensor networks, and the classification of multiparty quantum states.
Despite remarkable recent algorithmic progress in celebrated special cases, the computational complexity of general orbit closure intersection problems is currently quite unclear. In particular, tensors seem to give rise to the most difficult problems.

In this work we start a systematic study of orbit closure intersection problems from the complexity-theoretic viewpoint.
Our key contributions include:
\begin{itemize}
\item We define a complexity class $\OCIclass$ that captures the power of orbit closure intersection problems for general tensor actions.
\item We give an appropriate notion of \emph{algebraic reductions} that imply polynomial-time reductions in the usual sense, but are amenable to invariant-theoretic techniques.
\item We identify several natural tensor problems that are complete for $\OCIclass$, one of which is the equivalence of PEPS tensor networks considered by Acuaviva \emph{et al} (FOCS'23).
\item We show that the graph isomorphism problem can be reduced to these complete problems and hence~$\GIclass \subseteq \OCIclass$.
\end{itemize}
As such, our work establishes the first lower bound on the computational complexity of orbit closure intersection problems, and it explains the difficulty of finding unconditional polynomial-time algorithms beyond special cases, as has been observed in the recent literature.
\end{abstract}


\section{Introduction}
A common theme in computer science, mathematics, and many other sciences is to consider objects or structures up to some notion of ``isomorphism'' or ``equivalence''.
Very often the equivalence is defined by describing possible transformations of an object into an equivalent one.
Examples are ubiquitous:
graphs can be considered up to graph isomorphism;
curves, surfaces, and manifolds in general are considered up to homeomorphisms or diffeomorphism;
algebraic structures such as groups or rings always come with the corresponding notion of isomorphism;
subgroups in a given group can be described up to isomorphism or up to conjugacy;
for matrices and matrix spaces we have notions of similarity and conjugacy;
quantum states shared by multiple parties are considered up to operations performed locally by these parties;
etc.

Each such notion motivates a corresponding computational problem: decide if two given objects are indeed equivalent.
The most famous problem of this kind is the graph isomorphism problem.
This problem is known to lie in $\NP \cap \coAM$, but for many years resisted attempts to find a polynomial-time solution (currently asymptotically best algorithm for graph isomorphism is the recent quasipolynomial algorithm of Babai~\cite{DBLP:conf/stoc/Babai16}, see also~\cite{zbMATH07503184}).
It is conjectured to be an $\NP$-intermediate problem (although Babai's algorithm casts some doubt on this conjecture).
Many isomorphism problems for various discrete structures and finite algebraic structures (for example, finite groups represented by their Cayley tables) reduce to graph isomorphism.
In fact, the isomorphism problem for general finite algebraic structures with operations and relations given by listing all values is equivalent to graph isomorphism~\cite[\S15]{ZKT-survey}.
This leads to the introduction of the complexity class $\GIclass$, which contains all problems polynomial-time reducible to graph isomorphism.

\paragraph*{Group actions and notions of equivalence}
Isomorphism problems can often be stated in terms of group actions which describe allowed equivalence transformations, and indeed this is the case for all examples mentioned above.
For example, suppose we are given a pair of graphs with the vertex set $[n] = \{1, \dots, n\}$ (it is trivial to check if the graphs have the same number of vertices).
The symmetric group~$\Symg_n$, which consists of the bijections of~$[n]$, acts on the set of graphs with vertex set~$[n]$.
Then two graphs are isomorphic precisely when they lie in the same orbit of $\Symg_n$, that is, when there exists a permutation~$\sigma \in \Symg_n$ which transforms one graph into another.

For finite groups, the only meaningful way to formulate an equivalence relation in terms of a group action is to deem two objects equivalent if they can be related by the group action or, equivalently, if they have the same orbits.
But if we consider continuous groups acting on various spaces, such as general linear groups, the orthogonal, or the symplectic groups over~$\bbC$, which represent different notions of symmetries of a complex vector space, then there is a second way to define a meaningful equivalence relation.
It is possible that two points do not lie in the same orbit, but there exist two sequences, one in each orbit, that converge to the same point.
So in addition to the standard notion of equivalence we can consider a relaxed notion of equivalence under which two points are equivalent if their orbit closures contain a common point.
Remarkably, for a wide class of groups known as \emph{reductive algebraic groups}, which in particular contains the classical groups such as $\GL_n$, $\SL_n$, $\Og_n$, $\Sp_{2n}$ and products thereof, this second equivalence captures precisely the power of invariant functions to separate the orbits:
\emph{two points have intersecting orbit closures if and only if values of all continuous (or even polynomial) invariants on these points coincide}!

More formally, we consider a linear action of an algebraic group $\group{G}$ (such as the general linear group~$\GL_n$) on a finite-dimensional complex vector space~$X$.
The orbit of a point $x \in X$ is the set $\group{G}x = \{\group{g} \cdot x \mid \group{g} \in \group{G}\}$.
Two points $x, y \in X$ are \emph{equivalent} under the action of $\group{G}$ if~$y = \group{g} x$ for some~$\group{g} \in \group{G}$ (equivalently, the orbits $\group{G}x$ and $\group{G}y$ coincide).
The second, more relaxed, notion of equivalence is defined as follows.
Two points $x, y \in X$ are \emph{closure equivalent} if $\overline{\group{G}x} \cap \overline{\group{G}y} \neq \varnothing$ (in our setting the orbit closure $\overline{\group{G}x}$ is the same in Zariski and Euclidean topology).
We denote equivalence by $x \sim_{\group{G}} y$ and closure equivalence by $x \approx_{\group{G}} y$.
A result of Mumford states that for $\GL_n$ and other reductive groups, $x \approx_{\group{G}} y$ is equivalent to the fact that $F(x) = F(y)$ for every invariant polynomial~$F$ on~$X$ (or equivalently for any invariant continuous function).

A classical example shows the difference and relevance of both notions.
Consider the action~$\GL_n$ on~$\mathbb{C}^{n \times n}$ given by conjugation: $\group{g}\cdot x = \group{g} x \group{g}^{-1}$.
The orbits of this action can be described using the classical theory of Jordan normal forms.
That is, two matrices are equivalent if they have a common Jordan normal form.
However, note that every Jordan normal form contains in its orbit closure the diagonal matrix with the same eigenvalues.
Thus the description of the closure equivalence is much simpler:
two matrices are closure equivalent if and only if they have the same eigenvalues (counted with multiplicity) or, equivalently, the same characteristic polynomials.
This example highlights that closure equivalence is a meaningful relaxation of the traditional notion of equivalence that often has an easier description.
This is one of the insights behind Mumford's geometric invariant theory (GIT)~\cite{Mumford-GIT}, which is a powerful approach in mathematics to construct robust and well-behaved moduli spaces -- geometrical spaces that parameterize objects such as curves up to closure equivalence.

Since we have two notions of equivalence, there are two computational problems associated with the action of~$\group{G}$ on~$X$:
\begin{itemize}
  \item {\small\emph{Orbit equality:} given $x, y \in X$, decide if $x \sim_{\group{G}} y$.}
  \item {\small \emph{Orbit closure intersection:} given $x, y \in X$, decide if $x \approx_{\group{G}} y$.}
\end{itemize}
To study the complexity theory of these problems, we consider a sequence of groups $\group{G}_n$ acting on vector spaces $X_n$ depending on the problem size $n$.
For example, for the conjugation action, we have~$\group{G}_n = \GL_n$ and~$X_n = \bbC^{n\times n}$, where the problem size~$n$ determines the size of the matrices.
And for the graph isomorphism problem, we take~$n$ to be the number of vertices, the groups $\group{G}_n$ are symmetric groups~$\Symg_n$ identified with the subgroups of permutation matrices in~$\GL_n$, and we represent graphs by their adjacency matrices, which can be thought of as elements in~$X_n = \mathbb{C}^{n \times n}$.
In the latter case the group is finite and hence both computational problems coincide.

\paragraph*{Orbit equality and the complexity class $\TIclass$}
Many problems involving isomorphisms of algebraic structures or matrix spaces can be either directly presented as orbit equality problems, or reduced to them in a straightforward manner.
For example, the isomorphism problem for modules over a finitely generated algebra~\cite{DBLP:conf/issac/ChistovIK97,BROOKSBANK20084020,Sergeichuk-Canonical} is the orbit equality problem for the action of $\GL_n$ on the space $\left(\mathbb{F}^{n \times n}\right)^{\oplus m}$ of matrix tuples given by $$\group{g} \cdot (x_1, \dots, x_m) = (\group{g} x_1 \group{g}^{-1}, \dots, \group{g} x_m \group{g}^{-1}),$$ known as the \emph{simultaneous conjugation} action.
The study of the complexity theory of orbit equality problems for general linear groups, their products, and other classical groups was pioneered by Grochow and Qiao and has seen tremendous recent progress over the past years~\cite{DBLP:journals/siamcomp/GrochowQ23,DBLP:journals/toct/GrochowQ24,DBLP:conf/stacs/GrochowQT21,DBLP:conf/coco/GalesiGPS23,DBLP:conf/innovations/ChenGQTZ24,DBLP:conf/stoc/GrochowQ25,DBLP:conf/stoc/GrochowQSS25}.
Similarly to how graph isomorphism is central for the isomorphism problems of discrete structures, it has become clear that many orbit equality problems for classical groups and various algebraic isomorphism problems can be reduced to the basic problem of \emph{$3$-tensor isomorphism}, which is the orbit equality problem for the standard action of $\GL_{n_1} \times \GL_{n_2} \times \GL_{n_3}$ on 3-tensors in~$\mathbb{F}^{n_1} \otimes \mathbb{F}^{n_2} \otimes \mathbb{F}^{n_3}$,
\begin{align*}
    (\group{g}_1,\group{g}_1,\group{g}_3) \cdot t = (\group{g}_1 \ot \group{g}_2 \ot \group{g}_3) t.
\end{align*}
Accordingly, as an analogue of~$\GIclass$ the complexity class~$\TIclass$ has been introduced, which contains all problems that are polynomially reducible to the $3$-tensor isomorphism problem for tensors over a field~$\bbF$, and it plays a foundational role in the theory.
Recent work has also studied this complexity class when restricted to classical groups such as the orthogonal or unitary subgroups.
In addition, group equality problems involving symmetric groups and general linear groups over finite fields are studied in coding theory and cryptography in the context of code equivalence problems~\cite{DBLP:journals/corr/abs-2011-04611,DBLP:journals/dcc/ReijndersST24,DBLP:journals/dcc/DAlconzo24}.

\paragraph*{Orbit closure intersection, algorithms and complexity}
Recent work has made important algorithmic progress in special cases of the orbit closure intersection problem.
Mulmuley~\cite{GCT-V} investigated orbit closure intersection problems in the context of geometric complexity theory.
He provided a randomized efficient algorithm solving the orbit closure intersection problem for the simultaneous conjugation action, which was derandomized in~\cite{DBLP:conf/approx/ForbesS13}.
He also conjectured that all orbit closure intersection problems for explicit actions of reductive groups are efficiently solvable.
Another special case that attracted a lot of interest because of its relation to noncommutative polynomial identity testing is the \emph{left-right action} of $\SL_{n_1} \times \SL_{n_2}$ on matrix tuples in $(\bbF^{n_1 \times n_2})^{\oplus m}$ defined as
\[(\group{g}, \group{h}) \cdot (x_1, \dots x_m) = (\group{g}x_1\group{h}^\top,\dots,\group{g}x_n\group{h}^\top).\]
Identity testing for polynomials and rational functions in noncommutative variables can be reduced to the following special case of the orbit closure intersection problem.
\begin{itemize}
    \item \emph{Null cone problem:} given $x \in X$, decide if $x \approx_{\group{G}} 0$.
\end{itemize}
The orbit closure intersection problem for the left-right action is also efficiently solvable.
Known algorithms for the above problems use one of two different approaches:
invariant-theoretic algorithms~\cite{DBLP:journals/cc/IvanyosQS17,DM17,DBLP:journals/cc/IvanyosQS18,chatterjee2023noncommutativeedmondsproblemrevisited,DM20,DBLP:conf/soda/IvanyosQ23} are based on finding polynomial invariants which as discussed can separate orbit closures.
Optimization-based algorithms~\cite{DBLP:conf/stoc/Gurvits03,DBLP:conf/focs/GargGOW16,DBLP:journals/focm/GargGOW20,Hamada-Hirai,DBLP:conf/stoc/Allen-ZhuGLOW18}, first use geodesic convex optimization to try and move the input point as close to~$0$ as possible, which reduces the orbit closure intersection problem to an orbit equality for a compact subgroup by a result of Kempf and Ness~\cite{Kempf-Ness} (the first step suffices for the null cone problem).
Both approaches have their advantages and can also output witnesses for the intersection or separation of orbit closures in addition to solving the decision problem.
Moreover, it has been understood~\cite{garg2017algorithmic,DBLP:conf/focs/BurgisserFGOWW19,iwamasa2024algorithmic} that many algorithmic results can be generalized to so-called \emph{quiver representations}, which describe linear algebraic actions in terms of graphs~\cite{derksen2011combinatorics,baldoni2023horn}.
In fact, the orbit closure intersection problem for the~$\GL$-action on an arbitrary quiver can be reduced to the simultaneous conjugation action action on $k$-tuples of matrices~\cite{le1990semisimple}, and similar results reduce acyclic quivers to the left-right action~\cite{Derksen-Weyman,iwamasa2024algorithmic}.
The orbit closure intersection problem is also efficiently solvable for actions of abelian groups~$\GL_1^{\times k}$ (called \emph{tori})~\cite{DBLP:conf/coco/BurgisserDMWW21}.

Despite this important progress, however, not much is known about the computational complexity of the orbit closure intersection problem beyond the well-behaved actions discussed above.
Purely invariant theoretic approaches are unlikely to be efficient for general reductive group actions because polynomial invariants in general have high degree~\cite{DM-lowerbound} and are hard to compute~\cite{DBLP:conf/coco/GargIMOWW20}.
Similarly, the power of optimization-based analytic methods is currently unclear.
In~\cite{DBLP:conf/focs/BurgisserFGOWW19} a general framework of geodesic optimization is developed, which leads to the construction of algorithms solving null cone problems for all reductive group actions, but the parameters controlling the performance of their algorithms can be exponential in some cases~\cite{DBLP:conf/coco/FranksR21}.
Moreover, while this also offers a plausible path for generalizing the Kempf--Ness approach of \cite{DBLP:conf/stoc/Allen-ZhuGLOW18}, this approach quickly runs into $\GIclass$-hardness~\cite{DBLP:conf/innovations/ChenGQTZ24} or versions of the $abc$-conjecture~\cite{DBLP:conf/coco/BurgisserDM0W24}!
In fact, the complexity of orbit closure intersection is very much unclear even for concrete actions.
For example, no polynomial-time algorithms are known for the 3-tensor action of~$\SL_{n_1} \times \SL_{n_2} \times \SL_{n_3}$ on $\bbC^{n_1} \ot \bbC^{n_2} \ot \bbC^{n_3}$.
Another interesting action for which no efficient algorithms are known is the action of~$\GL_{n_1} \times \GL_{n_2}$ on tensors in~$\bbC^p \otimes \bbC^{n_1 \times n_1} \otimes \bbC^{n_2 \times n_2}$, arising in the study of tensor networks.
Equivalently, this action can be thought of as the ``double conjugation'' action of $\GL_{n_1} \times \GL_{n_2}$ on the tuples in $(\bbC^{n_1 n_2 \times n_1 n_2})^{\oplus p}$ via
\begin{multline*}
    (\group{g}, \group{h}) \cdot (x_1, \dots, x_p) \\ = \left((\group{g} \boxtimes \group{h}) x_1 (\group{g} \boxtimes \group{h})^{-1}, \dots, (\group{g} \boxtimes \group{h}) x_p (\group{g} \boxtimes \group{h})^{-1}\right),
\end{multline*}
where $\boxtimes$ denotes the Kronecker product of matrices.
These are very natural problems that arise in a plethora of applications, such as in the classification of entanglement for multipartite quantum states~\cite{PhysRevA.68.012103,Briand-moduli333,PhysRevLett.111.060502,SAWICKI201881,Slowik2020linkbetween},
in multilinear cryptography~\cite{DBLP:conf/tcc/JiQ0Y19}, in geometric complexity theory~\cite{GCT-V}, and in the theory and practice of tensor network states~\cite{DBLP:conf/focs/AcuavivaMNPS0W23}.

It is interesting to note that the simultaneous conjugation action corresponds to one-dimensional tensor networks (known as matrix product states or MPS), which are computationally tractable, while the ``double conjugation'' action is the relevant gauge group for two-dimensional tensor networks (known as projected entangled pair states or PEPS), for which many computationally problems are hard.
\emph{To summarize, while no polynomial-time algorithms are known for essentially any action other than the special cases discussed above, including not for essentially any nontrivial tensor action, neither are there any sharp hardness results}!

\subsection{Main results}\label{main:results}
In this work we start a systematic study of orbit closure intersection (OCI) problems from the complexity-theoretic viewpoint.
We work over~$\bbC$ as our techniques are invariant-theoretic and work best in characteristic 0 (this setting also captures all applications mentioned above).

\paragraph*{Orbit closure intersection problems}
To formally define orbit closure intersection as a decision problem, we specify a sequence $X_n$ of vector spaces and a sequence of groups $\group{G}_n$ such that $X_n$ has a structure of $\group{G}_n$-representation.
The parameter~$n$ is a positive integer which measures the size of the problem; we require that the dimension of $X_n$ is polynomially bounded in terms of~$n$.
More generally, we allow the parameter to be a tuple~$\ituple{n}=(n_1,\dots,n_d)$.
An instance of an OCI problem $\OCI\rep{X_n}{\group{G}_n}$ is then given by $n$ and a pair of elements $x, y \in X_n$ with rational coordinates.\footnote{More generally, we can allow coordinates from a field $\bbF$ such that its elements can be represented by bitstrings and arithmetic operations on them can be performed efficiently.}

\paragraph*{Tensors and the complexity class $\OCIclass$}
We focus on the OCI problems for tensors and tensor tuples.
We can specify a sequence of representations by specifying type of each tensor factor for each tensor in a tuple and varying the dimensions of the vector spaces involved.
This captures several important group actions, such as:
\begin{itemize}
    \item \emph{Simultaneous conjugation:} $\rep{\left(V \otimes V^*\right)^{\oplus p}}{\GL(V)}$;
    \item \emph{Left-right action on matrix tuples:} \[\rep{\left(V \otimes W\right)^{\oplus p}}{\SL(V) \times \SL(W)};\]
    \item \emph{Gauge equivalence for 2D PEPS tensor networks:}
        \[\rep{\left(V \otimes V^* \otimes W \otimes W^*\right)^{\oplus p}}{\GL(V) \times \GL(W)};\]
    \item \emph{SLOCC classification of tripartite quantum states:} \[\rep{V_1 \otimes V_2 \otimes V_3}{\SL(V_1) \times \SL(V_2) \times \SL(V_3)}.\]
\end{itemize}
More formally, we specify a sequence of representations by giving a formula $\mathcal{F}(V_1,\dots,V_m)$, where the variables $V_1,\dots,V_m$ correspond to vector spaces, and allowed operations are taking duals, direct sums, and tensor products.
Such a formula defines a sequence of \emph{tensor tuple representations} by setting $\group{G}_{\ituple{n}} = \GL_{n_1} \times \dots \times \GL_{n_m}$ and~$X_{\ituple{n}} = \mathcal{F}(\bbC^{n_1},\dots,\bbC^{n_m})$.
All examples above can be interpreted in this way.
For convenience, we denote the sequence by $\rep{\mathcal{F}(V_1,\dots,V_n)}{\GL(V_1) \times \dots \times \GL(V_n)}$;
We call this a \emph{tensor tuple sequence}.
The corresponding orbit closure intersection problem is denoted by \[\OCI\rep{\mathcal{F}(V_1,\dots,V_n)}{\GL(V_1) \times \dots \times \GL(V_n)}.\]
By analogy with $\GIclass$ and $\TIclass$, we define the \emph{complexity class $\OCIclass$}, which contains all decision problems that are polynomial-time Karp-reducible to some OCI problem for actions of general linear groups.

\paragraph*{Completeness of tensor OCI problems}
We prove that there are orbit closure intersection problems that are complete for $\OCIclass$.
Note that any such problem must involve tensors with both covariant and contravariant factors.
For example, while the action of $\GL(V)$ on $V^{\otimes k}$ is complete for~$\TIclass$ (if $k \geq 3$), it is trivial from the perspective of OCI (every orbit contains $0$ in the closure).
We show that the natural problem of orbit closure intersection for the gauge group of 2D PEPS tensor networks for qutrits is complete for $\OCIclass$ (\cref{thm:complete PEPS}):

\begin{theorem}\label{thm:complete-PEPS}
  Gauge equivalence of 2D PEPS tensor networks with physical dimension~$3$ is $\OCIclass$-complete.
\end{theorem}

Written in our notation, this problem is $\OCI\rep{\left(V \otimes V^* \otimes W \otimes W^*\right)^{\oplus 3}}{\GL(V) {\times} \GL(W)}$.
We note that the corresponding problem for 1D MPS tensor networks (simultaneous conjugaton) can be solved in polynomial-time.
This is rather suggestive as many problems that are computationally tractable for 1D tensor networks become computationally hard for 2D and higher dimensions, as discussed earlier.

We also reduce any problem in $\OCIclass$ to a problem of the form $\OCI\rep{\bigoplus_{k = 1}^q V^{\otimes a_k} \otimes (V^*)^{\otimes a_k}}{\GL(V)}$ where the spaces and dual spaces are ``balanced''.
When $a_1 = \dots = a_q = 1$ this correponds to the simultaneous conjugation action, for which the OCI problem can be solved in polynomial time.
In almost all remaining cases, we establish the following hardness result (\cref{cor:complete-balanced-GL}):

\begin{theorem}\label{thm:complete-balanced-GL}
    The problem~{\small$\OCI\rep{\bigoplus\limits_{k = 1}^q V^{\otimes a_k} \otimes (V^*)^{\otimes a_k}\!}{\GL(V)}$} is $\OCIclass$-complete if $\sum_{k = 1}^q a_k \geq 4$ and at least one of $a_k$ exceeds $1$.
\end{theorem}

There are three remaining cases that are not handled by the above result and discussion, namely, {$\OCI\rep{V^{\otimes 2} \otimes (V^*)^{\otimes 2}}{\GL(V)}$}, {$\OCI\rep{V^{\otimes 3} \otimes (V^*)^{\otimes 3}}{\GL(V)}$}, and {$\OCI\rep{V^{\otimes 2} \otimes (V^*)^{\otimes 2} \oplus V \otimes V^*}{\GL(V)}$}.
We conjecture that these are also $\OCIclass$-complete.

\paragraph*{Completeness for the orthogonal, symplectic, and unitary groups}
Additionally, we study the OCI problems for tensor tuple sequences with the action of orthogonal (over $\bbC$, so there are nontrivial orbit closures) and symplectic groups or their products.
We denote the corresponding complexity classes by $\OCIclass_\Og$ and $\OCIclass_\Sp$, respectively.
We show that the complexity of these problems is at the same level as for general linear groups and identify several complete problems, both for the orthogonal and symplectic groups as well as for the general linear groups (\cref{subsec:o sp}):

\begin{corollary}
We have that $\OCIclass_\Og = \OCIclass_\Sp = \OCIclass$.
Moreover, the following orbit closure problems are complete for this class:
    \begin{itemize}
        \item $\rep{(V^{\otimes 3})^{\oplus 2}}{\Og(V)}$ and $\rep{(V^{\otimes 3})^{\oplus 2}}{\Sp(V)}$ ;
        \item $\rep{(V^{\otimes 3})^{\oplus 2} \oplus (V^*)^{\otimes 2}}{\GL(V)}$;
        \item $\rep{V^{\otimes 7}}{\Og(V)}$ and $\rep{V^{\otimes 7}}{\Sp(V)}$;
        \item $\rep{V^{\otimes 7} \oplus (V^*)^{\otimes 2}}{\GL(V)}$.
    \end{itemize}
\end{corollary}

Interestingly, for the tensor isomorphism classes it is only known that $\TIclass_\Og, \TIclass_\Sp \subseteq \TIclass$, while the converse inclusion is open~\cite[Question~10]{DBLP:conf/innovations/ChenGQTZ24}.

Furthermore, we show that tensor isomorphism for the \emph{unitary} group can be reduced to complete problems for~$\OCIclass$ as well as for~$\TIclass$ (\cref{thm:U vs GL}):

\begin{theorem}
We have $\TIclass_\Ug \subseteq \OCIclass$ as well as $\TIclass_\Ug \subseteq \TIclass$.
\end{theorem}

The latter answers \cite[Question~1.10]{DBLP:conf/innovations/ChenGQTZ24}.
Because the unitary groups are not algebraic this requires different techniques from geometric invariant geometry~\cite{DBLP:conf/coco/BurgisserDMWW21}.

\paragraph*{Graph isomorphism-hardness of tensor problems}
Finally, we show that the graph isomorphism problem lies in~$\OCIclass$ (\cref{thm:gi}):

\begin{theorem}\label{thm:gi-hardness}
    The problem~{\small$\OCI\rep{V^{\otimes 2} \oplus V^{\otimes 3} \oplus (V^*)^{\otimes 2}}{\GL(V)}$} is $\GIclass$-hard.
    Thus, $\GIclass \subseteq \OCIclass$.
    In particular, all $\OCIclass$-complete problems listed above are $\GIclass$-hard.
\end{theorem}

To the best of our knowledge, this is the first hardness result for an orbit closure intersection problem.
We sketch the reductions involved in its  proof in \cref{subsec:sketch}.
We also note that Mulmuley has conjectured that OCI problems for arbitrary connected reductive groups are in~$\class{P}$~\cite{GCT-V}.
If so, \cref{thm:gi-hardness} would imply that the graph isomorphism problem is also in~$\class{P}$.

\subsection{Ideas and techniques}

Following~\cite{DBLP:journals/siamcomp/GrochowQ23}, we only consider a special class of reductions between OCI problems, which we call \emph{efficient algebraic reductions}.
These are kernel reductions in the sense of~\cite{DBLP:journals/iandc/FortnowG11} implemented by polynomial maps.
Most of the reductions between orbit equality problems in~\cite{DBLP:journals/siamcomp/GrochowQ23} are constructed in such a way that the output always lies in the nullcone, that is, is closure equivalent to $0$.
As such, these cannot be used as reductions between the corresponding OCI problems.
This also highlights the fact that even representations with the trivial closure equivalence problem (such as the one for the action of $\GL(V)$ on $V^{\otimes k}$) can still have a very intricate orbit structure.
On the other hand, while some of our reductions can be adapted to work on associated orbit equality problems, other do not work on the level of orbits, but only orbit closures.
It seems that very different techniques are required for the analysis of OCI problems compared to orbit equality problems.

Consider groups $\group{G}$ and $\group{H}$ acting on vector spaces $X$ and $Y$ respectively.
We say that a polynomial map $\alpha \colon X \to Y$ \emph{preserves closure equivalence} if the implication \[x \approx_{\group{G}} y  \Rightarrow \alpha(x) \approx_{\group{H}} \alpha(y)\] holds for every $x, y \in X$, and that $\alpha$ \emph{reflects closure equivalence} if the opposite implication \[\alpha(x) \approx_{\group{H}} \alpha(y)  \Rightarrow x \approx_{\group{G}} y\] holds for every $x, y \in X$.

An efficient algebraic reduction from $\OCI\rep{X_n}{\group{G}_n}$ to $\OCI\rep{Y_n}{\group{H}_n}$ is an efficiently computable sequence of polynomial maps $\alpha_n \colon X_n \to Y_{q(n)}$ such that $q(n)$ is polynomially bounded in terms of $n$, and each $\alpha_n$ preserves and reflects closure equivalence, that is, $x \approx_{\group{G}_n} y$ if and only if $\alpha_n(x) \approx_{\group{H}_{q(n)}} \alpha_n(y)$.
We require that $\alpha_n$ are polynomials with rational coefficients, so that the reduction always transforms input vectors with rational coordinates into output vectors with rational coordinates.
Our results do not depend much on the notion of efficient computability used in this definition.
In fact, in all reductions the maps~$\alpha_n$ always can be implemented by tuples of monomials.

We use invariant-theoretic methods to construct our reductions.
If a group $\group{G}$ acts on $X$, then a polynomial $F \in \bbC[X]$ is \emph{invariant} if $F(x) = F(\group{g} \cdot x)$ for all $x \in X$ and $\group{g} \in \group{G}$, that is, $F$ is constant on orbits.
Polynomial invariants are in addition constant on orbit closures, because polynomials are continuous (with respect to Zariski topology).
The set of all invariant polynomials is closed under algebraic operations, so it is a subalgebra of $\bbC[X]$, denoted by $\bbC[X]^{\group{G}}$.
The groups we consider ($\GL_n$, $\Og_n$, $\Sp_{2n}$, $\Symg_n$ and their products) belong to the class of algebraic groups called \emph{reductive groups}.
For reductive groups the relation between polynomial invariants and closure equivalence is especially close, as shown by the following theorem.
\begin{theorem}[{\cite{Mumford-GIT,Derksen-Kemper}}]
    Let $X$ be a representation of a reductive group $\group{G}$ and let $x,y\in X$.
    Then $x \approx_{\group{G}} y$ if, and only if, $F(x) = F(y)$ for all polynomial invariants $F \in \bbC[X]^{\group{G}}$.
\end{theorem}
\noindent In other words, invariants separate orbit closures.
More generally, we can define a separating set of invariants as follows (for reductive groups this definition is equivalent to the definition of a separating set given in~\cite{Derksen-Kemper}).
\begin{definition}
    Let $X$ be a representation of $\group{G}$.
    A set of invariants $\Gamma \subset \bbC[X]^{\group{G}}$ is a \emph{separating set} if $x \approx_{\group{G}} y$ if, and only if, $F(x) = F(y)$ for all $F \in \Gamma$.
\end{definition}
As an example, the coefficients of the characteristic polynomial of a square matrix form a separating set of invariants for the conjugation action on the space of square matrices.

The idea that invariants can be used to construct algorithms for equivalence problems is well-known~\cite{DBLP:journals/siamcomp/BlassG84} (in this setting, a more general notion of a separating set is called a \emph{complete invariant}).
Because of the special position of invariant polynomials in our setting, we can also use them to guide the construction of reductions between orbit closure intersection problems.
Instead of using gadgets to relate closure equivalence for different spaces directly, the reductions will provide tools to construct gadgets relating invariants.
More specifically, to establish reductions, we need to
prove that certain polynomial maps preserve and reflect closure equivalence.
To do this, we use the following key technical lemmata which characterize these properties in terms of invariants.
They can be found as \cref{lem:closure equivalence vs invariants,thm:reflect} in the body of the paper.

\begin{lemma}
    Let $X$ and $Y$ be representations of groups $\group{G}$ and $\group{H}$ respectively, with $\group{H}$ reductive.
    Then a polynomial map $\alpha \colon X \to Y$ preserves closure equivalence if, and only if, for every invariant $F \in \bbC[Y]^{\group{H}}$ the pullback $F \circ \alpha$ is an invariant in $\bbC[X]^{\group{G}}$.
\end{lemma}

\begin{lemma}\label{mainlemma:reflect}
    Let $X$ and $Y$ be representations of groups $\group{G}$ and $\group{H}$ respectively, with $\group{H}$ reductive.
    Suppose a polynomial map $\alpha \colon X \to Y$ preserves closure equivalence.
    The following statements are equivalent:
    \begin{itemize}
        \item[(a)] $\alpha$ reflects closure equivalence;
        \item[(b)] the set $\{ F \circ \alpha \mid F \in \bbC[Y]^{\group{H}} \} \subseteq \bbC[X]^{\group{G}}$ is separating;
        \item[(c)] for some separating set $\Gamma \subseteq \bbC[Y]^{\group{H}}$, the set \[\Gamma \circ \alpha = \{ F \circ \alpha \mid F \in \Gamma \} \subseteq \bbC[X]^{\group{G}}\] is separating;
        \item[(d)] for every separating set $\Gamma \subseteq \bbC[Y]^{\group{H}}$, the set $\Gamma \circ \alpha$ is separating.
    \end{itemize}
\end{lemma}

It is usually not hard to construct maps so that they preserve closure equivalence.
The main part of our proofs is usually concerned with establishing the second part, for which we use the description of invariants for $\GL_n$ and reductive groups on tensor tuple representations in terms of tensor contractions.
Note that for mixed tensors over $V$ (elements of tensor products with factors $V$ and $V^*$) operations of tensor product and tensor contraction (that is, application of the trace map $\Tr \colon V \otimes V^* \to \bbC$ to two factors of a tensor) are $\GL(V)$-equivariant.
As a consequence, if from a tuple of tensors we form a tensor product with equal numbers of factors of type $V$ and $V^*$ and then contract them in some order, the resulting expression is invariant under the action of $\GL(V)$.
We call these \emph{contraction invariants}.
The following result is a main ingredient in the construction of many of our reductions (\cref{thm:contraction-invariants-multiple}).

\begin{theorem}\label{main:contraction-invariants}
    The invariant algebra of a tensor tuple representation is spanned by contraction invariants.
\end{theorem}
For $\GL(V)$ this description is classical and can be viewed as a reformulation of the classical first fundamental theorem of invariants~\cite{Weyl-book}.
Similar description of invariants also hold for reductive subgroups of $\GL(V)$.
In this case the contraction invariants can contain a finite number of fundamental invariant tensors alongside the argument tensors.
For $\Og(V)$ and $\Sp(V)$ these are the tensors in $V \otimes V$ and $V^* \otimes V^*$ representing invariant bilinear forms, for $\SL(V)$ these are tensors representing determinants as multilinear functions of several vectors in $V$ and $V^*$, and for the symmetric group $\Symg_n$ the fundamental invariants are the bilinear form $\sum_{i = 1}^n e_i^{*} \otimes e_i^* \in \bbC^{n*} \otimes \bbC^{n*}$ and the diagonal tensor $\sum_{i = 1}^n e_i^{\otimes 3} \in (\bbC^n)^{\otimes 3}$.
The descriptions for other classical groups also follow from the corresponding first fundamental theorems, and the case of general reductive groups is considered in~\cite{Sch08-first-fundamental-MR2379100,DM23-wheeled-props-MR4568127}.

To prove that a polynomial map $\alpha \colon X \to Y$ between tensor tuple representations of the groups $\group{G}$ and $\group{H}$ reflects closure equivalence we most often use \cref{mainlemma:reflect} and \cref{main:contraction-invariants} as follows.
We consider an arbitrary contraction invariant $G$ on $X$ and represent it as a multiple of $F \circ \alpha$ for some contraction invariant $F$ on $Y$.
The main part of the proof is the construction of the invariant $F$ as a tensor contraction given $G$ as a tensor contraction.
This construction in fact gives a condition much stronger than required by \cref{mainlemma:reflect}, namely, for our reductions the pullback map $F \mapsto F \circ \alpha$ is a surjection from $\bbC[Y]^{\group{H}}$ to $\bbC[X]^{\group{G}}$.

We note that similar surjective maps relating invariant algebras for representations of various groups are used in~\cite{DM-lowerbound} to prove exponential lower bounds on the degree of separating invariants, but for this the surjections between invariant algebras must be degree-preserving, which is a much stronger condition than needed for our goals.

The description of invariants given in~\Cref{main:contraction-invariants} is only valid in characteristic $0$, which is why we restrict our attention to group actions over $\mathbb{C}$.
In prime characteristic some of the simpler reductions can be proven using different methods, but the general picture remains unclear.

\subsection{Sketch of reduction from graph isomorphism to 2D PEPS}\label{subsec:sketch}

In this section we will illustrate our techniques and sketch how the graph isomorphism problem can be reduced to the orbit closure intersection problem for the simultaneous ``double conjugation'' corresponding to the gauge group of 2D PEPS tensor networks.

As explained earlier, the graph isomorphism problem can be modeled in terms of the action of the symmetric group $\Symg_n$ on matrices in $\bbC^{n \times n}$ or, equivalently, on tensors in $(\bbC^n)^{\otimes 2}$.
Since $\Sym_n$ is a finite group, its orbits are finite and, therefore, closed.
Thus, $\OCI\rep{(\bbC^n)^{\otimes 2}}{\Symg_n}$ captures graph isomorphism.
Our goal is to reduce this problem to $\OCI\rep{(V \otimes V^* \otimes W \otimes W^*)^{\oplus 8}}{\GL(V) \times \GL(W)}$, which describes the equivalence of 2D PEPS tensor networks with physical dimension~8.

As discussed, our approach is based on constructing algebraic reductions that relate the invariants for these groups and actions.
In particular, we will be able to obtain any arbitrary polynomial invariant of graphs from a corresponding invariant for the PEPS tensor action.
Consider, e.g.,
\begin{align*}
    F = \sum\limits_{\text{$v$ vertex}} (\deg v)^2,
\end{align*}
which is an interesting invariant because it can be used to distinguish, e.\,g., certain graphs that are \emph{non-isomorphic but isospectral} (that is, their adjacency matrices have the same eigenvalues).
For example, it takes different values on the two non-isomorphic isospectral connected graphs on $6$ vertices.
This graph invariant is captured by a $\Symg_n$-invariant polynomial $F(x) = \sum_{i = 1}^n (\sum_{j = 1}^n x^{ij} )^2$ if we substitute for $x$ the adjacency matrix of the graph.
Denote $V = \bbC^n$.
In general, the invariant algebra $\bbC[V^{\otimes 2}]^{\Symg_n}$ is spanned by contraction invariants, which are polynomials obtained from a tensor $x \in V^{\otimes 2}$ by completely contracting tensor products of copies of~$x$ and of the diagonal tensors~$g = \sum_{i = 1}^n e_i^* \otimes e_i^* \in (V^*)^{\otimes 2}$ and~$h = \sum_{i = 1}^n e_i^{\otimes 3} \in V^{\otimes 3}$ (these two tensors can be used to construct diagonal tensors of all orders, which characterize $\Symg_n$, see~\cref{lem:sym invariants} below).
Our invariant $F$ is such a contraction invariant: we can present it as a tensor contraction
\[F(x) = x^{ab} x^{cd} h^{eij} h^{fkl} g_{ac} g_{be} g_{df} g_{ij} g_{kl} \]
using Einstein summation notation, which can be seen by first noting that the contraction $g_{ab} g_{cd} h^{bcd}$ appearing twice in the above expression represents the covector $u = \sum_{i = 1}^n e_i^*$, which is used to obtain from $x$ the vector with coordinates $\sum_{j = 1}^n x^{ij} = x^{ij} u_j$.

\begin{figure}[ht]
\centering
\begin{tikzpicture}[>=latex]
  \node[draw,regular polygon,regular polygon sides=4] (x1) at (1.2,0) {$x$};
  \node[draw,regular polygon,regular polygon sides=3,shape border rotate=-30] (h1) at (-1.2,-1) {$h$};
  \node[draw,circle] (g1) at (0,0) {$g$};
  \node[draw,circle] (g2) at (0,-1) {$g$};
  \draw[->] (x1.west)--(g1.east);
  \draw[->] (h1.north east)--(g1.west);
  \draw[->] (h1.east)--(g2.west);
  \draw[->] (h1.south east)--(g2.south west);
  \draw[color=gray,fill] (x1.south)+(0,0.08) circle (0.02);
  \draw[color=gray,fill] (x1.west)+(0.08,0.03) circle (0.02);
  \draw[color=gray,fill] (x1.west)+(0.08,-0.03) circle (0.02);
  \node[draw,regular polygon,regular polygon sides=4] (x2) at (2.4,0) {$x$};
  \node[draw,regular polygon,regular polygon sides=3,shape border rotate=30] (h2) at (4.8,-1) {$h$};
  \node[draw,circle] (g3) at (3.6,0) {$g$};
  \node[draw,circle] (g4) at (3.6,-1) {$g$};
  \draw[->] (x2.east)--(g3.west);
  \draw[->] (h2.north west)--(g3.east);
  \draw[->] (h2.west)--(g4.east);
  \draw[->] (h2.south west)--(g4.south east);
  \draw[color=gray,fill] (x2.south)+(0,0.08) circle (0.02);
  \draw[color=gray,fill] (x2.east)+(-0.08,0.03) circle (0.02);
  \draw[color=gray,fill] (x2.east)+(-0.08,-0.03) circle (0.02);
  \node[draw,circle] (g5) at (1.8,-1.5) {$g$};
  \draw[->] (x1.south)--(g5.north west);
  \draw[->] (x2.south)--(g5.north east);
\end{tikzpicture}
  \caption{Diagrammatic representation of the invariant $F\nobreak\in\nobreak\bbC[V^{\otimes 2}]^{\Symg_n}$. We use different numbers of small dots to indicate the two tensor factors (indices) of the matrix~$x$. The other tensors are symmetric.}
  \label{fig:diagram-1}
\end{figure}
The index notation is good for discussing specific contractions of specific tensors with known number of factors, but because we often consider arbitrary contraction invariants on arbitrary tensor tuple representations, we will use the notation with explicit contractions.
The symbol $\Tr^{p_1 p_2 \dots p_k}_{q_1 q_2 \dots q_k}$ will denote the contraction of a tensor where the $p_i$-th contravariant factors are contracted with $q_i$-th covariant factors.
In this notation, the above contraction reads
\[F(x) = \Tr^{1,2,3,4,5,6,7,8,9,10}_{1,3,2,5,4,7,8,6,9,10} \left(x^{\otimes 2} \otimes h^{\otimes 2} \otimes g^{\otimes 5}\right).\]
In case all factors are contracted, we also write $\Tr_\pi$ where $\pi$ is a permutation of $m$ indices for $\Tr^{\phantom{\pi(}1\phantom{)},\dots,\phantom{\pi(}m\phantom{)}}_{\pi(1),\dots,\pi(m)}$.
Both notations are unwieldy already in this relatively simple case.
A diagrammatic representation of this contraction is presented on Figure~\ref{fig:diagram-1}.

We will now describe the reduction from $\OCI\rep{(\bbC^n)^{\otimes 2}}{\Symg_n}$ to $\OCI\rep{V^{\otimes 2} \oplus V^{\otimes 3} \oplus (V^*)^{\otimes 2}}{\GL(V)}$, thus proving \cref{thm:gi-hardness}.
Let $V = \bbC^n$ and consider the map
\begin{align*}
  \alpha \colon V^{\otimes 2} &\to V^{\otimes 2} \oplus V^{\otimes 3} \oplus (V^*)^{\otimes 2} \\
  \alpha(x) & = (x, h, g)
\end{align*}
with $g$ and $h$ as above.
Since $g$ and $h$ are invariant under $\Symg_n$, for every invariant $F \in \bbC[V^{\otimes 2} \oplus V^{\otimes 3} \oplus (V^*)^{\otimes 2}]^{\GL(V)}$ the polynomial $F' \circ \alpha(x) = F'(x, g, h)$ is an invariant in $\bbC[V^{\otimes 2}]^{\Symg_n}$, therefore, $\alpha$ preserves closure equivalence.\\
\begin{figure}[ht]
\centering
\begin{tikzpicture}[>=latex]
  \node[draw,regular polygon,regular polygon sides=4] (x1) at (1.2,0) {$x$};
  \node[draw,regular polygon,regular polygon sides=3,shape border rotate=-30] (h1) at (-1.2,-1) {$y$};
  \node[draw,circle] (g1) at (0,0) {$z$};
  \node[draw,circle] (g2) at (0,-1) {$z$};
  \draw[->] (x1.west)--(g1.east);
  \draw[->] (h1.north east)--(g1.west);
  \draw[->] (h1.east)--(g2.west);
  \draw[->] (h1.south east)--(g2.south west);
  \draw[color=gray,fill] (x1.south)+(0,0.08) circle (0.02);
  \draw[color=gray,fill] (x1.west)+(0.08,0.03) circle (0.02);
  \draw[color=gray,fill] (x1.west)+(0.08,-0.03) circle (0.02);
  \draw[color=gray,fill] (h1.north east)+(-0.08,0) circle (0.02);
  \draw[color=gray,fill] (h1.east)+(-0.08,0.03) circle (0.02);
  \draw[color=gray,fill] (h1.east)+(-0.08,-0.03) circle (0.02);
  \draw[color=gray,fill] (h1.south east)+(-0.08,0.05) circle (0.02);
  \draw[color=gray,fill] (h1.south east)+(-0.08,0) circle (0.02);
  \draw[color=gray,fill] (h1.south east)+(-0.08,-0.05) circle (0.02);
  \draw[color=gray,fill] (g1.east)+(-0.08,0) circle (0.02);
  \draw[color=gray,fill] (g1.west)+(0.08,0.03) circle (0.02);
  \draw[color=gray,fill] (g1.west)+(0.08,-0.03) circle (0.02);
  \draw[color=gray,fill] (g2.west)+(0.08,0) circle (0.02);
  \draw[color=gray,fill] (g2.south west)+(0.02,0.06) circle (0.02);
  \draw[color=gray,fill] (g2.south west)+(0.06,0.02) circle (0.02);
  \node[draw,regular polygon,regular polygon sides=4] (x2) at (2.4,0) {$x$};
  \node[draw,regular polygon,regular polygon sides=3,shape border rotate=30] (h2) at (4.8,-1) {$y$};
  \node[draw,circle] (g3) at (3.6,0) {$z$};
  \node[draw,circle] (g4) at (3.6,-1) {$z$};
  \draw[->] (x2.east)--(g3.west);
  \draw[->] (h2.north west)--(g3.east);
  \draw[->] (h2.west)--(g4.east);
  \draw[->] (h2.south west)--(g4.south east);
  \draw[color=gray,fill] (x2.south)+(0,0.08) circle (0.02);
  \draw[color=gray,fill] (x2.east)+(-0.08,0.03) circle (0.02);
  \draw[color=gray,fill] (x2.east)+(-0.08,-0.03) circle (0.02);
  \draw[color=gray,fill] (h2.north west)+(0.08,0) circle (0.02);
  \draw[color=gray,fill] (h2.west)+(0.08,0.03) circle (0.02);
  \draw[color=gray,fill] (h2.west)+(0.08,-0.03) circle (0.02);
  \draw[color=gray,fill] (h2.south west)+(0.08,0.05) circle (0.02);
  \draw[color=gray,fill] (h2.south west)+(0.08,0) circle (0.02);
  \draw[color=gray,fill] (h2.south west)+(0.08,-0.05) circle (0.02);
  \draw[color=gray,fill] (g3.west)+(0.08,0) circle (0.02);
  \draw[color=gray,fill] (g3.east)+(-0.08,0.03) circle (0.02);
  \draw[color=gray,fill] (g3.east)+(-0.08,-0.03) circle (0.02);
  \draw[color=gray,fill] (g4.east)+(-0.08,0) circle (0.02);
  \draw[color=gray,fill] (g4.south east)+(-0.02,0.06) circle (0.02);
  \draw[color=gray,fill] (g4.south east)+(-0.06,0.02) circle (0.02);
  \node[draw,circle] (g5) at (1.8,-1.5) {$z$};
  \draw[->] (x1.south)--(g5.north west);
  \draw[->] (x2.south)--(g5.north east);
  \draw[color=gray,fill] (g5.north west)+(0.05,-0.05) circle (0.02);
  \draw[color=gray,fill] (g5.north east)+(-0.02,-0.06) circle (0.02);
  \draw[color=gray,fill] (g5.north east)+(-0.06,-0.02) circle (0.02);
\end{tikzpicture}
  \caption{Diagrammatic representation of the invariant $F'\nobreak\in\nobreak\bbC[V^{\otimes 2} \oplus V^{\otimes 3} \oplus (V^*)^{\otimes 2}]^{\GL(V)}$ corresponding to the invariant~$F$ in the previous diagram. Different numbers of small dots indicate different tensor factors.}
  \label{fig:diagram-2}
\end{figure}
To prove that it also reflects closure equivalence, consider an arbitrary contraction invariant $F$ for $\Symg_n$ on $V^{\otimes 2}$ given by $F(x) = \Tr_\pi(x^{\otimes a} \otimes h^{\otimes b} \otimes g^{\otimes c})$.
Then it is easy to see that there exists a $\GL(V)$-invariant $F'$ on $V^{\otimes 2} \oplus V^{\otimes 3} \oplus (V^*)^{\otimes 2}$ such that $F' \circ \alpha = F$,
namely, $F'(x, y, z) = \Tr_\pi(x^{\otimes a} \otimes y^{\otimes b} \otimes c^{\otimes c})$.
This proves that $\alpha$ reflects closure equivalence.
The map $\alpha$ depends on the dimension $n$.
It is clear that it can be efficiently computed for a given rational tensor $x \in (\bbQ^{n})^{\otimes 2}$.
Collecting all these maps into a sequence, we obtain the required efficient algebraic reduction.

As the next step, we construct a reduction from the problem $\OCI\rep{V^{\otimes 2} \oplus V^{\otimes 3} \oplus (V^*)^{\otimes 2}}{\GL(V)}$ considered in the previous step to $\OCI\rep{V^{\otimes 6} \otimes (V^*)^{\otimes 6} \oplus V^{\otimes 2} \otimes (V^*)^{\otimes 2}}{\GL(V)}$.
To do this, we note that in every contraction invariant $F'(x, y, z) = \Tr_\pi(x^{\otimes a} \otimes y^{\otimes b} \otimes z^{\otimes c})$ on $V^{\otimes 2} \oplus V^{\otimes 3} \oplus (V^*)^{\otimes 2}$ the tensor product $x^{\otimes a} \otimes y^{\otimes b} \otimes z^{\otimes c}$ must have the same number of covariant and contravariant factors to contract.
Counting the number of factors, we obtain the equation $2a + 3b = 2c$.
It follows that $b$ is even and $c$ can be expressed as $c = d + a$ with $2d = 3b$.
In other words, the tensors in the above tensor product can be divided into blocks of two types: $x \otimes z$ and $y^{\otimes 2} \otimes z^{\otimes 3}$.
This suggest the construction based on the map
\begin{align*}
  \alpha'\colon\! V^{\otimes 2} \oplus V^{\otimes 3} \oplus (V^*)^{\otimes 2} & \to V^{\otimes 6} {\otimes} (V^*)^{\otimes 6} \oplus V^{\otimes 2} {\otimes} (V^*)^{\otimes 2}
  \\
  \alpha'(x,y,z) &= (y^{\otimes 2} \otimes z^{\otimes 3}, x \otimes z)
\end{align*}
The map $\alpha'$ preserves closure equivalence because it is equivariant, and reflects closure equivalence because every contraction invariant $F'(x, y, z) = \Tr_\pi(x^{\otimes a} \otimes y^{\otimes b} \otimes z^{\otimes c})$ on $V^{\otimes 2} \oplus V^{\otimes 3} \oplus (V^*)^{\otimes 2}$ can be represented as $F' = F'' \circ \alpha'$ with $F''(s, t) = \Tr_{\pi'} (s^{\otimes \frac{b}{2}} \otimes t^{\otimes a})$.
Again, the maps $\alpha'$ are efficiently computable for spaces of various dimension and thus give an efficient algebraic reduction.
In a similar way we can reduce OCI problems for arbitrary tensor tuple sequences to OCI problems for \emph{balanced} sequences for which each tensor in the tuple has equal number of covariant and contravariant factors.

\begin{figure}[ht]
\centering
\begin{tikzpicture}[>=latex]
  \node[draw,regular polygon,regular polygon sides=4] (x1) at (1.2,0) {$x$};
  \node[draw,regular polygon,regular polygon sides=3,shape border rotate=-30] (h1) at (-1.2,-1) {$y$};
  \node[draw,circle] (g1) at (0,0) {$z$};
  \node[draw,circle] (g2) at (0,-1) {$z$};
  \draw[->] (x1.west)--(g1.east);
  \draw[->] (h1.north east)--(g1.west);
  \draw[->] (h1.east)--(g2.west);
  \draw[->] (h1.south east)--(g2.south west);
  \draw[color=gray,fill] (x1.south)+(0,0.08) circle (0.02);
  \draw[color=gray,fill] (x1.west)+(0.08,0.03) circle (0.02);
  \draw[color=gray,fill] (x1.west)+(0.08,-0.03) circle (0.02);
  \draw[color=gray,fill] (h1.north east)+(-0.08,0) circle (0.02);
  \draw[color=gray,fill] (h1.east)+(-0.08,0.03) circle (0.02);
  \draw[color=gray,fill] (h1.east)+(-0.08,-0.03) circle (0.02);
  \draw[color=gray,fill] (h1.south east)+(-0.08,0.05) circle (0.02);
  \draw[color=gray,fill] (h1.south east)+(-0.08,0) circle (0.02);
  \draw[color=gray,fill] (h1.south east)+(-0.08,-0.05) circle (0.02);
  \draw[color=gray,fill] (g1.east)+(-0.08,0) circle (0.02);
  \draw[color=gray,fill] (g1.west)+(0.08,0.03) circle (0.02);
  \draw[color=gray,fill] (g1.west)+(0.08,-0.03) circle (0.02);
  \draw[color=gray,fill] (g2.west)+(0.08,0) circle (0.02);
  \draw[color=gray,fill] (g2.south west)+(0.02,0.06) circle (0.02);
  \draw[color=gray,fill] (g2.south west)+(0.06,0.02) circle (0.02);
  \node[draw,regular polygon,regular polygon sides=4] (x2) at (2.4,0) {$x$};
  \node[draw,regular polygon,regular polygon sides=3,shape border rotate=30] (h2) at (4.8,-1) {$y$};
  \node[draw,circle] (g3) at (3.6,0) {$z$};
  \node[draw,circle] (g4) at (3.6,-1) {$z$};
  \draw[->] (x2.east)--(g3.west);
  \draw[->] (h2.north west)--(g3.east);
  \draw[->] (h2.west)--(g4.east);
  \draw[->] (h2.south west)--(g4.south east);
  \draw[color=gray,fill] (x2.south)+(0,0.08) circle (0.02);
  \draw[color=gray,fill] (x2.east)+(-0.08,0.03) circle (0.02);
  \draw[color=gray,fill] (x2.east)+(-0.08,-0.03) circle (0.02);
  \draw[color=gray,fill] (h2.north west)+(0.08,0) circle (0.02);
  \draw[color=gray,fill] (h2.west)+(0.08,0.03) circle (0.02);
  \draw[color=gray,fill] (h2.west)+(0.08,-0.03) circle (0.02);
  \draw[color=gray,fill] (h2.south west)+(0.08,0.05) circle (0.02);
  \draw[color=gray,fill] (h2.south west)+(0.08,0) circle (0.02);
  \draw[color=gray,fill] (h2.south west)+(0.08,-0.05) circle (0.02);
  \draw[color=gray,fill] (g3.west)+(0.08,0) circle (0.02);
  \draw[color=gray,fill] (g3.east)+(-0.08,0.03) circle (0.02);
  \draw[color=gray,fill] (g3.east)+(-0.08,-0.03) circle (0.02);
  \draw[color=gray,fill] (g4.east)+(-0.08,0) circle (0.02);
  \draw[color=gray,fill] (g4.south east)+(-0.02,0.06) circle (0.02);
  \draw[color=gray,fill] (g4.south east)+(-0.06,0.02) circle (0.02);
  \node[draw,circle] (g5) at (1.8,-1.5) {$z$};
  \draw[->] (x1.south)--(g5.north west);
  \draw[->] (x2.south)--(g5.north east);
  \draw[color=gray,fill] (g5.north west)+(0.05,-0.05) circle (0.02);
  \draw[color=gray,fill] (g5.north east)+(-0.02,-0.06) circle (0.02);
  \draw[color=gray,fill] (g5.north east)+(-0.06,-0.02) circle (0.02);
  \draw[dashed] (-0.5,0.5)--(1.7,0.5)--(1.7,-0.5)--(-0.5,-0.5)--cycle;
  \draw[dashed] (4.1,0.5)--(1.9,0.5)--(1.9,-0.5)--(4.1,-0.5)--cycle;
  \draw[dashed,dash pattern=on 3 off 1 on 1 off 1] (-2,-0.2)--(-0.7,-0.2)--(-0.7,-0.6)--(4.3,-0.6)--(4.3,-0.2)--(5.6,-0.2)--(5.6,-1.9)--(-2,-1.9)--cycle;
  \draw[white] (1,-3)--(2,-3);
\end{tikzpicture}
  \\
  \vspace{-3em}
\begin{tikzpicture}[>=latex]
  \draw (0,0.5)--(2,0.5)--(2,-0.5)--(0,-0.5)--cycle;
  \node at (1,0) {$t$};
  \draw (3,0.5)--(5,0.5)--(5,-0.5)--(3,-0.5)--cycle;
  \node at (4,0) {$t$};
  \draw (1,-1.5)--(4,-1.5)--(4,-2.5)--(1,-2.5)--cycle;
  \draw[color=gray,fill] (1.25,-0.4) circle (0.02);
  \draw[color=gray,fill] (0.78,-0.4) circle (0.02);
  \draw[color=gray,fill] (0.72,-0.4) circle (0.02);
  \draw[color=gray,fill] (1.25,0.4) circle (0.02);
  \draw[color=gray,fill] (0.78,0.4) circle (0.02);
  \draw[color=gray,fill] (0.72,0.4) circle (0.02);
  \draw[color=gray,fill] (3.75,-0.4) circle (0.02);
  \draw[color=gray,fill] (4.28,-0.4) circle (0.02);
  \draw[color=gray,fill] (4.22,-0.4) circle (0.02);
  \draw[color=gray,fill] (3.75,0.4) circle (0.02);
  \draw[color=gray,fill] (4.28,0.4) circle (0.02);
  \draw[color=gray,fill] (4.22,0.4) circle (0.02);
  \node at (2.5, -2) {$s$};
  \draw[->] (1.25,-0.5) to[out=-90,in=90] (2.25,-1.5);
  \draw[->] (3.75,-0.5) to[out=-90,in=90] (2.75,-1.5);
  \draw[->] (0.75,-0.5) to[out=-150,in=-90,looseness=2] (-0.25,0) to[out=90,in=150,looseness=1.5] (1.25,0.5);
  \draw[->] (4.25,-0.5) to[out=-30,in=-90,looseness=2] (5.25,0) to[out=90,in=30,looseness=1.5] (3.75,0.5);
  \draw[->] (1.25,-2.5) to[out=-150,in=-90,looseness=1] (-0.5,-1) to[out=90,in=135,looseness=2] (0.75,0.5);
  \draw[->] (3.75,-2.5) to[out=-30,in=-90,looseness=1] (5.5,-1) to[out=90,in=45,looseness=2] (4.25,0.5);
  \draw[->] (1.75,-2.5) to[out=-135,in=-90,looseness=1] (0.5,-2) to[out=90,in=150,looseness=1] (1.25,-1.5);
  \draw[->] (2.25,-2.5) to[out=-135,in=-90,looseness=1] (0.25,-2) to[out=90,in=135,looseness=1] (1.75,-1.5);
  \draw[->] (2.75,-2.5) to[out=-45,in=-90,looseness=1] (4.5,-2) to[out=90,in=30,looseness=1] (3.25,-1.5);
  \draw[->] (3.25,-2.5) to[out=-45,in=-90,looseness=1] (4.75,-2) to[out=90,in=45,looseness=1] (3.75,-1.5);
  \draw[color=gray,fill] (1.25,-2.4) circle (0.02);
  \draw[color=gray,fill] (1.72,-2.4) circle (0.02);
  \draw[color=gray,fill] (1.78,-2.4) circle (0.02);
  \draw[color=gray,fill] (2.20,-2.4) circle (0.02);
  \draw[color=gray,fill] (2.25,-2.4) circle (0.02);
  \draw[color=gray,fill] (2.30,-2.4) circle (0.02);
  \draw[color=gray,fill] (2.72,-2.43) circle (0.02);
  \draw[color=gray,fill] (2.78,-2.43) circle (0.02);
  \draw[color=gray,fill] (2.72,-2.37) circle (0.02);
  \draw[color=gray,fill] (2.78,-2.37) circle (0.02);
  \draw[color=gray,fill] (3.20,-2.43) circle (0.02);
  \draw[color=gray,fill] (3.25,-2.43) circle (0.02);
  \draw[color=gray,fill] (3.30,-2.43) circle (0.02);
  \draw[color=gray,fill] (3.22,-2.37) circle (0.02);
  \draw[color=gray,fill] (3.28,-2.37) circle (0.02);
  \draw[color=gray,fill] (3.70,-2.43) circle (0.02);
  \draw[color=gray,fill] (3.75,-2.43) circle (0.02);
  \draw[color=gray,fill] (3.80,-2.43) circle (0.02);
  \draw[color=gray,fill] (3.70,-2.37) circle (0.02);
  \draw[color=gray,fill] (3.75,-2.37) circle (0.02);
  \draw[color=gray,fill] (3.80,-2.37) circle (0.02);
  \draw[color=gray,fill] (1.25,-1.6) circle (0.02);
  \draw[color=gray,fill] (1.72,-1.6) circle (0.02);
  \draw[color=gray,fill] (1.78,-1.6) circle (0.02);
  \draw[color=gray,fill] (2.20,-1.6) circle (0.02);
  \draw[color=gray,fill] (2.25,-1.6) circle (0.02);
  \draw[color=gray,fill] (2.30,-1.6) circle (0.02);
  \draw[color=gray,fill] (2.72,-1.63) circle (0.02);
  \draw[color=gray,fill] (2.78,-1.63) circle (0.02);
  \draw[color=gray,fill] (2.72,-1.57) circle (0.02);
  \draw[color=gray,fill] (2.78,-1.57) circle (0.02);
  \draw[color=gray,fill] (3.20,-1.63) circle (0.02);
  \draw[color=gray,fill] (3.25,-1.63) circle (0.02);
  \draw[color=gray,fill] (3.30,-1.63) circle (0.02);
  \draw[color=gray,fill] (3.22,-1.57) circle (0.02);
  \draw[color=gray,fill] (3.28,-1.57) circle (0.02);
  \draw[color=gray,fill] (3.70,-1.63) circle (0.02);
  \draw[color=gray,fill] (3.75,-1.63) circle (0.02);
  \draw[color=gray,fill] (3.80,-1.63) circle (0.02);
  \draw[color=gray,fill] (3.70,-1.57) circle (0.02);
  \draw[color=gray,fill] (3.75,-1.57) circle (0.02);
  \draw[color=gray,fill] (3.80,-1.57) circle (0.02);
\end{tikzpicture}
  \caption{Construction of the invariant $F''\nobreak\in\nobreak\bbC[V^{\otimes 6} \otimes (V^*)^{\otimes 6} \oplus V^{\otimes 2} \otimes (V^*)^{\otimes 2}]^{\GL(V)}$ corresponding to the invariant~$F'$ in the previous diagram. Again, we use different numbers of small dots to indicate different tensor factors.}
  \label{fig:diagram-3}
\end{figure}

Next, reduce {\small$\OCI\rep{V^{\otimes 6} \otimes (V^*)^{\otimes 6} \oplus V^{\otimes 2} \otimes (V^*)^{\otimes 2}}{\GL(V)}$} to $\OCI\rep{(V^{\otimes 6} \otimes (V^*)^{\otimes 6})^{\oplus 2}}{\GL(V)}$ in order to force all tensors in a tuple to have the same type.
This is a very easy reduction given by the maps $\alpha''(s, t) = (s, t \otimes \Id^{\otimes 4})$ where $\Id \in V \otimes V^*$ is the tensor corresponding to the identity map on $V$.
It is easy to see that $\alpha''$ embeds $V^{\otimes 6} \otimes (V^*)^{\otimes 6} \oplus V^{\otimes 2} \otimes (V^*)^{\otimes 2}$ into $(V^{\otimes 6} \otimes (V^*)^{\otimes 6})^{\oplus 2}$ as a subrepresentation, and therefore obviously preserves and reflects closure equivalence.
This reduction can of course still analyzed from the invariant-theoretic point of view.
To do this, note that $t = \frac{1}{n^4}\Tr^{3,4,5,6}_{3,4,5,6}(t \otimes \Id^{\otimes 4})$, which means that every contraction invariant $F''(s, t) = \Tr_\pi (s^{\otimes a} \otimes t^{\otimes b})$ on $V^{\otimes 6} \otimes (V^*)^{\otimes 6} \oplus V^{\otimes 2} \otimes (V^*)^{\otimes 2}$ can be represented as $F'' = F''' \circ \alpha''$ where
\[
  F'''(s, t) = n^{-4b} \Tr_\pi(s^{\otimes a} \otimes (\Tr^{3,4,5,6}_{3,4,5,6} t)^{\otimes b})
\] is a multiple of a contraction invariant on $(V^{\otimes 6} \otimes (V^*)^{\otimes 6})^{\oplus 2}$.

At the next step, reduce {\small$\OCI\rep{(V^{\otimes 6} \otimes (V^*)^{\otimes 6})^{\oplus 2}}{\GL(V)}$} to  {\small$\OCI\rep{(W \otimes W^*)^{\oplus 2} {\oplus} (V \otimes V^* \otimes W \otimes W^*)^{\oplus 6}}{\GL(V) \otimes \GL(W)}$}.
For a vector space $V$ we define $W = V^{\otimes 6}$.
Every tensor $t \in V^{\otimes 6} \otimes (V^*)^{\otimes 6}$ can be seen as a tensor $\bar{t} \in W \otimes W^*$, since the two spaces are isomorphic, but not every contraction of tensors in $V^{\otimes 6} \otimes (V^*)^{\otimes 6}$ can be expressed as a contraction of tensors in $W \otimes W^*$, since it is no longer possible to contract the six factors in arbitrary order.
To restore the ability to perform arbitrary contractions, we introduce additional helper tensors in $V \otimes V^* \otimes W \otimes W^*$.
Consider six inclusions $\rho_1, \dots, \rho_6$ of $V \otimes V^*$ into $W \otimes W^*$ defined as
\[\rho_i(m) = \Id^{\otimes (i - 1)} \otimes m \otimes \Id^{\otimes (6 - i)}.\]
Let $r(1), \dots, r(6)$ be the corresponding tensors in $V \otimes V^* \otimes W \otimes W^*$.
In the index notation, we can write
\[r(i)^{a \{b_1 \dots b_6\}}_{c \{d_1 \dots d_6\}} = \delta^a_{d_i} \delta^{b_i}_c \prod_{j\neq i} \delta^{b_i}_{d_i}\]
where index sextuplets surrounded by braces $\{b_1 \dots b_6\}$ and $\{d_1 \dots d_6\}$ correspond to one factor of type $W = V^{\otimes 6}$ or $W^* \cong (V^*)^{\otimes 6}$ each.
A straightforward algebraic manipulation shows that a tensor $t \in (V^{\otimes 6}) \otimes (V^*)^{\otimes 6}$ can be recovered from the corresponding $\bar{t} \in W \otimes W^*$ as the contraction
\[
  t = \Tr[W]^{1,2,3,4,5,6,7}_{2,3,4,5,6,7,1} \left(\bar{t} \otimes r(1) \otimes \dots \otimes r(6)\right)
\]
where the symbol $\Tr[W]$ shows that only factors of type $W$ and $W^*$ are contracted.
To give the required reduction, we define a sequence of maps
\[
  \alpha'''(s, t) = (\bar{s}, \bar{t}, r(1), \dots, r(6)).
\]
Using the fact that the tensors $r(i)$ are invariant under the action of $\GL(V)$ on $V$ and $W = V^{\otimes 6}$, we can prove that this map preserves closure equivalence.
And to prove that it reflects closure equivalence, we note that every contraction invariant $F'''(s, t) = \Tr_\pi(s^{\otimes a} \otimes t^{\otimes b})$ on $(V^{\otimes 6} \otimes (V^*)^{\otimes 6})^{\oplus 2}$ can be represented as $F'''' \circ \alpha'''$ where $F''''$ is a contraction invariant on $(W \otimes W^*)^{\oplus 2} \oplus (V \otimes V^* \otimes W \otimes W^*)^{\oplus 6}$, constructed by expressing $s$ and $t$ in terms of $\bar{s}$ and $\bar{t}$ respectively with the help of tensors $r(i)$ as above.

As the final step of the example, we reduce {\small$\OCI\rep{(W \otimes W^*)^{\oplus 2} {\oplus} (V \otimes V^* \otimes W \otimes W^*)^{\oplus 6}}{\GL(V) \otimes \GL(W)}$} to {\small$\OCI\rep{(V \otimes V^* \otimes W \otimes W^*)^{\oplus 8}}{\GL(V) \times \GL(W)}$}.
This step utilizes essentially the same construction as in the above reduction from {\small$\OCI\rep{V^{\otimes 6} \otimes (V^*)^{\otimes 6} \oplus V^{\otimes 2} \otimes (V^*)^{\otimes 2}}{\GL(V)}$} to {\small$\OCI\rep{(V^{\otimes 6} \otimes (V^*)^{\otimes 6})^{\oplus 2}}{\GL(V)}$}.
We again use the identity tensors to construct an injective morphism of representations
\[\alpha''''(s, t, r_1, \dots, r_6) = (\Id_V \otimes\, s, \Id_V \otimes\, t, r_1, \dots, r_6).\]

This sequence of reductions proves that {\small$\OCI\rep{(V \otimes V^* \otimes W \otimes W^*)^{\oplus 8}}{\GL(V) \times \GL(W)}$}, the 2D~PEPS gauge equivalence problem for local dimension (that is, number of summands) equal to~$8$, is $\GIclass$-hard.
Similar constructions can be used to give reductions from an arbitrary OCI problem for tensor tuples to a problem of the same form, with a possibly different local dimension, see~\S\ref{sec:reductions}.
Another reduction allows us to reduce the local dimension to~$3$ and thereby establish~\cref{thm:complete-PEPS}.

\subsection{Outlook and open problems}

It remains an open question to determine the exact complexity of OCI problems.
An interesting question related to this is whether orbit closure intersection problems are in fact easier than orbit equality problems.
On one hand, simple cases suggest that orbit closure intersection problems may have a simpler description.
On the other hand, in many situations a generic tensor tuple in large enough dimension has a closed orbit (that is, is \emph{polystable} in the sense of geometric invariant theory~\cite{Mumford-GIT}), which means that the two problems have the same answer for most inputs.
We can rephrase this question as the question about complexity classes $\OCIclass$ and $\TIclass$.
We note that $\TIclass$ is known to contain not only the graph isomorphism problem, but also the code equivalence problem~\cite{DBLP:journals/dcc/DAlconzo24}, for which no quasipolynomial-time algorithm is known.

\begin{question}
    What is the computational complexity of OCI problems for reductive groups? Can they be put in $\NP \cap \coAM$ like graph isomorphism?
    Does $\OCIclass$ contain the code equivalence problem?
    What is the relationship between our class $\OCIclass$ and the class $\TIclass$ from~\textnormal{\cite{DBLP:journals/siamcomp/GrochowQ23}}?
\end{question}

It seems that a tensor tuple sequence does not need to be too complicated to be $\OCIclass$-complete.
We conjecture that as long as an OCI problem is not trivial or efficiently solvable using quiver representation techniques, it is $\OCIclass$-complete, similarly to the classical dichotomy of tame and wild quivers~\cite{tame-wild}; see also~\cite{Belitskii-Sergeichuk,Futorny-Grochow-Sergeichuk}.
A first step in the investigation of this question may be the analysis of the three remaining unclassified balanced tensor problems discussed above in \cref{thm:complete-balanced-GL}.

\begin{question}
    Are the OCI problems for the actions of~$\GL(V)$ on~$V^{\otimes 2} \otimes (V^*)^{\otimes 2}$,
    on $V^{\otimes 2} {\otimes} (V^*)^{\otimes 2} \oplus V {\otimes} V^*$,
    and on $V^{\otimes 3} \otimes (V^*)^{\otimes 3}$ complete for $\OCIclass$?
    More generally, is any OCI problems for tensor tuple sequences either $\OCIclass$-complete, or efficiently solvable?
\end{question}

While we can use our techniques to work with some subgroups of general linear groups, it appears not straightforward to work with special linear groups or their products, in particular, the action of special linear groups and their products on tensors.
The main reason for this is that while the invariants of the orthogonal, symplectic, and symmetric groups can be described using tensors with bounded number of factors ($2$ in the case of orthogonal and symplectic groups, $2$ and $3$ in the case of symmetric groups), the invariants for $\SL_n$ require tensors with $n$ factors (the determinant).
We note that the relation between representations of general linear and special linear groups is similarly unclear for orbit equality~\cite[Question~9]{DBLP:conf/innovations/ChenGQTZ24}.

\begin{question}
    Are orbit closure intersection problems for special linear groups contained in $\OCIclass$?
    In particular, are the problems {\small$\OCI\rep{V_1 \otimes \dots \otimes V_m}{\SL(V_1) \times \dots \times \SL(V_m)}$} and {\small$\OCI\rep{V^{\otimes p}}{\SL(V)}$} in $\OCIclass$?
\end{question}

Our definition of a tensor tuple sequence can be extended to include symmetric/antisymmetric tensors or more general Schur functors, which specify symmetry of tensors with respect to the permutation of factors.
OCI problems for these more general sequences still lie in the class $\OCIclass$, as every Schur module can be included in an appropriate tensor product, and the orbit closures stay the same.
It is likely that many of these problems are also complete for $\OCIclass$.

\begin{question}
    Are there $\OCIclass$-complete problems involving symmetric tensors?
\end{question}

It is interesting to see if the analytic techniques developed recently for null cone and OCI problems~\cite{DBLP:conf/focs/BurgisserFGOWW19,DBLP:conf/stoc/Allen-ZhuGLOW18} can be applied to solve OCI instances derived from difficult instances of graph isomorphism.
Purely invariant-theoretic techniques are likely to fail here, because polynomial invariants for graph isomorphism include hard to compute invariants such as permanent and ``permanental characteristic polynomial'' $\operatorname{per}(A - \lambda I)$ of the adjacency matrix~\cite{Turner-GI} while the power of analytic techniques is currently less clear.

\begin{question}
    Can reductions to OCI problems be used to construct interesting algorithms for graph isomorphism?
\end{question}

\subsection{Organization of the paper}

In \cref{sec:preliminaries} we provide the necessary background information on algebraic groups, invariant theory, and some properties of rational polyhedral cones.
In \cref{sec:oci} we formally define OCI problems and efficient algebraic reductions between them.
In \cref{sec:tensors} we review algebraic structures related to tensor contractions and introduce tensor tuple sequences and the complexity class~$\OCIclass$ and its variants.
In \cref{sec:reductions} we describe various reductions between OCI problems for tensor tuples and prove our main complexity theoretic results.

\section{Preliminaries}\label{sec:preliminaries}
In this section we review some background concepts.
After fixing our notation in \cref{subsec:notation}, we discuss algebraic groups, their representations, and their invariants in \cref{subsec:algebraic groups}.
We also recall a useful fact from integer linear programming in \cref{subsec:ilp}.

\subsection{Notation and conventions}\label{subsec:notation}
Throughout, we work with computational problems over the field of complex numbers~$\bbC$.
To specify the input to such a problem, it needs to be encoded in a way that is amenable for computation.
To this end, we fix some subfield $\bbF \subset \bbC$ whose elements can be represented by bitstrings and such that all operations can be computed in polynomial time.
Typical choices are~$\bbQ$ or~$\bbQ(i)$.

We use bold font for tuples of data associated with product groups, for example, we write $\group{G}_{\ituple{n}} = \prod_{i = 1}^m \GL_{n_i}$\linebreak where $\ituple{n} = (n_1,
\dots,n_m) \in \bbN^m$.

We denote by $\bbN$ the set of nonnegative integers and by~$[n]$ the $n$-element set $\{1, \dots, n\}$.

All vector spaces are complex vector spaces.
By a \emph{group} we mean a linear algebraic group over $\bbC$, unless stated otherwise.

Unless stated otherwise, $\overline{\mathcal{X}}$ denotes Zariski closure of a set~$\mathcal{X}$.
As mentioned earlier, when $\mathcal{X}$ is an orbit of a complex reductive group, the Zariski closure coincides with the Euclidean one.

\subsection{Algebraic groups, representations, invariants}\label{subsec:algebraic groups}
We work with representations of linear algebraic groups over $\bbC$, except in \S\ref{subsec:u} where we consider the unitary groups $\group{U}_n$
For readers not familiar with algebraic groups, we introduce the relevant concepts below.
For our purposes, an affine variety is a subset of a vector space which is the common zero set of a set of polynomials with complex coefficients,
and a morphism of algebraic varieties is a map between varieties obtained as a restriction of a polynomial map between vector spaces.

\begin{definition}[General linear group]
    For a vector space $V$, the \emph{general linear group} $\GL(V)$ is the group of invertible linear operators on $V$.
    We define~$\GL_n = \GL(\bbC^n)$.

    The general linear group $\GL(V)$ is an \emph{affine algebraic group}: it can be identified with the affine variety
    \[\GL(V) = \{(A, c) \mid c \cdot \det A = 1\} \subset \mathrm{Hom}(V,V) \oplus \bbC.\]
    In this form the multiplication and inversion in $\GL(V)$ are given by polynomial maps:
    $(A_1, c_1) \cdot (A_2, c_2) = (A_1 A_2, c_1 c_2)$ and $(A, c)^{-1} = (c \cdot \operatorname{adj} A, \det A)$.
\end{definition}

\begin{definition}[Linear algebraic group]
    A \emph{linear algebraic group} is a subgroup of some general linear group which is also an closed subvariety.
    A \emph{subgroup} of a linear algebraic group $\group{G}$ is a subgroup which is also a closed subvariety.
    A \emph{homomorphism} of algebraic groups is a group homomorphism which is also a morphism of affine varieties.
\end{definition}

\begin{definition}[Representation]
    A \emph{representation} of a group $\group{G}$ is a vector space~$X$ equipped with a homomorphism $\rho\colon \group{G} \to \GL(X)$.
    We write $\group{g} x$ for $\rho(\group{g})x$ if $\rho$ is clear from the context.
    We will often use the notation $\rep{X}{\group{G}}$ to refer to a representation of~$\group{G}$ on~$X$ (the map~$\rho$ is left implicit in the notation).
    The motivation is that we will consider reductions between representations of different groups, hence it will be useful to make the group explicit in the notation.

    A \emph{subrepresentation} of $X$ is a vector subspace~$Y \subset X$ such that $\rho(\group{g}) Y \subseteq Y$ for every $\group{g} \in \group{G}$.
    A representation is \emph{irreducible} if it has no nontrivial subrepresentations.

    If $X$ and $Y$ are representaions of $\group{G}$, then a linear $\alpha \colon X \to Y$ is \emph{equivariant} if $\alpha(\group{g} x) = \group{g} \alpha(x)$ for all $x \in X$ and $\group{g} \in \group{G}$.
    An equivariant map is also called a \emph{morphism of representations}.
\end{definition}

As usual, if $X$ is a representation of $\group{G}$ via $\rho\colon \group{G} \to \GL(X)$, then the structure of a representation on the dual space $X^*$ is given by $\group{g}y = \rho(\group{g}^{-1})^* y$ for $y \in X^*$, so that $y(x) = \group{g}y(\group{g}x)$.
Additionally, if $X$ and $Y$ are representations of $\group{G}$, then $X \oplus Y$ and $X \otimes Y$ can also be given the structure of a representation by acting on both summands or tensor factors respectively.

\begin{definition}[Orbit, equivalence, closure equivalence]\label{def:orbit and equivalence}
    Let $X$ be a representation of $\group{G}$ and $x \in X$.
    The \emph{orbit} of $x$ is \[\group{G} x = \{ \group{g} x \mid \group{g} \in G\}.\]
    We say that two elements $x, y \in X$ are \emph{equivalent (under the action of $\group{G}$)}
    if there exists $\group{g} \in \group{G}$ such that $y = \group{g} x$ or, equivalently, if the orbits $\group{G} x$ and $\group{G} y$ coincide.
    We say that $x$ and $y$ are \emph{closure equivalent} if they have intersecting orbit closures, that is, if~$\overline{\group{G} x} \cap \overline{\group{G} y} \neq \varnothing$.
    We denote equivalence by $\sim_{\group{G}}$ and closure equivalence by $\approx_{\group{G}}$, omitting the subscript if the group is clear from the context.
\end{definition}

\begin{definition}[Invariant]
    Let $X$ be a representation of $\group{G}$.
    A polynomial $F \in \bbC[X]$ is \emph{invariant} under the action of $\group{G}$ if $F(\group{g} \cdot x) = F(x)$ for all $x \in X$ and $\group{g} \in \group{G}$.
    Invariant polynomials form a subalgebra $\bbC[X]^{\group{G}}$ of $\bbC[X]$ called the \emph{invariant algebra}.
\end{definition}

The invariants in $\bbC[X]^{\group{G}}$ are by definition constant on each orbit, and since they are also continuous with respect to Zariski topology, on each orbit closure.
Therefore, if $x \approx_{\group{G}} y$, then the values of all invariants coincide on $x$ and $y$.
The converse holds for so-called reductive groups, which are an important class of groups with a particular well-behaved representation and invariant theory.
In characteristic zero, they can be defined as follows:

\begin{definition}
    An algebraic group $\group{G}$ is \emph{(linearly) reductive} if every representation of $\group{G}$ can be decomposed into a direct sum of irreducible representations.
\end{definition}

Examples of reductive groups include the general linear groups $\GL(V)$ and products thereof.
If the group $\group{G}$ is reductive, then invariants completely characterize closure equivalence, as follows from results of Mumford.

\begin{theorem}[{\cite[Cor.~2.3.8]{Derksen-Kemper}}]\label{thm:mumford}
    Let $\group{G}$ be a reductive group acting on a vector space $X$, and let $x, y$ be two elements of~$X$.
    Then, $x \approx_{\group{G}} y$ if and only if $F(x) = F(y)$ for all $F \in \bbC[X]^{\group{G}}$.
\end{theorem}

This fact is the foundation of Mumford's geometric invariant theory, which uses the invariant algebra $\bbC[X]^{\group{G}}$ to construct spaces parameterizing orbits or orbit closures.

A famous result of Hilbert (see, e.g., \cite[Thm.~3.5]{Popov-Vinberg}) shows that if $\group{G}$ is reductive, then the invariant algebra $\bbC[X]^{\group{G}}$ is finitely generated.
This means that if $F_1, \dots, F_m$ is a generating set of $\bbC[X]^{\group{G}}$, then it is enough to check if $F_i(x) = F_i(y)$ to conclude that $x \approx_{\group{G}} y$.
Sometimes one can find smaller sets of invariants which, while not generating, still have this property.
To this end we make the following definition of a separating set of invariants.

\begin{definition}\label{def:sep}
    Let $X$ be a representation of reductive group $\group{G}$.
    A set of invariants $\Gamma \subset \bbC[X]^{\group{G}}$ is a \emph{separating set} if for every $x, y\in X$ we have $x \approx_{\group{G}} y$ if and only if $F(x) = F(y)$ for all $F \in \Gamma$.
\end{definition}

For reductive groups this definition is equivalent to \cite[Def.~2.4.1]{Derksen-Kemper}.
For non-reductive groups the definition in~\cite{Derksen-Kemper} use non-distinguishability by polynomial invariants instead of closure equivalence~$\approx_{\group{G}}$.
The relevance of \cref{def:sep} has also been observed in the context of algebraic and geometric complexity theory, see e.g., \cite{DBLP:conf/coco/GargIMOWW20} and references therein.

In some cases we require that linear algebraic groups and homomorphisms between them are defined over the subfield $\bbF \subset \bbC$ in which we perform computations, in the following sense.
\begin{definition}
    Suppose $V$ is a vector space with a fixed basis.
    We say that a closed subvariety of $V$ is \emph{defined over $\bbF$} if it is a common zero set of a set of polynomials with coefficients in $\bbF$.
    A morphism of varieties is \emph{defined over $\bbF$} if it is given by a polynomial map with coefficients in $\bbF$.
    In particular, we talk about linear algebraic groups and their homomorhisms defined over $\bbF$.
    If $X$ is a representation of $\group{G}$ with a fixed basis is \emph{defined over $\bbF$} if its defining homomorphism $\rho \colon \group{G} \to \GL(X)$ is defined over $\bbF$.
\end{definition}
We refer to~\cite{Borel-LAG} for a comprehensive treatment of groups defined over a subfield of an algebraically closed field.
An important fact about reductive linear algebraic groups defined over $\bbF$ is that the invariant algebra $\bbC[X]^{\group{G}}$ is generated by invariants with coefficients in $\bbF$~\cite[Theorem 1.1]{Mumford-GIT}.

\subsection{Integer linear programming}\label{subsec:ilp}

We recall a simple fact from the theory of integer linear programming  which can be found, for example, in~\cite[\S16]{Schrijver-Integer-Programming}. It states that the set of nonnegative integral points in a rational polyhedral cone is a finitely generated monoid.

\begin{definition}
    A \emph{rational polyhedral cone} is a set of the form
    $C = \{x \in \bbR^n \mid Ax \geq 0\}$
    for some rational matrix~$A \in \bbQ^{m \times n}$.
    An \emph{integral Hilbert basis} is a set of vectors~$a_1, \dots, a_q \in C \cap \bbZ^n$ such that every~$v \in C \cap \bbZ^n$ can be written as a nonnegative integer linear combination: $v = \sum_{k = 1}^q m_k a_k$ with $m_k \in \bbN$.
\end{definition}

\begin{theorem}[{\cite[Thm.\,16.4]{Schrijver-Integer-Programming}}]\label{thm:hilbert}
    Every rational polyhedral cone has an integral Hilbert basis.
\end{theorem}

\section{Orbit closure intersection problems and reductions}\label{sec:oci}
In this section we define the orbit closure intersection problem formally as a computational decision problem (\cref{subsec:orbit-closure-intersection-problems}).
We then discuss an appropriate notion of reductions (\cref{subsec:reductions}) and give criteria for verifying when maps gives rise to a reduction (\cref{subsec:criterion}).

\subsection{Sequences of representations and orbit closure intersection problems}\label{subsec:orbit-closure-intersection-problems}
We now formalize the orbit closure intersection as a computational problem.
More specifically, we will describe a template orbit closure intersection problem which can be adapted to any situation where we have a family of groups depending on a ``size parameter'' such as $\GL_n$ and a representation for each of these groups.

\begin{definition}[$p$-sequence]\label{def:p seq}
    A \emph{polynomially bounded sequence of representations}, or a \emph{$p$-sequence of representations}, is given by a sequence of groups $\group{G}_n$, a sequence of representations~$\rho_n \colon \group{G}_n \to \GL(X_n)$ on vector spaces~$X_n$ such that the dimension $\dim X_n$ is polynomially bounded as a function of~$n\in\bbN$, and a choice of basis in every vector space $X_n$.
    As with representations, we often write $\rep{X_n}{\group{G}_n}$ for such a $p$-sequence of representations, leaving the homomorphisms $\rho_n$ and bases implicit.

    A \emph{multivariate $p$-sequence of representations} is defined analogously but with~$\ituple{n} = (n_1, \dots, n_m) \in \bbN^m$ a tuple of parameters.
\end{definition}

This is a very general definition.
Tensor actions provide a rich source of examples of $p$-sequences, as we will see in \cref{subsec:tuples} below.
For now, the reader can imagine any of the concrete examples discussed in the introduction, such as the simultaneous conjugation action of $\group{G}_n = \GL_n$ on the space~$X_n = (\bbC^{n \times n})^{\oplus p}$ of $p$-tuples of $n\times n$ matrices.

We are now ready to formally define the notion of an orbit closure intersection problem.

\begin{definition}[Orbit closure intersection problem]
    Let~$\rep{X_n}{\group{G}_n}$ be a $p$-sequence of representations.
    Denote by~$X_n(\bbF)$ the set of elements in $X_n$ coordinates of which (in the chosen basis) lie in the subfield $\bbF$.
    The \emph{orbit closure intersection problem} $\OCI\rep{X_n}{\group{G}_n}$ is the following decision problem:

    \begin{quote}
      {\emph{Given $n$ (encoded in unary) and two elements\linebreak $x, y \in X_n(\bbF)$ (encoded as tuples of elements of~$\bbF$, which are the coordinates in the chosen basis), decide if~$x \approx_{\group{G}_n} y$.}}
    \end{quote}
\end{definition}

The requirement that~$n$ is encoded in unary is harmless when~$\dim X_n$ grows at least linearly with~$n$, as the encodings of $x$ and $y$ in this case will have length at least $n$.
It is only needed to ensure that the input size is at least linear in~$n$ even when the dimension $\dim X_n$ is sublinear (in particular, when it is constant).

\subsection{Reductions between representations and orbit closure intersection problems}\label{subsec:reductions}
Since we consider algebraic groups, it is natural to consider reductions that are given by polynomial functions.
This leads to the following definition.

\begin{definition}[Efficient algebraic reduction]\label{def:algebraic reduction}
    Let $\rep{X_n}{\group{G}_n}$ and $\rep{Y_n}{\group{H}_n}$ be two $p$-sequences of representations.
    An \emph{algebraic reduction} from $\rep{X_n}{\group{G}_n}$ to $\rep{Y_n}{\group{H}_n}$ is a sequence of polynomial maps $\alpha_n \colon X_n \to Y_{q(n)}$ for some function $q\colon \bbN \to \bbN$ such that the following holds for all $n\in \bbN$, $x,y \in X_n$:
    \begin{align*}
        x \approx_{\group{G}_n} y
        \quad\Leftrightarrow\quad
        \alpha_n(x) \approx_{\group{H}_{q(n)}} \alpha_n(y).
    \end{align*}
    Such a reduction is called \emph{efficient} if~$q$ is polynomially bounded, the maps~$\alpha_n$ are defined over $\bbF$, and the functions~$q$ and~$(n, x) \mapsto \alpha_n(x)$ for $x \in X_n(\bbF)$ are efficiently computable (with $n$ given in unary and elements of $X_n$ and $Y_{q(n)}$ encoded as tuples of coordinates in the chosen bases).
    Analogous definitions are made for multivariate $p$-sequences.
\end{definition}

We use the standard notion of efficiently computable function, which can be defined for example as a function implementable by a polynomial-time Turing machine.
Our results do not depend much on the notion of efficient computability used in this definition, and can be made to work in a more algebraic setting where $\alpha_n$ are implemented by polynomial-size circuits over a field $\bbF$, which in this case may even be uncountable.
In fact, in all reductions we construct in this paper the maps~$\alpha_n$ always can be implemented by tuples of monomials.

Efficient algebraic reductions between $p$-sequences of representations give rise to reductions between corresponding orbit closure intersection problems in the usual sense (polynomial-time Karp reductions).
These are special cases of \emph{kernel reductions} introduced by Fortnow and Grochow~\cite{DBLP:journals/iandc/FortnowG11} for general equivalence problems.

\begin{proposition}
    If there exist an efficient algebraic reduction from $\rep{X_n}{\group{G}_n}$ to $\rep{Y_n}{\group{H}_n}$,
    then there exists a polynomial-time Karp reduction from $\OCI\rep{X_n}{\group{G}_n}$ to $\OCI\rep{Y_n}{\group{H}_n}$.
\end{proposition}
\begin{proof}
    The required Karp reduction maps an instance $(n, x, y)$ of $\OCI\rep{X_n}{\group{G}_n}$ to the instance $(q(n), \alpha_n(x), \alpha_n(y))$ of $\OCI\rep{Y_n}{\group{H}_n}$.
\end{proof}

\begin{corollary}\label{cor:trivial}
    If there exist an efficient algebraic reduction from $\rep{X_n}{\group{G}_n}$ to the trivial $p$-sequence $\rep{\bbC^n}{\Eg}$, where~$\Eg$ is the trivial group with one element, then $\OCI\rep{X_n}{\group{G}_n} \in \class{P}$.
\end{corollary}
\begin{proof}
    The orbits and orbit closures for the action of the trivial group are singleton sets.
    Thus, the orbit closure intersection $\OCI\rep{\bbC^n}{\Eg}$ is just equality testing, which is in~$\class{P}$.
\end{proof}

Note that an algebraic reduction as in \cref{cor:trivial} consists of polynomial maps~$\alpha_n \colon X_n \to \bbC^{q(n)}$ such that $x \approx_{\group{G}} y$ if and only if $\alpha_n(x) = \alpha_n(y)$, which means that that the coordinates of~$\alpha_n$ form a separating set of invariants.
Thus the corollary expresses in the context of orbit closure intersection problems a well-known principle that in order to efficiently solve an equivalence problem, it is enough to find an efficiently computable system of separating invariants.

\subsection{Criteria for reductions}\label{subsec:criterion}
While the notion of an algebraic reduction is natural, it is not obvious how to verify this property.
In this section we give a number of useful characterizations.
We first introduce the following notions:

\begin{definition}
Consider groups $\group{G}$ and $\group{H}$ acting on vector spaces $X$ and $Y$ respectively.
We say that a polynomial map $\alpha \colon X \to Y$ \emph{preserves closure equivalence} if \[x \approx_{\group{G}} y \Rightarrow \alpha(x) \approx_{\group{H}} \alpha(y)\] holds for every $x, y \in X$.
Conversely, we say that $\alpha$ \emph{reflects closure equivalence} (or \emph{preserves closure inequivalence}) if the opposite implication \[\alpha(x) \approx_{\group{H}} \alpha(y)  \Rightarrow x \approx_{\group{G}} y\] holds for every $x, y \in X$.
\end{definition}

Thus, an algebraic reduction consists of maps~$\alpha_n$ that preserve \emph{and} reflect closure equivalence.
It is often not hard to construct maps that preserve closure equivalence.
The following lemma is applicable in many situations:

\begin{lemma}\label{lem:beta crit}
    Let $X$ and $Y$ be representations of groups $\group{G}$ and $\group{H}$ respectively, and let $\alpha \colon X \to Y$ be a polynomial map.
    Suppose that there exists a group homomorphism~$\beta\colon \group{G} \to \group{H}$ such that~$\alpha(\group{g} x) = \beta(\group{g}) \alpha(x)$ for all $\group{g} \in \group{G}$, $x\in X$.
    Then $\alpha$ preserves closure equivalence.
\end{lemma}
\begin{proof}
The assumption implies that if $x, y\in X$ are in the same $\group{G}$-orbit then $\alpha(x),\alpha(y)$ are in the same~$\group{H}$-orbit.
Since $\alpha$ is a continuous map, it is continuous, so if $z \in \overline{\group{G} x}$ then
\begin{align*}
    \alpha(z) \subseteq \alpha(\overline{\group{G} x}) \subseteq \overline{\alpha(\group{G} x)} = \overline{\beta(G) \alpha(x)} \subseteq \overline{\group{H} \alpha(x)}.
\end{align*}
Therefore, if $x \approx_{\group{G}} y$ then there exists $z \in \overline{\group{G} x} \cap \overline{\group{G} y}$, so
\begin{align*}
    \alpha(z) \in \alpha(\overline{\group{G} x} \cap \overline{\group{G} y}) \subseteq \overline{\group{H} \alpha(x)} \cap \overline{\group{H} \alpha(y)},
\end{align*}
and hence $\alpha(x) \approx_{\group{H}} \alpha(y)$.
Thus, $\alpha$ preserves closure equivalence.
\end{proof}

\begin{corollary}\label{ex:injective-morphism}
Let $X$ and $Y$ be representations of the same group $\group{G}$.
If $\alpha \colon X \to Y$ is an injective $G$-equivariant linear map, then $\alpha$ preserves and reflects closure equivalence.
\end{corollary}
\begin{proof}
We may assume without loss of generality that~$\alpha$ is an isomorphism (otherwise replace~$Y$ by the image~$\alpha(X)$).
Then we can apply \cref{lem:beta crit} with $\beta=\id_{\group{G}}$ to see that~$\alpha$ and~$\alpha^{-1}$ preserve closure equivalence.
But $\alpha^{-1}$ preserves closure equivalence if and only if $\alpha$ reflects closure equivalence.
\end{proof}

We also have the following useful criterion:

\begin{lemma}\label{lem:closure equivalence vs invariants}
    Let $X$ and $Y$ be representations of groups $\group{G}$ and $\group{H}$ respectively, with $\group{H}$ reductive.
    A polynomial map $\alpha \colon X \to Y$ preserves closure equivalence if, and only if, for every invariant $F \in \bbC[Y]^{\group{H}}$ the pullback $F \circ \alpha$ is an invariant in $\bbC[X]^{\group{G}}$.
\end{lemma}
\begin{proof}
Suppose that $\alpha$ preserves closure equivalence and let $F \in \bbC[Y]^{\group{H}}$.
For every $x \in X$, and~$\group{g} \in \group{G}$, it holds that $\alpha(x) \approx_{\group{H}} \alpha(\group{g} x)$, and hence $F(\alpha(x)) = F(\alpha(\group{g} x))$ because invariant polynomials are constant on orbit closures.
This means that $F \circ \alpha \in \bbC[X]^{\group{G}}$.

Conversely, suppose that for every invariant~$F \in \bbC[Y]^{\group{H}}$ the pullback $F \circ \alpha$ is in $\bbC[X]^{\group{G}}$.
If $x \approx_{\group{G}} y$, then we have $F(\alpha(x)) = F(\alpha(y))$ for every~$F \in \bbC[Y]^{\group{H}}$.
By \cref{thm:mumford}, since $\group{H}$ is reductive, it follows that $\alpha(x) \approx_{\group{H}} \alpha(y)$, and hence $\alpha$ preserves closure equivalence.
\end{proof}

It is often substantially more difficult to show that a map also \emph{reflects} closure equivalence.
To this end we provide the following criterion:

\begin{lemma}\label{thm:reflect}
    Let $X$ and $Y$ be representations of groups $\group{G}$ and $\group{H}$ respectively, with $\group{H}$ reductive.
    Suppose a polynomial map $\alpha \colon X \to Y$ preserves closure equivalence.
    The following~statements~are~equivalent:
    \begin{itemize}
        \item[(a)] $\alpha$ reflects closure equivalence;
        \item[(b)] for every separating set $\Gamma \subseteq \bbC[Y]^{\group{H}}$, the set \[\Gamma \circ \alpha = \{ F \circ \alpha \mid F \in \Gamma \} \subseteq \bbC[X]^{\group{G}}\] is separating;
        \item[(c)] the set $\{ F \circ \alpha \mid F \in \bbC[Y]^{\group{H}} \} \subseteq \bbC[X]^{\group{G}}$ is separating;
        \item[(d)] for some separating set $\Gamma \subseteq \bbC[Y]^{\group{H}}$, the set $\Gamma \circ \alpha$ is separating.
    \end{itemize}
\end{lemma}
\begin{proof}
\emph{(a) $\Rightarrow$ (b):}
Suppose $\alpha$ reflects closure equivalence and let~$\Gamma \subseteq \bbC[Y]^{\group{H}}$ be a separating set.
Then the following holds for all~$x,y \in X$:
\begin{align*}
x \approx_{\group{G}} y
& \Leftrightarrow \alpha(x) \approx_{\group{H}} \alpha(y)\\
& \Leftrightarrow \forall F \in \Gamma \colon F(\alpha(x)) = F(\alpha(y))\\
& \Leftrightarrow \forall G \in \Gamma \circ \alpha \colon G(x) = G(y),
\end{align*}
where we first used that~$\alpha$ preserves and reflects closure equivalence, then that $\Gamma$ is separating, and finally the definition of $\Gamma \circ \alpha$.
Thus we have proved that $\Gamma \circ \alpha$ is separating.

\emph{(b) $\Rightarrow$ (c) $\Rightarrow$ (d):} This is clear since $\Gamma = \bbC[Y]^{\group{H}}$ is separating set as $\group{H}$ is reductive (\cref{thm:mumford}).

\emph{(d) $\Rightarrow$ (a):}
Suppose $\Gamma \subseteq \bbC[Y]^{\group{H}}$ is such that $\Gamma \circ \alpha$ is separating.
Then the following holds for all~$x,y \in X$:
\begin{align*}
\alpha(x) \approx_{\group{H}} \alpha(y)
& \Rightarrow \forall F \in \Gamma \colon F(\alpha(x)) = F(\alpha(y))\\
& \Rightarrow \forall G \in \Gamma \circ \alpha \colon G(x) = G(y)\\
& \Rightarrow x \approx_{\group{G}} y,
\end{align*}
where we first used that $\Gamma$ consists of invariants, then the definition of $\Gamma \circ \alpha$, and finally that~$\Gamma \circ \alpha$ is separating.
Thus we have proved that $\alpha$ reflects closure equivalence.
\end{proof}

We note that in all our reductions below the sets $\Gamma$ and $\Gamma \circ \alpha$ will in fact generate the respective invariant algebras, but to state the general criterion, we need to use separating sets.

\section{Tensor tuple representations and invariants}\label{sec:tensors}
From now on we will focus our attention on the broad class of orbit closure intersection problems associated with tensors and tensor tuples.
In \cref{subsec:tensors} we discuss tensor spaces; these are naturally captured by the tensor algebra and we discuss its mathematical structure.
We first consider a single vector space and then also multiple vector spaces (and their dual).
By fixing the structure of the tensors while varying the dimension of the vector spaces, we obtain $p$-sequences of representations.
We explain this in \cref{subsec:tuples} and define an associated complexity class~$\OCIclass$ that captures the associated orbit closure intersection problems.
We end the section by describing the associated invariant theory in \cref{subsec:tensor invariants}.

\subsection{Tensor algebras and tensor contractions}\label{subsec:tensors}
We recall basic definitions related to mixed tensor spaces, mixed tensor algebras and tensor contractions following~\cite[\S III.3]{Greub-multilinear} and~\cite{Sch08-first-fundamental-MR2379100}.
Since these are the main objects that we work with, we usually omit the ``mixed'' qualification.

\begin{definition}
    Let $V$ be a vector space. The space \[ \tspace{V}{a}{b} = V^{\otimes a} \otimes (V^*)^{\otimes b} \] is called the \emph{(mixed) tensor space of type $\tworow{a}{b}$ over $V$}.
\end{definition}

Every tensor product with $V$ and $V^*$ as factors is isomorphic to $\tspace{V}{a}{b}$ for some $\tworow{a}{b} \in \bbN \times \bbN$ by rearranging the tensor factors, so we can restrict our attention to these spaces when studying tensors formed from elements $V$ and $V^*$.
In particular, the tensor product $\tspace{V}{a_1}{b_1} \otimes \tspace{V}{a_2}{b_2}$ is isomorphic to~$\tspace{V}{a_1+a_2}{b_1+b_2}$.
This allows us to view the tensor product $t_1 \otimes t_2$ with $t_1 \in \tspace{V}{a_1}{b_1}$ and $t_2 \in \tspace{V}{a_2}{b_2}$ as an element of $\tspace{V}{a_1+a_2}{b_1+b_2}$.
More specifically, we define the tensor product on elementary tensors as
\begin{multline*}
    (\bigotimes_{i = 1}^{a_1} v_i \otimes \bigotimes_{j = 1}^{b_1} w_j) \otimes (\bigotimes_{i = 1}^{a_2} x_i \otimes \bigotimes_{j = 1}^{b_2} y_j) \\ = (\bigotimes_{i = 1}^{a_1} v_i \otimes \bigotimes_{i = 1}^{a_2} x_i) \otimes (\bigotimes_{j = 1}^{b_1} w_i \otimes \bigotimes_{j = 1}^{b_2} y_j)
\end{multline*}
where $v_i, x_i \in V$ and $w_j, y_j \in V^*$, and extend it bilinearly to arbitrary tensors.
In what follows, the symbol $\otimes$ applied to two tensors always denotes this tensor product (or its extension to tensors over multiple vector spaces introduced in~\cref{subsec:tuples}).

The tensor product is associative, and we can combine the tensor spaces $\tspace{V}{a}{b}$ into a graded associative algebra.

\begin{definition}
    Let $V$ be a vector space.
    The \emph{(mixed) tensor algebra} over $V$ is the $\bbN \times \bbN$-graded algebra \[\mathcal{T}(V) = \bigoplus\limits_{\tworow{a}{b} \in \bbN \times \bbN} \tspace{V}{a}{b}\] with the tensor product operation $\otimes$ defined above.
\end{definition}

The standard action of $\GL(V)$ on $V$ induces the structure of $\GL(V)$-representation on each tensor space $\tspace{V}{a}{b}$ and therefore on the mixed tensor algebra $\mathcal{T}(V)$.
Explicitly, the action of $\group{g} \in \GL(V)$ is given on elementary tensors by
\begin{align}\label{eq:gl action}
    \group{g} \cdot (\bigotimes_{i = 1}^a v_i \otimes \bigotimes_{j = 1}^b w_j) = \bigotimes_{i = 1}^a (\group{g} v_i) \otimes \bigotimes_{j = 1}^b (\group{g}^*)^{-1} w_j.
\end{align}
In addition, each tensor space $\tspace{V}{a}{b}$ is also a representation of $\Symg_a \times \Symg_b$, acting by separately permuting the contravariant and covariant indices.
The action of $(\pi,\sigma) \in \Symg_a \times \Symg_b$ is given by
\begin{align}\label{eq:permutation action}
    (\pi, \sigma) \cdot (\bigotimes_{i = 1}^a v_i \otimes \bigotimes_{j = 1}^b w_j) = \bigotimes_{i = 1}^a v_{\pi^{-1}i} \otimes \bigotimes_{j = 1}^b w_{\sigma^{-1}j}.
\end{align}

The tensor algebra $\mathcal{T}(V)$ carries additional structure which comes from identity tensors and tensor contractions.
\begin{definition}
    Let $V$ be a vector space.
    The \emph{identity tensor} $\Id \in \tspace{V}{1}{1} = V \otimes V^*$ is the tensor corresponding to the identity map $V \to V$.
\end{definition}

\begin{definition}
    Let $V$ be a vector space.
    For every $\tworow{a}{b} \in \bbN \times \bbN$ and $p \in [a]$, $q \in [b]$ the \emph{contraction map} $\Tr^{p:a}_{q:b} \colon \tspace{V}{a}{b} \to \tspace{V}{a-1}{b-1}$ \emph{contracting the $p$-th contravariant factor with the $q$-th covariant factor} is the linear map defined on elementary tensors as
    \[
        \Tr^{p:a}_{q:b} \left(\bigotimes_{i = 1}^a v_i \otimes \bigotimes_{j = 1}^b w_j\right) = w_q(v_p) \cdot \bigotimes_{i = 1, i \neq p}^a v_i \otimes \bigotimes_{j = 1, j \neq q}^b w_j.
    \]
    We write $\Tr^p_q$ if the tensor type $\tworow{a}{b}$ is clear from the context.
\end{definition}

Clearly, the identity tensor is $\GL(V)$-invariant and tensor products and contractions are $\GL(V)$-equivariant.
We list some simple properties of tensor products and contractions which follow directly from the definitions.

\begin{proposition}\label{prop:contraction-simple}
  Let $t \in \tspace{V}{a}{b}$, $s \in \tspace{V}{c}{d}$, and let $p \in [a]$, $q \in [b]$.
  The following properties hold:
  \begin{enumerate}
    \item\label{prop:contraction-simple-1}
    $s \otimes t = (\pi, \sigma) (t \otimes s)$ where $\pi$ swaps the first $a$ and the last $c$ elements of $[a + c]$, and similarly $\sigma$ swaps the first $b$ and the last $d$ elements of $[b + d]$.
    \item\label{prop:contraction-simple-2} Commutation relations for tensor products and a contraction: \begin{align*}(\Tr^{p:a}_{q:b} t) \otimes s & = \Tr^{p:a + c}_{q:b + d} (t \otimes s) \\ s \otimes (\Tr^{p:a}_{q:b} t) & = \Tr^{c + p:a + c}_{d + q:a + d} (s \otimes t).\end{align*}
    \item\label{prop:contraction-simple-3} Properties of contraction with identity tensor: \begin{align*}\Tr^{p+1:a + 1}_{1:b+1} (\Id \otimes t) & = ((12\dots p),\id) \cdot t \\ \Tr^{1:a+1}_{q+1:b+1} (\Id \otimes t) & = (\id, (12\dots q)) \cdot t.\end{align*}
  \end{enumerate}
\end{proposition}

Interestingly, the last property means that the permutation action~\cref{eq:permutation action} can be recovered by using contractions and tensoring with identity tensors, because the cyclic permutations $(12\dots p)$ generate the whole symmetric group.

Contractions can be composed to create more complicated tensor contractions.
After applying one contraction map $\Tr^{p:a}_{q:b}$, all contravariant indices higher than~$p$ and all covariant indices higher than~$q$ decrease by one.
This leads to the following commutation relation for contraction maps:
\[
    \Tr^{r:a-1}_{s:b-1} \circ \Tr^{p:a}_{q:b} = \begin{cases}
    \Tr^{p-1:a-1}_{q-1:b-1} \circ \Tr^{r:a}_{s:b}, &\text{if $r < p, s < q$} \\
    \Tr^{p-1:a-1}_{q\phantom{-1}:b-1} \circ \Tr^{r\phantom{+1}:a}_{s+1:b}, &\text{if $r < p, s \geq q$} \\
    \Tr^{p\phantom{-1}:a-1}_{q-1:b-1} \circ \Tr^{r+1:a}_{s\phantom{+1}:b}, &\text{if $r \geq p, s < q$} \\
    \Tr^{p:a-1}_{q:b-1} \circ \Tr^{r+1:a}_{s+1:b}, &\text{if $r \geq p, s \geq q$}
    \end{cases}
\]
To avoid the necessity of iteratively tracking these changes of indices, we introduce a convenient notation for contracting multiple pairs of indices at once.

\begin{definition}
A \emph{partial bijection} is a bijection $\gamma \colon P \to Q$ between subsets~$P \subseteq [a]$ and~$Q \subseteq [b]$.
Given such a partial bijection, we define
$\Tr_{\gamma} \colon \tspace{V}{a}{b} \to \tspace{V}{a-k}{b-k}$, where $k = \abs P = \abs Q$, as the contraction of all contravariant factors with indices in $P$ with the covariant factors with indices in~$Q$ according to the bijection $\gamma$, that is,
\[
    \Tr_{\gamma} (\bigotimes_{i = 1}^a v_i \otimes \bigotimes_{j = 1}^b w_j)
  = \prod_{p \in P} w_{\gamma(p)}(v_p) \cdot \left(\bigotimes_{i \notin P} v_i \otimes \bigotimes_{j \notin Q} w_j\right)
\]
We also write $\Tr^{p_1,\dots,p_k:a}_{q_1,\dots,q_k:b} := \Tr_\gamma$ when $P = \{p_1,\dots,p_k\}$, $Q = \{q_1,\dots,q_k\}$, $\gamma(p_j) = q_j$ for all~$j\in[k]$.
\end{definition}

Compositions of contraction maps are exactly $\Tr_{\gamma}$ for some partial bijection $\gamma$:
\begin{itemize}
  \item The identity function on $\tspace{V}{a}{b}$ is $\Tr_\varnothing$ where $\varnothing$ is the nowhere-defined partial function.
  \item If $\gamma$ is a partial bijection with domain $P$ and image $Q$ of cardinality $k$,
    then $\Tr^{p:a-k}_{q:a-k} \circ T_{\gamma} = T_{\bar{\gamma}}$ for the partial bijection $\bar\gamma$ extending $\gamma$ with $\bar\gamma(\bar{p}) = \bar{q}$, where $\bar{p}$ is the $p$-th element of $[a] \setminus P$ and $\bar{q}$ is the $q$-th element of $[b] \setminus Q$.
\end{itemize}
In particular, if $a=b$ and $\pi \in \Symg_a$ is a permutation of $[a]$, then $\Tr_{\pi} \colon \tspace{V}{a}{a} \to \tspace{V}{0}{0} = \bbC$ maps tensors to scalars as follows.
\[
    \Tr_{\pi} (\bigotimes_{i = 1}^a v_i \otimes \bigotimes_{j = 1}^a w_j) = \prod_{p = 1}^a w_{\pi(p)}(v_p).
\]
In this case, $\Tr_{\pi}$ can also be expressed using the action~\eqref{eq:permutation action} of the permutation group as
\[
    \Tr_{\pi}(x) = \Tr \left((\pi, \id) \cdot x\right) = \Tr \left((\id, \pi^{-1}) \cdot x\right)
\]
in terms of the canonical $\Tr = \Tr_{\id}$ that contracts the covariant and contravariant in the same order:
\[
    \Tr (\bigotimes_{i = 1}^a v_i \otimes \bigotimes_{j = 1}^a w_j) = \prod_{p = 1}^a w_p(v_p)
\]

Several other constructions related to tensors can be expressed using contractions.
For example, the tensor space $\tspace{V}{a}{b}$ can be identified with the space of linear maps from $V^{\otimes b}$ to $V^{\otimes a}$.
The composition of linear maps induces a bilinear map $\tspace{V}{a}{b} \times \tspace{V}{b}{c} \to \tspace{V}{a}{c}$, which we denote by
\begin{align}\label{eq:tensor composition}
t \cdot s = \Tr^{a+1,\dots,a+b:a+b}_{\phantom{a+}1,\dots,\phantom{a+}b:b+c} (t \otimes s)
\end{align}
for $t \in \tspace{V}{a}{b}$, $s \in \tspace{V}{b}{c}$.
When $c=a$ we can take the trace of the above to obtain the canonical bilinear pairing between the spaces $\tspace{V}{a}{b}$ and $\tspace{V}{b}{a}$:
\[\left<t, s\right> = \Tr(t \cdot s) = \Tr_\sigma (t \otimes s) \]
where $\sigma$ is the permutation swapping the first $a$ and the last $b$ indices, that is, $\sigma(i) = b + i$ for $i \leq a$, and $\sigma(i) = i - a$ for $i \geq a$.
Finally, we can also express the permutation action using tensor contractions, as follows.
For every permutation $\pi \in \Symg_a$, consider the linear map $P_{\pi} \colon V^{\otimes a} \to V^{\otimes a}$ defined on elementary tensors as
\[
P_\pi( \bigotimes_{i = 1}^a v_i ) = \bigotimes_{i = 1}^a v_{\pi^{-1}(i)}.
\]
Identifying $P_\pi$ with a tensor in~$\tspace{V}{a}{a}$, it can be written as
\[
    P_{\pi} = \Tr^{a+1,a+2,\dots,\phantom{\pi(}2a:2a}_{\pi(1),\pi(2),\dots,\pi(a):2a}\, ( \Id^{\otimes (2a)} ),
\]
Then the action~\eqref{eq:permutation action} of $\Symg_a \times \Symg_b$-action on $\tspace{V}{a}{b}$ can be expressed in terms of two compositions~\eqref{eq:tensor composition}:
\[
(\pi, \sigma) \cdot x = P_{\pi} \cdot x \cdot P_{\sigma^{-1}}
\quad\text{for $\pi \in \Symg_a$, $\sigma \in \Symg_b$.}
\]

The algebraic structure of $\mathcal{T}(V)$, with tensor products, contractions, and identity tensors, is called a \emph{wheeled PROP}~\cite{DM23-wheeled-props-MR4568127} (\emph{PROP} is short for \emph{pro}duct and \emph{p}ermutation category).
We will not require the general definition of a wheeled PROP, but only define what it means for a subalgebra of~$\mathcal{T}(V)$:

\begin{definition}
    A graded subalgebra
    \[ A = \bigoplus\limits_{\tworow{a}{b} \in \bbN \times \bbN} A^a_b \]
    of $\mathcal{T}(V)$ is called a \emph{sub-wheeled PROP} if $\Id \in A^1_1$ and $\Tr^{p:a}_{q:b}(A^a_b) \subseteq A^{a - 1}_{b - 1}$ for every contraction.

    Given a set of tensors $\mathcal{H} \subset \mathcal{T}(V)$, the \emph{sub-wheeled PROP generated by $\mathcal{H}$} is the minimal (with respect to inclusion) sub-wheeled PROP containing $\mathcal{H}$.
    We denote it by $\left<\mathcal{H}\right>_{\Tr}$.
    A sub-wheeled PROP is called \emph{finitely generated} if it is generated by some finite set~$\mathcal{H}$.
\end{definition}

As mentioned, if $X$ is a representation of a reductive group $\group{G}$ then the invariant algebra~$\bbC[X]^{\group{G}}$ is finitely generated.
There is a variant of finite generation theorem for mixed tensor algebras, as was proved by Schrijver~\cite{Sch08-first-fundamental-MR2379100} over~$\bbC$ and by Derksen and Makam~\cite{DM23-wheeled-props-MR4568127} over an arbitrary field of characteristic zero.
For the following result, recall that we have an action~\eqref{eq:gl action} of $\GL(V)$ on the tensor algebra~$\mathcal{T}(V)$.

\begin{theorem}[{\cite[Thm. 6.2, Cor. 6.6]{DM23-wheeled-props-MR4568127}}]\label{thm:dm prop}
    Let $V$ be a vector space and $\group{G}$ be a reductive subgroup of $\GL(V)$.
    Then the invariant tensors
    $\mathcal{T}(V)^{\group{G}} = \{ t \in \mathcal{T}(V) \mid \group{g} \cdot t = t \ \forall \group{g} \in \group{G} \}$
    form a finitely generated sub-wheeled PROP.
\end{theorem}

The following lemma shows that any finitely generated sub-wheeled PROP has a concrete description in terms of contracting tensor products of tensors of the generators and identity tensors.
It applies in particular in the situation of \cref{thm:dm prop}.

\begin{lemma}\label{lem:sub-wheeled-prop}
  Let $\mathcal{H} = \{h_1, \dots, h_r\}$ with $h_k \in \tspace{V}{a_k}{b_k}$.
  Then the sub-wheeled PROP $\left<\mathcal{H}\right>_{\Tr}$ is linearly spanned by tensors of the form
  \begin{equation}\label{eq:contraction-expression}
    \Tr_{\gamma} \left(\Id^{\otimes d_0} \otimes \bigotimes_{k = 1}^r h_k^{\otimes d_k}\right),
  \end{equation}
  where each $d_k \in \bbN$ and $\gamma$ is a partial bijection.
\end{lemma}
\begin{proof}
  Note that the set of all tensors of the form~\eqref{eq:contraction-expression} is closed under contraction maps.
  To prove that it is also closed under tensor product, one needs to use the simple properties listed in \cref{prop:contraction-simple} to rearrange the order of operations in the expression
  \[
    \Tr_{\gamma} \left(\Id^{\otimes d_0} \otimes \bigotimes_{k = 1}^r h_k^{\otimes d_k}\right) \otimes \Tr_{\gamma'} \left(\Id^{\otimes d'_0} \otimes \bigotimes_{k = 1}^r \otimes h_k^{\otimes d'_k}\right).
  \]
  Taking into account the bilinearity of tensor product and linearity of contractions we see that that the linear span of the tensors of the form~\eqref{eq:contraction-expression} is also closed under tensor products and contractions.
  Since it also contains $\Id$, it is a sub-wheeled PROP.

  It is also clear that all tensors of the form~\eqref{eq:contraction-expression} are contained in all sub-wheeled PROPs containing $\mathcal{H}$.
  We conclude that their linear span coincides with $\left<h_1, \dots, h_r\right>_{\Tr}$.
\end{proof}

We will also see expressions similar to~\eqref{eq:contraction-expression} where instead of fixed tensors~$h_1,\dots,h_r$ we have variable tensors~$x_1, \dots, x_p$.
If the value of such an expression is scalar, we can get rid of identity tensors in the expression.

\begin{lemma}\label{lem:contraction-remove-id}
  Every function of the form
  \begin{equation}
    F(x_1, \dots, x_p) = \Tr_{\pi} \left(\Id^{\otimes d_0} \otimes \bigotimes_{k = 1}^p x_k^{\otimes d_k}\right),
  \end{equation}
  where $\pi$ is a permutation, can be expressed as
  \begin{equation}
    F(x_1, \dots, x_p) = (\dim V)^c \cdot \Tr_{\pi'} \bigotimes_{k = 1}^p x_k^{\otimes d_k}
  \end{equation}
  for some permutation $\pi'$ and some $c \in \bbN$.
\end{lemma}
\begin{proof}
  Use \cref{prop:contraction-simple}(\ref{prop:contraction-simple-3}) to remove the identity tensors contracted with other tensors (including other identity tensors).
  The remaining identity tensors contribute some power of $\Tr \left(\Id\right) = \dim V$.
\end{proof}

So far we considered tensor spaces built from a single vector space and its dual, but the definitions generalize readily to multiple vector spaces.

\begin{definition}[Tensor spaces and algebra]
    Let $\vtuple{V} = (V_1, \dots, V_m)$ be a tuple of vector spaces.
    The \emph{(mixed) tensor space of type~$\tworow{\ituple{a}}{\ituple{b}} \in \bbN^m \times \bbN^m$} over $\vtuple{V}$ is defined as
    \[
        \tspace{\vtuple{V}}{\ituple{a}}{\ituple{b}} = \bigotimes_{i = 1}^m \tspace{V_i}{a_i}{b_i} = \bigotimes_{i = 1}^m \left(V_i^{\otimes a_i} \otimes (V_i^*)^{\otimes b_i}\right).
    \]
    We abbreviate $\vtuple{V}^{\otimes\ituple{a}} = \bigotimes_{i = 1}^m V_i^{\otimes a_i}$.
    The \emph{(mixed) tensor algebra} over $\vtuple{V}$ is the $\bbN^m \times \bbN^m$-graded algebra
    \[ \mathcal{T}(\vtuple{V})
    = \bigotimes_{i = 1}^m \mathcal{T}(V_i)
    = \bigoplus_{\tworow{\ituple{a}}{\ituple{b}} \in \bbN^m \times \bbN^m} \tspace{\vtuple{V}}{\ituple{a}}{\ituple{b}}
    \]
    with the tensor product operation~$\ot$ inherited from the~$\mathcal{T}(V_i)$.
\end{definition}

Each~$\GL(V_i)$ acts on the vector spaces~$V_i$ and~$V_i^*$.
This induces a structure of a representation of
\[ \GL(\vtuple{V}) = \prod_{i = 1}^m \GL(V_i). \]
on the tensor spaces~$\tspace{\vtuple{V}}{\ituple{a}}{\ituple{b}}$ and therefore on~$\mathcal{T}(\vtuple{V})$.
The tensor algebras $\mathcal{T}(V_i)$ also come with the structure of a wheeled PROP.
This additional structure transfers to $\mathcal{T}(\vtuple{V})$ as follows.
\begin{itemize}
\item For every $i \in [m]$, the tensor space $\tspace{\vtuple{V}}{\ituple{e_i}}{\ituple{e_i}} = V_i \otimes V_i^*$ contains a distinguished element corresponding to the identity map, which we call the \emph{identity tensor on~$V_i$} and denote by~$\Id[V_i]$.
\item The contraction maps $\Tr_{\gamma_i} \colon \tspace{V_i}{a_i}{b_i} \to \tspace{V_i}{a_i - k_i}{b_i - k_i}$ compose to form contraction maps
    \[\bigotimes_{i=1}^m \Tr_{\gamma_i} \colon \tspace{\vtuple{V}}{\ituple{a}}{\ituple{b}} \to \tspace{\vtuple{V}}{\ituple{a}-\ituple{k}}{\ituple{b}-\ituple{k}}\] for any choice of partial bijections~$\gamma_1,\dots,\gamma_m$.
In particular, we denote by $\Tr[V_i]_{\gamma_i}$ the contraction map $\bigotimes_{j = 1}^{i - 1} \Tr_\varnothing \otimes \Tr_{\gamma_i} \otimes \bigotimes_{j = i + 1}^n \Tr_\varnothing$ that acts as $\Tr_{\gamma_i}$ on $\tspace{V_i}{a_i}{b_i}$ and as the identity on the rest.
\end{itemize}

We will also allow tuples of such tensor spaces:

\begin{definition}[Tensor tuple representation]\label{def:tens tup rep}
Let $\vtuple{V} = (V_1, \dots, V_n)$ be a tuple of vector spaces, and let $\group{G}$ be a reductive subgroup of $\GL(\vtuple{V})$.
A \emph{tensor tuple representation} of $\group{G}$ is a representation of the form $\rep{\bigoplus_{k = 1}^p \tspace{\vtuple{V}}{\ituple{a}_k}{\ituple{b}_k}}{\group{G}}$ for some tensor types $\tworow{\ituple{a}_1}{\ituple{b}_1}, \dots, \tworow{\ituple{a}_p}{\ituple{b}_p} \in \bbN^m \times \bbN^m$.
\end{definition}

\subsection{Tensor tuple sequences and the complexity class \OCIclass}\label{subsec:tuples}
We can obtain multivariate $p$-sequences of representations (\cref{def:p seq}) by fixing the structure of a tensor tuple representation (\cref{def:tens tup rep}) while allowing the dimension of the vector spaces to vary.
The underlying vector spaces will always be coordinate vector spaces.
The corresponding tensor spaces~$\tspace{(\bbC^{n_1}, \dots, \bbC^{n_m})}{\ituple{a}_k}{\ituple{b}_k}$ will be called \emph{standard tensor spaces}.
Each standard tensor space is equipped with a standard basis given by the tensor products of standard basis elements in~$\bbC^{n_i}$ and dual basis elements in~$\bbC^{n_i*}$.

\begin{definition}[Tensor tuple sequences]
Let $\tworow{\ituple{a}_1}{\ituple{b}_1}$, $\tworow{\ituple{a}_2}{\ituple{b}_2}$, $\dots$, $\tworow{\ituple{a}_p}{\ituple{b}_p} \in \bbN^m \times \bbN^m$.
The corresponding \emph{tensor tuple sequence} is the multivariate {$p$-sequence} $\rep{X_{\ituple{n}}}{\group{G}_{\ituple{n}}}$ where
$\group{G}_{\ituple{n}} = \prod_{i = 1}^m \GL_{n_i}$ and $X_{\ituple{n}} = \bigoplus_{k = 1}^p \tspace{(\bbC^{n_1}, \dots, \bbC^{n_m})}{\ituple{a}_k}{\ituple{b}_k}$, and $X_{\ituple{n}}$ is equipped with the standard basis as defined above.
When $p=1$, this is called a \emph{tensor sequence}.
We will also the convenient symbolic notation
\[ \rep{\bigoplus_{k = 1}^p \tspace{\vtuple{V}}{\ituple{a}_k}{\ituple{b}_k}}{\GL(\vtuple{V})}, \]
where we think of~$V_1, \dots, V_m$ as variables ranging over the coordinate vector spaces, and $\vtuple{V} = (V_1, \dots, V_m)$.

We also allow~$\group{G}_{\ituple{n}} \subset \prod_{i = 1}^m \GL_{n_i}$ to be a suitable subgroup and use analogous notation:
$\SL(\vtuple{V})$ corresponds to the group sequence $\prod_{i = 1}^{m} \SL_{n_i}$, $\Og(\vtuple{V})$~corresponds to $\prod_{i = 1}^{m} \Og_{n_i}$, and $\Sp(\vtuple{V})$ corresponds to $\prod_{i = 1}^{m} \Sp(n_i)$ (in the last case only even dimensions $n_i$ are used).
\end{definition}

\begin{example}
Using the symbolic notation, the sequence of simultaneous conjugation actions on $p$-tuples of $n\times n$-matrices is given by $(\oplus_{k=1}^p V \ot V^*, \GL(V))$ rather than by the tuples $(1,1),\dots,(1,1)$.
\end{example}

\begin{definition}[$\OCIclass$]
    We define the complexity class $\OCIclass = \OCIclass_\GL$ as the class of computational problems that are reducible (by a polynomial-time Karp reduction) to any orbit closure intersection problem for a tensor tuple sequence $\rep{\bigoplus_{k = 1}^p \tspace{\vtuple{V}}{\ituple{a}_k}{\ituple{b}_k}}{\GL(\vtuple{V})}$, where $p \in \bbN$ and $\tworow{\ituple{a}_1}{\ituple{b}_1}, \dots, \tworow{\ituple{a}_p}{\ituple{b}_p} \in \bbN^m {\times} \bbN^m$ for some $m\in\bbN$.

    The classes $\OCIclass_\Og$ and $\OCIclass_\Sp$ are defined analogously using tensor tuple sequences of the form
    $\rep{\bigoplus_{k = 1}^p \vtuple{V}^{\otimes \ituple{a}_k}}{\Og(\vtuple{V})}$ and
    $\rep{\bigoplus_{k = 1}^p \vtuple{V}^{\otimes \ituple{a}_k}}{\Sp(\vtuple{V})}$ respectively.
\end{definition}

The above classes depend on the choice of the field~$\bbF \subseteq \bbC$ over which we perform computations.
All our results hold for an arbitary field of characteristic zero with efficient arithmetic, so we omit the dependence on~$\bbF$ in the notation.

\subsection{Invariant theory of tensors tuple representations}\label{subsec:tensor invariants}
We now describe the relevant invariant theory.
We consider individual tensor tuple representations (not sequences) and hence $\vtuple{V} = (V_1,\dots,V_m)$ refers to a concrete tuple of vector spaces.

\begin{definition}[Balancedness]
    A tensor space over $\vtuple{V}$ is \emph{balanced} if it has the form $\tspace{\vtuple{V}}{\ituple{c}}{\ituple{c}}$.
    Elements of a balanced tensor space are called \emph{balanced tensors}.
    A tensor tuple representation is called \emph{balanced} if all its summands are balanced, that is, if it has the form $\bigoplus_{i = 1}^{p} \tspace{\vtuple{V}}{\ituple{c}_k}{\ituple{c}_k}$.
\end{definition}

Recall from \cref{thm:dm prop} that the invariant tensor algebras are sub-wheeled PROPs.
Using \cref{lem:sub-wheeled-prop}, this implies a fairly explicit description of the invariants for the tensor tuple representations:

\begin{theorem}\label{thm:contraction-invariants-multiple}
    Let $\vtuple{V} = (V_1, \dots, V_m)$ be vector spaces, let $\group{G}_i$ be reductive subgroups of $\GL(V_i)$, and let~$\ituple{a}_1,\ldots,\ituple{b}_p \in \bbN^m$.
    Consider the tensor tuple representation of $\group{G} = \prod_{i = 1}^m \group{G}_i$ on $X = \bigoplus_{k = 1}^p \tspace{\vtuple{V}}{\ituple{a}_k}{\ituple{b}_k}$.
    Let $h_{i1}, \dots, h_{ir_i}$ denote homogeneous generators of $\mathcal{T}(V_i)^{\group{G_i}}$.
    Then the invariant algebra $\bbC[X]^{\group{G}}$ is linearly spanned by invariants of the form
    {\small
    \begin{equation}\label{eq:contraction-invariant-multiple}
        F(x_1, \dots, x_p) \\= \bigotimes_{i = 1}^m \Tr_{\pi_i} \left(\bigotimes_{k = 1}^p x_k^{\otimes d_k} \otimes \bigotimes_{i = 1}^m \bigotimes_{\ell = 1}^{r_i} h_{i\ell}^{\otimes f_{i\ell}}\right)
    \end{equation}}%
    where the tensor in brackets is balanced and $\pi_1, \dots, \pi_m$ are permutations of appropriate sizes.
\end{theorem}

\begin{proof}
  Since tensor products and contractions are $\GL(\vtuple{V})$-equivariant, and $h_{i1}, \dots, h_{ir_i}$ are invariant under $\group{G}_i$, all polynomials of the form~\eqref{eq:contraction-invariant-multiple} are invariants.
  We need to argue that they span~$\bbC[X]^{\group{G}}$.

  Let $X_k = \tspace{\vtuple{V}}{\ituple{a}_k}{\ituple{b}_k}$, so that $X = \bigoplus_{k = 1}^p X_k$.
  The space of multihomogeneous polynomials $\bbC[X]_{(d_1, \dots, d_p)}$, where~$d_k$ is the degree in~$X_k$, is isomorphic as a representation of~$\GL(\vtuple{V})$ (and therefore of~$\group{G}$) to $\bigotimes_{k = 1}^p S^{d_k} X_k^*$.
  The latter can be equivariantly embedded into a tensor space as follows:
  \[
    \bigotimes_{k = 1}^p S^{d_k}X_k^* \subset \bigotimes_{k = 1}^p (X_k^*)^{\otimes d_k} \cong \bigotimes_{k = 1}^p \tspace{\vtuple{V}}{\ituple{b}_k}{\ituple{a}_k}^{\otimes d_k} \cong \tspace{\vtuple{V}}{\ituple{b}}{\ituple{a}}
  \]
  where $\ituple{a} = \sum_{k = 1}^p d_k \ituple{a}_k$ and $\ituple{b} = \sum_{k = 1}^p d_k \ituple{b}_k$.
  Thus, any invariant polynomial~$F$ can be recovered from a corresponding invariant tensor~$t \in \mathcal{T}(\vtuple{V})^{\group{G}}$ as \[F(x_1, \dots, x_p) = \left<t, \bigotimes_{k = 1}^p x_k^{\otimes d_k}\right>.\]
  Applying \cref{lem:sub-wheeled-prop} to each factor of the tensor product $\mathcal{T}(\vtuple{V})^{\group{G}} = \bigotimes_{i = 1}^m \mathcal{T}(V_i)^{\group{G}_i}$ we see that the invariant tensors in~$\mathcal{T}(\vtuple{V})^{\group{G}}$ are spanned by contraction expressions
  \begin{multline*}
    \bigotimes_{i = 1}^m \left(\Tr_{\gamma_i} \left(\Id[V_i]^{\otimes {f_{i0}}} \otimes \bigotimes_{\ell = 1}^r h_{i\ell}^{\otimes f_{i\ell}}\right)\right) \\= \bigotimes_{i = 1}^m \Tr_{\gamma_i} \left(\bigotimes_{i = 1}^m \Id[V_i]^{\otimes {f_{i0}}} \otimes \bigotimes_{i = 1}^m \bigotimes_{\ell = 1}^r h_{i\ell}^{\otimes f_{i\ell}}\right).
  \end{multline*}
  Therefore, the invariants are spanned by scalar contraction expressions of the form
  \[
    \bigotimes_{i = 1}^m \Tr_{\pi} \left( \bigotimes_{k = 1}^p x_k^{\otimes d_k} \otimes \bigotimes_{i = 1}^m \left(\Id[V_i]^{\otimes f_{i0}} \otimes \bigotimes_{\ell = 1}^r h_{i\ell}^{\otimes f_{i\ell}}\right)\right).
  \]
  The identity tensors can be removed using \cref{lem:contraction-remove-id}.
\end{proof}

The following definition formalizes the objects that appear in \cref{thm:contraction-invariants-multiple}, which can then be summarized as follows:
if we know tensor invariant generators for each~$\group{G}_i$, the invariants of any tensor tuple representation of $\group{G} = \prod_{i=1}^m \group{G}_i$ are spanned by the corresponding contraction invariants.

\begin{definition}
  Let $V$ be a vector space and let $\group{G}$ be a reductive subgroup of $\GL(V)$.
  If $\mathcal{T}(V)^{\group{G}} = \left<h_1, \dots, h_r\right>_{\Tr}$, with homogeneous~$h_k \in \tspace{V}{u_k}{v_k}$, then we call $h_1, \dots, h_k$ \emph{tensor invariant generators} for $\group{G}$.

  The invariants of the form~\eqref{eq:contraction-invariant-multiple} are called \emph{contraction invariants} for the tensor tuple representation $\rep{\bigoplus_{k = 1}^p \tspace{\vtuple{V}}{\ituple{a}_k}{\ituple{b}_k}}{\group{G}}$.
\end{definition}

We list here tensor invariant generators for some classical groups, determined in~\cite{Sch08-first-fundamental-MR2379100} over $\bbC$, but having coordinates in $\bbQ$, so valid over every field of characteristic zero, and the corresponding contraction invariants.
We first discuss the case that~$\group{G} = \GL(\vtuple{V})$.
In this case, there are only trace invariants.

\begin{lemma}\label{lem:gl invariants}
  It holds that
  $\mathcal{T}(V)^{\GL(V)} = \left<\varnothing\right>_{\Tr}$.
  Therefore, the invariant algebra for the action of~$\GL(\vtuple{V})$ on any~$X\nobreak=\nobreak \bigoplus_{k = 1}^p \tspace{\vtuple{V}}{\ituple{a}_k}{\ituple{b}_k}$ is linearly generated by invariants of the form
    \[
        F(x_1, \dots, x_p) = \bigotimes_{i = 1}^m \Tr_{\pi_i} \left(\bigotimes_{k = 1}^p x_k^{\otimes d_k}\right)
    \]
\end{lemma}

Next we consider products of special linear groups.

\begin{lemma}
    It holds that $\mathcal{T}(V)^{\SL(V)} = \left<\omega, \omega^*\right>_{\Tr}$, where $\omega \in V^{\otimes (\dim V)}$ and $\omega^* \in (V^*)^{\otimes (\dim V)}$ are nonzero antisymmetric tensors.
  Therefore, the invariant algebra for the action of~$\SL(\vtuple{V})$ on $X = \bigoplus_{k = 1}^p \tspace{\vtuple{V}}{\ituple{a}_k}{\ituple{b}_k}$ is linearly generated by invariants of the form
  {\small\[
    F(x_1, \dots, x_p) = \bigotimes_{i = 1}^m \Tr_{\pi_i}  \left(\bigotimes_{k = 1}^p x_k^{\otimes d_k} \otimes \bigotimes_{i = 1}^m (\omega_i^{\otimes f_i} \otimes (\omega^*_i)^{\otimes f'_i}) \right).
  \]}%
  where $\omega_i \in V_i^{\otimes (\dim V_i)}$ and $\omega^*_i \in (V_i^*)^{\otimes (\dim V_i)}$ are antisymmetric tensors.
\end{lemma}

For the orthogonal and symplectic groups, one analogously uses the corresponding bilinear forms.

\begin{lemma}\label{lem:o sp invariants}
  Let $g \in V^* \otimes V^*$ be a nondegenerate symmetric bilinear form.
  Then, $\mathcal{T}(V)^{\Og(V, g)} = \left<g, \bar{g}\right>_{\Tr}$ where $\bar{g} \in V \otimes V$ is a nonzero bilinear form on $V^*$ invariant under $\Og(V, g)$.
  The same is true for $\Sp(V, g)$ with $g, \bar{g}$ being tensors corresponding to the symplectic forms.
\end{lemma}

The nondegenerate form allows identifying~$V \cong V^*$ and accordingly one can always reduce to tensor tuple representations without covariant indices.
In this case we do not need the dual form to span the invariants:

\begin{corollary}
    Let $g_i \in V^*_i \otimes V^*_i$ be nondegenerate symmetric bilinear forms.
    Then, the invariant algebra for the action of~$\Og(\vtuple{V}, \ituple{g}) = \prod_{i=1}^m \Og(V_i, g_i)$ on $X\nobreak=\nobreak \bigoplus_{k = 1}^p \vtuple{V}^{\otimes a_k}$ is linearly generated by invariants of the form
    \[
        F(x_1, \dots, x_p) = \bigotimes_{i = 1}^m \Tr_{\pi_i}  \left(\bigotimes_{k = 1}^p x_k^{\otimes d_k} \otimes \bigotimes_{i = 1}^m g_i^{\otimes f_i} \right)
    \]
    The same is true for $\Sp(\vtuple{V},\ituple{g})$ if we take the $g_i$ to be the symplectic forms of the~$V_i$.
\end{corollary}
\begin{proof}
    \Cref{thm:contraction-invariants-multiple} gives a spanning set consisting of invariants
    \[
      F(x_1, \dots, x_p) = \Tr_{\pi} \left(\bigotimes_{k = 1}^p x_k^{\otimes d_k} \otimes \bigotimes_{i=1}^m g_i^{\otimes f_{i1}} \otimes \bar{g}_i^{\otimes f_{i2}} \right)
    \]
    Since among the tensors $x_1, \dots, x_p, g_1, \dots, \bar{g}_m$ only the~$g_i$ contains $V_i^*$-factors, in this expression every factor of every $\bar{g_i}$ tensor must be contracted with some factor of some $g_i$ tensor.
    Using the identity $\Tr^1_1 (g_i \otimes \bar{g}_i) = \Id$ and \cref{lem:contraction-remove-id}, we can transform the expression into the required form without $\bar{g_i}$ in the product.
\end{proof}

Finally we consider the symmetric group~$\group{G} = \Symg_n$, which we can think of the subgroup of permutation matrices in~$\GL_n$.

\begin{lemma}\label{lem:sym invariants}
    Denote by $\Symg_n$ the symmetric group on $n$ embedded in $\GL_n$ as permutation matrices.
    It holds that
    $\mathcal{T}(\bbC^n)^{\Symg_n} = \left<g, h\right>_{\Tr}$,
    where $g = \sum_{i = 1}^n e_i^* \otimes e_i^* \in (\bbC^{n*})^{\otimes 2}$ and $h = \sum_{i = 1}^n e_i \otimes e_i \otimes e_i \in (\bbC^n)^{\otimes 3}$.
    Therefore, the invariant algebra for the action of~$\Symg_n$ on~$X = \bigoplus_{k = 1}^p (\bbC^n)^{\otimes a_k}$ is linearly generated by invariants of the form
    \[
        F(x_1, \dots, x_p) = \Tr_{\pi} \left(\bigotimes_{k = 1}^p x_k^{\otimes d_k} \otimes g^{\otimes e_1} \otimes h^{\otimes e_2} \right).
    \]
\end{lemma}
\begin{proof}
    Let $V = \bbC^n$.
    Schrijvers~\cite[Eq.~43]{Sch08-first-fundamental-MR2379100} gives two systems of generators for the contraction-closed algebra $\mathcal{T}(V)^{\group{\Symg_n}}$.
    One of them consists of tensors $h_k = \sum_{i = 1}^n e_i^{\otimes k} \in V^{\otimes k}$ for all $k \geq 1$ and the tensors $f = \sum_{i = 1}^n e_i \otimes e_i$ and $g = \sum_{i = 1}^n e_i^* \otimes e_i^*$, which are generators for the orthogonal group containing $\Symg_n$. We have $h_2 = f$, so $f$ can be removed in this case.
    This system of generators is redundant as $h_1, h_2$ and $h_k$ for $k > 3$ can be constructed as contractions of $h_3$ and $g$ using the identities $h_{k + 1} = \Tr^{1,4}_{1,2} (g \otimes h_3 \otimes h_k)$ and $h_{k - 2} = \Tr^{1,2}_{1,2} (g \otimes h_k)$.
\end{proof}

\section{Reductions between tensor tuple sequences}\label{sec:reductions}
In this section, we describe our main results, which are efficient reductions that clarify the computational complexity of tensor tuple sequences.
In \cref{subsec:gl}, we establish the study the action of the general linear group and identify several problems that are complete for~$\OCIclass$.
In \cref{subsec:o sp}, we show that actions of the orthogonal and symplectic groups can always be reduced to the preceding, so~$\OCIclass_\Og \cup \OCIclass_\Sp \subseteq \OCIclass$.
In \cref{subsec:gi}, we relate the graph isomorphism problem to a tensor tuple sequence; this proves that~$\GIclass \subseteq \OCIclass$.
In \cref{subsec:u}, we use related methods to show that the tensor isomorphism problem for the unitary group can be reduced to an orbit closure intersection problem for closed orbits; in particular, this shows that $\TIclass_\Ug \subseteq \OCIclass \cap \TIclass$.

\subsection{Reductions and completeness for products of general linear groups}\label{subsec:gl}
In this section we prove the completeness of several specific orbit closure intersection problems for~$\OCIclass$.
To achieve this, we first establish several efficient algebraic reductions between tensor tuple sequences of certain formats.
These reductions are then composed to establish completeness.

Balanced tensors play an important role in the construction of $\GL(\vtuple{V})$-invariants.
The following lemma shows that every tensor tuple sequence can be reduced to a balanced one.

\begin{lemma}\label{lem:GL-to-balanced}
    For every $\tworow{\ituple{a}_1}{\ituple{b}_1}, \dots, \tworow{\ituple{a}_p}{\ituple{b}_p} \in \bbN^{m} \times \bbN^{m}$, there exist~$\ituple{c}_1, \dots, \ituple{c}_q \in \bbN^m$ and an efficient algebraic reduction from
    $\rep{\bigoplus_{k = 1}^p \tspace{\vtuple{V}}{\ituple{a\ituple{}}_k}{\ituple{b\ituple{}}_k}}{\GL(\vtuple{V})}$
    to the balanced sequence
    $\rep{\bigoplus_{\ell = 1}^{q}\tspace{\vtuple{V}}{\ituple{\ituple{c}_\ell}}{\ituple{\ituple{c}_\ell}}}{\GL(\vtuple{V})}$.
\end{lemma}
\begin{proof}
    Consider the set of all tuples $\ituple{d} = (d_1, \dots, d_p) \in \bbN^p$ such that the tensor product $\bigotimes_{k = 1}^p x_k^{\otimes d_k}$ is balanced for $x_k \in \tspace{\vtuple{V}}{\ituple{a}_k}{\ituple{b}_k}$.
    These tuples are exactly the integer points in the rational polyhedral cone~$C$
    given by inequalities $d_k \geq 0$ and linear equations $\sum_{k = 1}^p a_k d_k = \sum_{k = 1}^p b_k d_k$ enforcing the balance.
    By \cref{thm:hilbert}, there is an integral Hilbert basis~$\ituple{h}_1, \dots, \ituple{h}_q \in C \cap \bbZ^p$.

    Set~$\ituple{c}_\ell = \sum_{k = 1}^p \ituple{a}_k h_{\ell k}$ and define a map
    \[\alpha \colon \bigoplus_{k = 1}^p \tspace{\vtuple{V}}{\ituple{a}_k}{\ituple{b}_k} \to \bigoplus_{\ell = 1}^{q}\tspace{\vtuple{V}}{\ituple{c_\ell}}{\ituple{c_\ell}}\]
    given by
    \[ \alpha(x_1, \dots, x_p) = (y_1, \dots, y_q)\quad \text{where } y_\ell = \bigotimes_{k = 1}^p x_k^{\otimes h_{\ell k}}.\]
    Because the map is $\GL(\vtuple{V})$-equivariant it preserves closure equivalence by \cref{lem:beta crit}.

    We will now prove that it also reflects closure equivalence.
    To see this, recall from \cref{lem:gl invariants} that the invariant algebra for $\rep{\bigoplus_{k = 1}^p \tspace{\vtuple{V}}{\ituple{a}_k}{\ituple{b}_k}}{\GL(\vtuple{V})}$ is linearly spanned by contraction invariants of the form
    \[
        G(x_1, \dots, x_p) = \bigotimes_{i = 1}^m \Tr_{\pi_i} \left( \bigotimes_{k = 1}^p x_k^{\otimes d_k} \right),
    \]
    for some permutations $\pi_i$ and where the expression in parentheses is balanced.
    The latter means that $\sum_{k = 1} a_k d_k = \sum_{k = 1} b_k d_k$, so~$\ituple{d} = (d_1,\dots,d_p) \in C \cap \bbZ^p$, hence it can be represented as $\ituple{d} = \sum_{\ell = 1}^q e_\ell \ituple{h}_\ell$ for certain $e_\ell \in \bbN$.
    It follows that $\bigotimes_{k = 1}^p x_k^{\otimes d_k} = \bigotimes_{\ell=1}^q y_\ell^{\otimes e_\ell}$, where $(y_1, \dots, y_q) = \alpha(x_1, \dots, x_p)$.
    Thus the invariant~$G$ can be represented as $F \circ \alpha$, where
    \[
        F(y_1, \dots, y_q) = \bigotimes_{i = 1}^m \Tr_{\pi_i} \left( \bigotimes_{\ell = 1}^q y_\ell^{\otimes e_\ell} \right)
    \]
    is a contraction invariant for $\rep{\bigoplus_{\ell = 1}^{q}\tspace{\vtuple{V}}{\ituple{c_\ell}}{\ituple{c_\ell}}}{\GL(\vtuple{V})}$.
    Applying \cref{thm:reflect}, we see that the map~$\alpha$ reflects closure equivalence.

    By applying this construction to~$\vtuple{V} = (\bbC^{n_1}, \dots, \bbC^{n_m})$, we get an algebraic reduction~$\alpha_{\ituple{n}}$ from $\rep{\bigoplus_{k = 1}^p \tspace{\vtuple{V}}{\ituple{a}_k}{\ituple{b}_k}}{\GL(\vtuple{V})}$ to $\rep{\bigoplus_{\ell = 1}^{q}\tspace{\vtuple{V}}{\ituple{c_\ell}}{\ituple{c_\ell}}}{\GL(\vtuple{V})}$, where~$\ituple{n}=(n_1,\dots,n_m)$.
    It is clear that these maps are efficiently computable.
    Thus we have constructed an efficient algebraic reduction.
\end{proof}

\begin{remark}
    It is possible that no balanced tensor can be constructed from the original tensor tuple.
    In this case we get a reduction to the trivial sequence $\rep{0}{\GL(\vtuple{V})}$, which means that every two vectors are closure equivalent because $0$ is contained in every orbit closure (every tensor tuple is in the null cone).
\end{remark}

The following lemma allows to further simplify balanced tensor tuple sequences to direct sums of the \emph{same} tensor space.

\begin{lemma}\label{lem:GL-balanced-to-sum}
    Let $\ituple{c}_1,\dots,\ituple{c}_q$ and $\ituple{c} \in \bbN^m$ with $\ituple{c} \geq \ituple{c_\ell}$ for all~$\ell$ (where~$\geq$ is co\-ord\-i\-nate-wise inequality).
    Then there is an efficient algebraic reduction from~$\rep{\bigoplus_{\ell = 1}^{q}\tspace{\vtuple{V}}{\ituple{c_\ell}}{\ituple{c_\ell}}}{\GL(\vtuple{V})}$ to~$\rep{\left(\tspace{\vtuple{V}}{\ituple{c}}{\ituple{c}}\right)^{\oplus q}}{\GL(\vtuple{V})}$.
\end{lemma}
\begin{proof}
    Define $\bar{\ituple{c}}_\ell = \ituple{c} - \ituple{c}_\ell$ for~$\ell\in[q]$.
    Note that every balanced tensor space
    $\tspace{\vtuple{V}}{\ituple{c}_\ell}{\ituple{c}_\ell}$
    contains a $\GL(\vtuple{V})$-invariant tensor~$\Id_{\bar{c}_\ell}$ corresponding to the identity map on $\vtuple{V}^{\ot \bar{\ituple{c}}_\ell}$.
    Define $\GL(\vtuple{V})$-equivariant injective linear maps~$\iota_\ell \colon \tspace{\vtuple{V}}{\ituple{c}_\ell}{\ituple{c}_\ell} \to \tspace{\vtuple{V}}{\ituple{c}}{\ituple{c}}$ by tensor with this identity: $\iota_\ell(x) = x \otimes \Id_{\bar{c}_\ell}$.
    Then, $\iota = \iota_1 \oplus \dots \oplus \iota_q$ is an~$\GL(\vtuple{V})$-equivariant injective linear map from
    $\bigoplus_{\ell = 1}^{q} \tspace{\vtuple{V}}{\ituple{c}_\ell}{\ituple{c}_\ell}$ to~$\left(\tspace{\vtuple{V}}{\ituple{c}}{\ituple{c}}\right)^{\oplus q}$.
    By \cref{ex:injective-morphism}, it follows that~$\iota$ preserves and reflects closure equivalence.
    By applying this construction to~$\vtuple{V} = (\bbC^{n_1}, \dots, \bbC^{n_m})$, we obtain an efficient algebraic reduction.
\end{proof}

We can also reduce any tensor tuple sequences to one for a single~$\GL(W)$.

\begin{lemma}\label{lem:GL-to-one-space}
    Let $\tworow{\ituple{a}_1}{\ituple{b}_1}, \dots, \tworow{\ituple{a}_p}{\ituple{b}_p} \in \bbN^{m} \times \bbN^{m}$ and define~$A_k = \sum_{i = 1}^m a_{k,i}$ and~$B_k = \sum_{i = 1}^m b_{k,i}$.
    Then there is an efficient algebraic reduction from~$\rep{\bigoplus_{k = 1}^p \tspace{\vtuple{V}}{\ituple{a}_k}{\ituple{b}_k}}{\GL(\vtuple{V})}$ to~$\rep{\bigoplus_{k = 1}^p \tspace{W}{A_k}{B_k}}{\GL(W)}$.
\end{lemma}
\begin{proof}
    Let $\vtuple{V} = (V_1, \dots, V_m)$ be a tuple of vector spaces, and let $W = \bigoplus_{i = 1}^m V_i$.
    Let $\iota_i$ be the injection of $V_i$ into $W$ as a direct summand, and let $\bar\iota_i$ be the injection of $V_i^*$ into $W^* \cong \bigoplus_{i = 1}^m V_i^*$.
    For every~$\tworow{\ituple{a}}{\ituple{b}}$ we have an injective linear map $\iota_{\tworow{\ituple{a}}{\ituple{b}}}\colon \tspace{\vtuple{V}}{\ituple{a}}{\ituple{b}} \to \tspace{W}{A}{B}$ where~$A = \sum_{i = 1}^m a_i$ and~$B = \sum_{i = 1}^m b_i$, defined by $\iota_{\tworow{\ituple{a}}{\ituple{b}}} = \bigotimes_{i = 1}^m \iota_i^{\otimes a_i} \otimes \bigotimes_{i = 1}^m \bar\iota_i^{\otimes b_i}$.
    This gives us an injective linear map~$\iota = \bigoplus_{k = 1}^p \iota_{\tworow{\ituple{a}_k}{\ituple{b}_k}}$ from $\bigoplus_{k = 1}^p \tspace{\vtuple{V}}{\ituple{a}_k}{\ituple{b}_k}$ to $\bigoplus_{k = 1}^p \tspace{W}{A_k}{B_k}$.
    Let~$\beta$ be the standard embedding of~$\GL(\vtuple{V})$ into~$\GL(W)$ given by $\beta(\group{g}_1, \dots, \group{g}_m) = \group{g}_1 \oplus \dots \oplus \group{g}_m$.
    Then the maps~$\iota_i, \bar\iota_i$ are equi\-vari\-ant with respect to~$\beta$, so the same is true for~$\iota$, and \cref{lem:beta crit} shows that~$\iota$ preserves closure equivalence.

    To prove that it also reflects closure equivalence, recall from \cref{lem:gl invariants} that the invariants for $\rep{\bigoplus_{k = 1}^p \tspace{\vtuple{V}}{\ituple{a}_k}{\ituple{b}_k}}{\GL(\vtuple{V})}$ are linearly spanned by contraction invariants of the form
    \begin{align}\label{eq:G in single GL reduction}
        G(x_1, \dots, x_p) = \bigotimes_{i = 1}^m \Tr_{\pi_i} \left( \bigotimes_{k = 1}^p x_k^{\otimes d_k} \right).
    \end{align}
    Note that if $x \in \tspace{\vtuple{V}}{\ituple{a}}{\ituple{b}}$ and $y \in \tspace{\vtuple{V}}{\ituple{c}}{\ituple{d}}$, then the tensors $\iota_{\tworow{\ituple{a}}{\ituple{b}}}(x) \otimes \iota_{\tworow{\ituple{c}}{\ituple{d}}}(y)$ and $\iota_{\tworow{\ituple{a} + \ituple{c}}{\ituple{b} + \ituple{d}}} (x \otimes y)$ are obtained from each other by a permutation of factors.
    In addition, since for $u \in V_i^*$ and $v \in V_i$ we have $(\bar\iota_i u) \cdot (\iota_i v) = u \cdot v$, the injections $\iota_{\tworow{\ituple{a}}{\ituple{b}}}$ respect contractions in the following way:
    \[\iota_{\tworow{\ituple{a} - \ituple{e}_j}{\ituple{b} - \ituple{e}_j}} (\Tr[V_j]^p_q x) = \Tr^P_Q (\iota_{\tworow{\ituple{a}}{\ituple{b}}} x)\]
    where $P = \sum_{i = 1}^{j - 1} a_i + p$ and $Q = \sum_{i = 1}^{j - 1} b_i + q$ are the indices of factors of $\iota_{\tworow{a}{b}} x$ corresponding to the $p$-th~$V_i$ and the $q$-th $V_i^*$ factors of $x$.
    Using these two properties, the contraction invariant in \cref{eq:G in single GL reduction} can be rewritten as $G = F \circ \iota$, where $F(y_1, \dots, y_p) = \Tr_{\pi} \bigotimes_{k = 1}^p y_k^{\otimes d_k}$ for some permutation~$\pi$.
    The latter is an invariant for $\rep{\bigoplus_{k = 1}^p (W^{\otimes A_k} \otimes (W^*)^{\otimes B_k})}{\GL(W)}$.
    By \cref{thm:reflect}, it follows that~$\iota$ reflects closure equivalence.
    As in the preceding proofs, this map can be efficiently implemented (simply copy each input coordinates to the appropriate output coordinate and set all other to~$0$) and hence we obtain an efficient algebraic reduction.
\end{proof}

The following lemma establishes a partial converse to the preceding one:
it shows that tensor tuple sequences over one vector space can be reduced to smaller tensor tuple sequences over several vector spaces, at the expense of some overhead.

\begin{lemma}\label{lem:GL-from-one-space}
    Let $\tworow{\ituple{a}_1}{\ituple{b}_1}, \dots, \tworow{\ituple{a}_p}{\ituple{b}_p} \in \bbN^{m} \times \bbN^{m}$ and define~$A_k = \sum_{i = 1}^m a_{k,i}$ and~$B_k = \sum_{i = 1}^m b_{k,i}$.
    Then there is an efficient algebraic reduction from~$\rep{\bigoplus_{k = 1}^p \tspace{W}{A_k}{B_k}}{\GL(W)}$
    to
    \[ \rep{\bigoplus_{k = 1}^p \tspace{\vtuple{V}}{\ituple{a}_k}{\ituple{b}_k} \;\op\; \bigoplus_{i = 2}^m (V_i \otimes V_1^* \oplus V_1 \otimes V_i^*)}{\GL(\vtuple{V})}. \]
\end{lemma}
\begin{proof}
    Let $W$ be a vector space, and let $\vtuple{V} = (V_1, \dots, V_m)$ be a tuple consisting of $m$ spaces isomorphic to $W$.
    Fix isomorphisms $\alpha_i \colon W \to V_i$ and let $\bar \alpha_i = (\alpha_i^{-1})^*$ be the isomorphisms~$W^* \to V_i^*$, so that $(\bar \alpha_i u) \cdot (\alpha_i v) = u \cdot v$.
    For every $\tworow{\ituple{a}}{\ituple{b}} \in \bbN^m$ such that $\sum_{i = 1}^m a_i = A$ and $\sum_{i = 1}^m b_i = B$, define the map $\alpha_{\tworow{\ituple{a}}{\ituple{b}}} \colon W^{\otimes A} \otimes (W^*)^{\otimes B} \to \tspace{\vtuple{V}}{\ituple{a}}{\ituple{b}}$ as $\alpha_{\tworow{\ituple{a}}{\ituple{b}}} = \bigotimes_{i = 1}^m \alpha_i^{\otimes a_i} \otimes \bigotimes_{i = 1}^m \bar\alpha_i^{\otimes b_i}$.
    Let $f_i = \alpha_i \alpha_1^{-1}$ and $g_i = f_i^{-1} = \alpha_1 \alpha_i^{-1}$.
    Note that $f_1 = g_1 = \id$.
    We can identify $f_i,g_i$ with tensors in~$V_i \otimes V_1^*$ and~$V_1 \otimes V_i^*$, respectively.
    Now define a map $\alpha$ from $\bigoplus_{k = 1}^p (W^{\otimes A_k} \otimes (W^*)^{\otimes B_k})$ to $\bigoplus_{k = 1}^p \tspace{\vtuple{V}}{\ituple{a}_k}{\ituple{b}_k} \oplus \bigoplus_{i = 2}^m (V_i \otimes V_1^* \oplus V_1 \otimes V_i^*)$
    given by
    \begin{multline*}
    \alpha(x_1, \dots, x_p) \\ = (\alpha_{\tworow{a_1}{b_1}}(x_1), \dots, \alpha_{\tworow{a_p}{b_p}}(x_p), f_2, g_2, \dots, f_m, g_m).
    \end{multline*}
    The linear space isomorphisms $\alpha_i$ induce group isomorphisms $\beta_i\colon \GL(W) \to \GL(V_i)$ defined as $\beta_i(\group{g}) = \alpha_i \group{g} \alpha_i^{-1}$.
    They can be combine into an embedding~$\beta \colon \GL(W) \to \GL(\vtuple{V})$ given by~$\beta(\group{g}) = (\beta_1(\group{g}), \dots, \beta_m(\group{g}))$.
    The maps~$\alpha_i$ are $\GL(W)$-equivariant in the sense that $\alpha_i (\group{g} v) = \beta_i(\group{g}) \alpha_i(v)$ for $\group{g} \in \GL(W)$ and~$v \in W$, and likewise for the maps~$\bar\alpha_i$.
    In addition, $f_i$ and $g_i$ are invariant under the action of $\GL(W)$ via $\beta_1$ on $V_1$ and $\beta_i$ on $V_i$.
    Using \cref{lem:beta crit} it follows that~$\alpha$ preserves closure equivalence.

    To see that it also reflects closure equivalence, note that for every tensor~$x$ we have
    \begin{multline*}
        \left( \bigotimes_{i = 1}^m f_i^{\otimes a_i} \otimes \bigotimes_{i = 1}^m (g_i^*)^{\otimes b_i} \right) \alpha_{\tworow{a}{b}}(x)
     \\ = \left(\bigotimes \alpha_1^{\otimes A} \otimes \bigotimes \bar\alpha_1^{\otimes B}\right)(x).
    \end{multline*}
    Now, every contraction invariant 
    \begin{align*}
        G(x_1, \dots, x_p) = \bigotimes_{i = 1}^m \Tr_{\pi_i} \left( \bigotimes_{k = 1}^p x_k^{\otimes d_k} \right)
    \end{align*}
    for $\rep{\bigoplus_{k = 1}^p (W^{\otimes A_k} \otimes (W^*)^{\otimes B_k})}{\GL(W)}$ can be written as
    \[\bigotimes_{i = 1}^m \Tr_{\pi_i} \left( \bigotimes_{k = 1}^p z_k^{\otimes d_k} \right), z_k = \left(\bigotimes_{i = 1}^m f_i^{\otimes a_{ki}} \otimes \bigotimes_{i = 1}^m (g_{i}^*)^{\otimes b_{ki}}\right) y_k, \]
    where $f_1 = g_1 = \Id$ and $(y_1, \dots, y_p, f_2, g_2, \dots, f_m, g_m) = \alpha(x_1, \dots, x_p)$.
    Note that the tensors $z_k$ can be constructed from $y_1,\dots,y_p$, $f_2, g_2$, \dots, $f_m, g_m$ using tensor contractions ($f_1 = g_1 = \Id$ can be omitted by~\Cref{lem:contraction-remove-id}), so
    this expression gives a contraction invariant~$F$ for {\small $\rep{\bigoplus_{k = 1}^p \tspace{\vtuple{V}}{\ituple{a}_k}{\ituple{b}_k} \oplus \bigoplus_{i = 1}^m (V_i \otimes V_1^* \oplus V_1 \otimes V_i^*)}{\GL(\vtuple{V})}$} such that $G = F \circ \alpha$.
    By \cref{lem:gl invariants,thm:reflect}, we conclude that~$\alpha$ reflects closure equivalence.

    The maps~$\alpha$ can be efficiently computed if $W = V_1 = \dots = V_m = \bbC^n$ and we take $\alpha_i$ to be the identity maps, since then the $\alpha_{\tworow{\ituple{a}}{\ituple{b}}}$ are also identity maps, and the~$f_i$ and~$g_i$ are the identity tensors~$\sum_{i = 1}^n e_i \otimes e_i^*$ seen as elements of~$V_i \otimes V_1^*$ and~$V_1 \otimes V_i^*$ respectively.
    Thus, we obtain an efficient algebraic reduction.
\end{proof}

Together with the previous statements, the following lemma allows reducing any tensor tuple sequence to tuples where each tensor space is a tensor product of the different vector spaces and their duals (that is, all tensor powers are equal to one).

\begin{lemma}\label{lem:GL-balanced-sum-to-multilinear-sum}
    Let $\ituple{c} \in \bbN^m$. Suppose $M = \sum_{i = 1}^m c_i >0$.
    There is an efficient algebraic reduction from~$\rep{\left(\tspace{\vtuple{V}}{\ituple{\ituple{c}}}{\ituple{\ituple{c}}}\right)^{\oplus p}}{\GL(\vtuple{V})}$ to $\rep{\left(\bigotimes_{i = 1}^M W_i \otimes W_i^*\right)^{\oplus(p+M-1)}}{\GL(\vtuple{W})}$.
\end{lemma}
\begin{proof}
    This is a combination of the preceding reductions.
    First, use \cref{lem:GL-to-one-space} to reduce from $\rep{(\tspace{\vtuple{V}}{\ituple{c}}{\ituple{c}})^{\oplus p}}{\GL(\vtuple{V})}$ to $\rep{(\tspace{X}{M}{M})^{\oplus p}}{\GL(X)}$.
    Then use \cref{lem:GL-from-one-space} with $\ituple{a}_k = \ituple{b}_k = (1,\dots,1)$ for all $k\in[p]$ to get a reduction to
    \[ \rep{\left(\bigotimes_{i = 1}^M W_i {\otimes}  W_i^*\right)^{\oplus p} \hspace{-6pt} \oplus \bigoplus_{i = 2}^M (W_i {\otimes} W_1^* \oplus W_1 {\otimes} W_i^*)}{\GL(\vtuple{W})}. \]
    Next, follow the proof of \cref{lem:GL-to-balanced} to reduce to (the right-hand side sums become tensor products)
    \[ \rep{\left(\bigotimes_{i = 1}^M W_i {\otimes} W_i^*\right)^{\oplus p} \hspace{-6pt} \oplus \bigoplus_{i = 2}^{M} (W_1 {\otimes} W_1^* {\otimes} W_i {\otimes} W_i^*)}{\GL(\vtuple{W})}. \]
    And finally, apply \cref{lem:GL-balanced-to-sum} to construct a reduction to
    $\rep{(\bigotimes_{i = 1}^M W_i \otimes W_i^*)^{\oplus(p + M - 1)}}{\GL(\vtuple{W})}$.
\end{proof}

Now we establish a reduction of the tensor tuple sequences obtained in the preceding lemma ones that have a \emph{constant} number of tensor factors in each summand, for a single~$\GL(X)$.

\begin{lemma}\label{lem:GL-multilinear-sum-to-deg2-sum}
    For every~$p,m\in\bbN$, there is an efficient algebraic reduction from $\rep{\left(\bigotimes_{i = 1}^m V_i \otimes V_i^*\right)^{\oplus p}}{\GL(\vtuple{V})}$ to $\rep{(X \otimes X^*)^{\oplus p} \;\oplus\; (X^{\otimes 2} \otimes (X^*)^{\otimes 2})^{\oplus m}}{\GL(X)}$.
\end{lemma}
\begin{proof}
    Let $\vtuple{V} = (V_1, \dots, V_m)$ be an arbitrary tuple of vector spaces, and let $X = \bigotimes_{i = 1}^m V_i$.
    Note that the tensor space $X^{\otimes a} \otimes (X^*)^{\otimes b}$ is isomorphic to $\vtuple{V}^{\otimes\tworow{m \times a}{m \times b}}$ by rearranging of factors.
    Here and in the following we use the notation $m \times a=(a,\dots,a)\in\bbN^m$.
    For a permutation $\pi \in \Symg_d$ denote by $r_{i,\pi}\colon V_i^{\otimes d} \to V_i^{\otimes d}$ the map permuting the tensor factors according to the permutation~$\pi$, and by $R_{i, \tau} \colon X^{\otimes d} \to X^{\otimes d}$ the map applying $r_{i, \pi}$ to the factor~$V_i^{\otimes d}$.
    Let $s_i \in V_i^{\otimes 2} \otimes (V_i^*)^{\otimes 2}$ be the tensor corresponding to the linear map $r_{i,(12)} \colon V_i^{\otimes 2} \to V_i^{\otimes 2}$ swapping the two factors.
    Then $S_i = \bigotimes_{j = 1}^{i - 1} \Id_{V_j}^{\otimes 2} \otimes s_i \otimes \bigotimes_{j = i + 1}^{m} \Id_{V_j}^{\otimes 2} \in \vtuple{V}^{\otimes\tworow{m \times 2}{m \times 2}} \cong X^{\otimes 2} \otimes (X^*)^{\otimes 2}$ corresponds to $R_{i,(12)}$.
    Note that if $\tau = (pq)$ is a transposition in $\Symg_d$, then the map $R_{i,\tau}$ permuting two $V_i$ factors can be obtained by applying $R_{i,(12)}$ to the $p$-th and $q$-th factors of the product, or via contraction with $S_i$.
    Now define the map
    \begin{align*}
        \alpha \colon \left(\bigotimes_{i = 1}^m V_i {\otimes} V_i^*\right)^{\oplus p} \! &\to \left( X {\otimes} X^* \right)^{\oplus p} \oplus \left( X^{\otimes 2} {\otimes} (X^*)^{\otimes 2} \right)^{\oplus m}, \\
        \alpha(x_1, \dots, x_p) &= (x_1, \dots, x_p, S_1, \dots, S_m).
    \end{align*}
    Because $\alpha$ is equivariant with respect to the embedding $\beta \colon \GL(\vtuple{V}) \to \GL(X)$ given by~$\beta(\group{g}_1, \dots, \group{g}_m) = \bigotimes_{i = 1}^m \group{g}_i$, by \cref{lem:beta crit} it preserves closure equivalence.

    To show that it also reflects closure equivalence, we proceed as in the preceding proofs.
    Consider an arbitrary contraction invariant for $\rep{\left(\bigotimes_{i = 1}^m V_i \otimes V_i^*\right)^{\oplus p}}{\GL(\vtuple{V})}$,
    \begin{align*}
        G(x_1, \dots, x_n) & = \bigotimes_{i = 1}^m \Tr_{\pi_i} \left( \bigotimes_{k = 1}^p x_k^{\otimes d_k} \right) \\ & = \Tr^{\otimes m} \left( \prod_{i = 1}^m R_{i,\pi} \bigotimes_{k = 1}^p x_k^{\otimes d_k} \right).
    \end{align*}
    Decomposing the product of permutations into a product of transpositions, we can rewrite this as
    \[
        G(x_1, \dots, x_n) = \Tr^{\otimes m} \left( \prod_{j = 1}^\ell R_{i_j,\tau_j} \bigotimes_{k = 1}^p x_k^{\otimes d_k} \right).
    \]
    for suitable~$i_1,\dots,i_\ell$ and transpositions~$\tau_1,\dots,\tau_\ell$.
    Since the transpositions $R_{i_j, \tau_j}$ can be implemented by contractions with $S_{i_j}$, we see that $G = G \circ \alpha$ for some contraction invariant $F$ for $\rep{(X \otimes X^*)^{\oplus p} \oplus (X^{\otimes 2} \otimes (X^*)^{\otimes 2})^{\oplus m}}{\GL(X)}$.
    As in the previous proofs, \cref{thm:reflect} now shows that~$\alpha$ reflects closure equivalence, and we obtain an efficient algebraic reduction.
\end{proof}

We can combine the preceding results to reduce from certain tensor tuple sequences for a single~$\GL(X)$ to ones for a product group~$\GL(X) \times \GL(Y)$.

\begin{lemma}\label{lem:GL-deg2-to-peps}
    For every~$p,m\in\bbN$, there is an efficient algebraic reduction from
    \[\rep{(X \otimes X^*)^{\oplus p} \;\oplus\; (X^{\otimes 2} \otimes (X^*)^{\otimes 2})^{\oplus m}}{\GL(X)}\]
    to
    \[\rep{(X {\otimes} X^*)^{\oplus p} \oplus (X {\otimes} X^* {\otimes} Y {\otimes} Y^*)^{\oplus (m + 1)}}{\GL(X) {\times} \GL(Y)}.\]
\end{lemma}
\begin{proof}
    First apply \cref{lem:GL-from-one-space} (with $m=2$ and $\ituple{a}_k = \ituple{b}_k$ equal to $(1,0)$ for $p \leq k$ and $(1,1)$ for~$p>k$) to reduce from $\rep{(X \otimes X^*)^{\oplus p} \;\oplus\; (X^{\otimes 2} \otimes (X^*)^{\otimes 2})^{\oplus m}}{\GL(X)}$ to
    \begin{multline*}\left((X \otimes X^*)^{\oplus p} \oplus (X \otimes X^* \otimes Y \otimes Y^*)^{\op m}  \right. \\ \left. {} \oplus (Y\otimes X^* \oplus X \otimes Y^* );\GL(X) \times \GL(Y)\right).\end{multline*}
    To balance the summand $Y\otimes X^* \oplus X \otimes Y^*$, we can use \cref{lem:GL-to-balanced}.
    Following its proof, a Hilbert basis of the relevant cone is given by~$(1,1)$, corresponding to~$X \otimes X^* \otimes Y \otimes Y^*$.
    Overall, we obtain reduction to {\small\[\rep{(X {\otimes} X^*)^{\oplus p} \oplus (X \otimes X^* \otimes Y \otimes Y^*)^{\oplus (m + 1)}}{\GL(X) \times \GL(Y)}.\]}
\end{proof}

Our final technical lemma allows reducing the number of summands to a constant:

\begin{lemma}\label{lem:GL-sum-to-summand}
    For every $p\in\bbN$ there exists an efficient algebraic reduction from the tensor tuple sequence $\rep{(X \otimes X^* \otimes Y \otimes Y^*)^{\oplus p}}{\GL(X)\times \GL(Y)}$ to $\rep{Z \otimes Z^* \otimes Y \otimes Y^* \;\oplus\; (Z \otimes Z^*)^{\oplus 2}}{\GL(Z) \times \GL(Y)}$, and to $\rep{Z \otimes Z^* \otimes Y \otimes Y^* \oplus Z^{\otimes 2} \otimes (Z^*)^{\otimes 2}}{\GL(Z) \times \GL(Y)}$.
\end{lemma}
\begin{proof}
    Let $X$ and $Y$ be arbitrary vector spaces, set $Z = X^{\oplus p}$ and define $\iota_i \colon X \to Z$ and $\bar\iota_i$ as the injections into $i$-th direct summand.
    Let $r,s$ be the tensors in~$Z \otimes Z^*$ that correspond to the linear maps $r(x_1, \dots, x_p) = (x_2, \dots, x_p, x_1)$ and $s(x_1, \dots, x_p) = (x_1, 0, \dots, 0)$, that is, $r$ cyclically permutes the~$p$~direct summands and $s$ is the projection onto the first summand.
    Note that $s r^{k-1}$ maps $(x_1, \dots, x_p)$ to $\iota_1(x_k) = (x_k, 0, \dots, 0)$.
    Define a map
    \begin{align*}
    \alpha\colon (X {\otimes} X^* {\otimes} Y {\otimes} Y^*)^{\oplus p} &\to (Z {\otimes} Z^* {\otimes} Y {\otimes} Y^*) \oplus (Z {\otimes} Z^*)^{\oplus 2}, \\
        \alpha(x_1, \dots, x_p) &= (y, r, s) \\ \text{where } y &= \sum_{k = 1}^p (\iota_k \otimes \bar \iota_k) x_k.
    \end{align*}
    By \cref{lem:beta crit}, $\alpha$ preserves closure equivalence because it is equivariant with respect to the embedding $\beta \colon \GL(X)\times \GL(Y) \to \GL(Z) \times \GL(Y)$ given by $(\group{g}, \group{h} \mapsto (\group{g}^{\oplus p}, \group{h})$.

    We follow the familiar recipe to prove that it also reflects closure equivalence.
    Consider a contraction invariant for $\rep{(X \otimes X^* \otimes Y \otimes Y^*)^{\oplus p}}{\GL(X) \times \GL(Y)}$,
    \[
        G(x_1, \dots, x_p) = (\Tr_{\pi_1} \otimes \Tr_{\pi_2}) \left( \bigotimes_{k = 1}^p x_k^{\ot d_k} \right).
    \]
    Because $y_k = (r^{k - 1} s) \otimes (s^* (r^*)^{k - 1})) y = (\iota_1 \otimes \bar\iota_1) x_k$, we can write this as
    \[
        G(x_1, \dots, x_p) = \left(\Tr_{\pi_1} \otimes \Tr_{\pi_2}\right) \left( \bigotimes_{k = 1}^p y_k^{\otimes d_k} \right),
    \]
    and since each~$y_k$ can be expressed as tensor contractions of $y$, $r$ and $s$, we have $G = F \circ \alpha$ for some contraction invariant $F(y, r, s)$ on $Z \otimes Z^* \otimes Y \otimes Y^* \oplus (Z \otimes Z^*)^{\oplus 2}$.
    Thus $\alpha$ reflects closure equivalence by \cref{thm:reflect} and we obtain the first of the desired efficient algebraic reduction.

    For the second, observe $r = \Tr[Z]^2_2(r \otimes s)$ and \[s = \frac{1}{\dim Z}\Tr[Z]_1^1(\Id[Z] \otimes s) = \frac{1}{\dim Z}\Tr[Z]_1^1( (r \otimes s)^p ),\]
    every tensor contraction of $y$, $r$ and $s$ can be expressed as a tensor contraction of $y$ and $r \otimes s$.
    Thus, if we define
    \begin{align*}
        \alpha' \colon \!(X {\otimes} X^* {\otimes} Y {\otimes} Y^*)^{\oplus p} &\to Z {\otimes} Z^* {\otimes} Y {\otimes} Y^* \oplus Y^{\otimes 2} {\otimes} (Y^*)^{\otimes 2}, \\
        \alpha'(x_1, \dots, x_p) &= (y, r \otimes s),
    \end{align*}
     with~$y$ as above, then one can proceed as above to see that~$\alpha'$ preserves and reflects closure equivalence.
     In this way we get the second of the desired efficient algebraic reductions.
\end{proof}

We now have all ingredients to establish the completeness of the orcit closure equivalence problems for the following simultaneous ``double conjugation'' action (that is, for the action of the gauge group for PEPS tensor networks with physical dimension~3).

\begin{theorem}\label{thm:complete PEPS}
    The OCI problem for the tensor tuple sequence
    $\rep{\left(X \otimes X^* \otimes Y \otimes Y^*\right)^{\oplus 3}}{\GL(X) \times \GL(Y)}$
    is complete for $\OCIclass$.
\end{theorem}
\begin{proof}
We need to argue that any tensor tuple sequence $\rep{\bigoplus_{k = 1}^p \tspace{\vtuple{V}}{\ituple{a}_k}{\ituple{b}_k}}{\group{\GL(\vtuple{V})}}$ can be reduced to the sequence above.
This can be seen as follows.

First, use \cref{lem:GL-to-balanced} we reduce to a balanced tensor tuple sequence
$\rep{\bigoplus_{k = 1}^q \tspace{\vtuple{V}}{\ituple{c}_k}{\ituple{c}_k}}{\group{\GL(\vtuple{V})}}$.
Next, apply \cref{lem:GL-balanced-to-sum} and \cref{lem:GL-balanced-sum-to-multilinear-sum} to reduce to the tensor tuple sequence of the form
\[\rep{\left(\bigotimes_{i = 1}^m W_i \otimes W_i^* \right)^{\oplus r}}{\GL(\vtuple{W})},\] with $\vtuple{W} = (W_1, \dots, W_M)$ for some $M$.
Then, use \cref{lem:GL-multilinear-sum-to-deg2-sum} to reduce from this sequence to
\[\rep{(X \otimes X^*)^{\oplus r} \oplus (X^{\otimes 2} \otimes (X^*)^{\otimes 2})^{\oplus M}}{\GL(X)},\]
apply \cref{lem:GL-balanced-to-sum} to reduce to \[\rep{(X^{\otimes 2} \otimes (X^*)^{\otimes 2})^{\oplus (r + M)}}{\GL(X)},\]
and \cref{lem:GL-deg2-to-peps} to further reduce to \[\rep{(X \otimes X^* \otimes Y \otimes Y^*)^{\oplus (r + M + 1)}}{\GL(X) \times \GL(Y)}.\]
Finally, we can use \cref{lem:GL-sum-to-summand} to reduce to \[\rep{X \otimes X^* \otimes Y \otimes Y^* \oplus (X \otimes X^*)^{\oplus 2}}{\GL(X) \times \GL(Y)}\] and \cref{lem:GL-balanced-to-sum} to reduce to the required
\[\rep{\left(X \otimes X^* \otimes Y \otimes Y^*\right)^{\oplus 3}}{\GL(X) \times \GL(Y)}.\qedhere\]
\end{proof}

Next, we prove completeness of tensor tuple sequences with actions of single general linear group.

\begin{theorem}\label{thm:complete X}
The orbit closure intersection problems for the following tensor tuple sequences are complete for $\OCIclass$:
    \begin{enumerate}
        \item $\rep{\left(X^{\otimes 2} \otimes (X^*)^{\otimes 2}\right) \oplus \left(X \otimes X^*\right)^{\oplus 2}}{\GL(X)}$
        \item $\rep{\left(X^{\otimes 2} \otimes (X^*)^{\otimes 2}\right)^{\oplus 2}}{\GL(X)}$
    \end{enumerate}
\end{theorem}
\begin{proof}
    Following the proof of the previous lemma up we obtain a reduction from an arbitrary tensor tuple sequence to
    $\rep{(X \otimes X^* \otimes Y \otimes Y^*)^{\oplus p}}{\GL(X) \times \GL(Y)}$ for some $p$.
    As the next step, apply \cref{lem:GL-sum-to-summand} to obtain a reduction to $\rep{X \otimes X^* \otimes Y \otimes Y^* \oplus (X \otimes X^*)^{\oplus 2}}{\GL(X) \times \GL(Y)}$ or
    $\rep{X \otimes X^* \otimes Y \otimes Y^* \oplus X^{\otimes 2} \otimes (X^*)^{\otimes 2}}{\GL(X) \times \GL(Y)}$.
    Finally, apply \cref{lem:GL-to-one-space} to reduce to the sequences listed in the statement of the theorem.
\end{proof}

We state one more lemma, which allows us to reduce the number of summands in a direct sum at the expense of increasing the number of tensor factors.

\begin{lemma}\label{lem:GL-sum-to-one}
    If $p \leq d!$, then there exists an efficient algebraic reduction from $\rep{(\tspace{X}{c}{c})^{\oplus p}}{\GL(X)}$ to $\rep{\tspace{X}{c + d}{c + d}}{\GL(X)}$.
\end{lemma}
\begin{proof}
    For every permutation $\pi \in \Symg_d$ consider the tensor $P_{\pi} \in \tspace{X}{d}{d}$ corresponding to the linear map $X^{\otimes d} \to X^{\otimes d}$ permuting the factors according to the permutation $\pi$. All tensors $P_\pi$ are $\GL(X)$-invariant and linearly independent.
    It follows that for any set of permutations $\{\pi_1, \dots, \pi_p\} \subseteq \Symg_d$, the map $\iota \colon (\tspace{X}{c}{c})^{\oplus p} \to \tspace{X}{c + d}{c + d}$ defined by $\iota(x_1, \dots, x_p) = \sum_{k = 1}^p x_k \otimes P_{\pi_k}$ is an injective $\GL(X)$-equivariant linear map, and hence it is preserves and reflects closure equivalence by \cref{ex:injective-morphism}.
    As before we obtain an efficient algebraic reduction.
\end{proof}

We thus obtain completeness for a natural tensor sequence (no tuples required).

\begin{theorem}\label{thm:complete-multiconj-4}
    The orbit closure intersection problem for $\rep{X^{\otimes 4} \otimes (X^*)^{\otimes 4}}{\GL(X)}$ is complete for~$\OCIclass$.
\end{theorem}
\begin{proof}
    By \cref{thm:complete X}, the orbit closure intersection problem for $\rep{\left(X^{\otimes 2} \otimes (X^*)^{\otimes 2}\right)^{\oplus 2}}{\GL(X)}$ is $\OCIclass$-complete.
    By \cref{lem:GL-sum-to-one} (with $p=c=d=2$), it reduces to $\rep{X^{\otimes 4} \otimes (X^*)^{\otimes 4}}{\GL(X)}$.
\end{proof}

Using a similar idea, we can also prove completeness for another simple balanced sequence.

\begin{theorem}\label{thm:complete-multiconj-31}
    The orbit closure intersection problem for $\rep{X^{\otimes 3} \otimes (X^*)^{\otimes 3} \oplus X \otimes X^*}{\GL(X)}$ is complete for~$\OCIclass$.
\end{theorem}
\begin{proof}
    We will construct an efficient algebraic reduction from $\rep{X^{\otimes 2} \otimes (X^*)^{\otimes 2} \oplus \left(X \otimes (X^*)\right)^{\oplus 2}}{\GL(X)}$, which has $\OCIclass$-complete orbit closure intersection problem by \cref{thm:complete X}.
    Let $P_{(12)} \in \tspace{X}{2}{2}$ be the invariant tensor corresponding to the linear map $X^{\otimes 2} \to X^{\otimes 2}$ swapping two tensor factors.
    Consider the linear map
    \[\alpha \colon X^{\otimes 2} {\otimes} (X^*)^{\otimes 2} \oplus \left(X {\otimes} (X^*)\right)^{\oplus 2} \to X^{\otimes 3} {\otimes} (X^*)^{\otimes 3} \oplus X {\otimes} X^*\]
    given by $\alpha(x, y, z) = (x \otimes \Id + y \otimes P_{(12)}, z)$.
    Note that if $\dim X \geq 2$, then the subspaces $\tspace{X}{2}{2} \otimes \Id$ and $\tspace{X}{1}{1} \otimes P_{(12)}$ of $\tspace{X}{3}{3}$ intersect only in $0$.
    Indeed, consider $\tspace{X}{3}{3}$ as a tensor product $\tspace{X}{2}{2} \otimes \tspace{X}{1}{1}$.
    The matrix rank of a nonzero tensor of the form $x \otimes \Id$ with respect to this decomposition is $1$, but the rank of a nonzero tensor of the form $y \otimes P_{(12)}$ is equal to $\dim X$.
    Therefore, the sum $\tspace{X}{2}{2} \otimes \Id + \tspace{X}{1}{1} \otimes P_{(12)}$ is a direct sum and the map $\alpha$ is an injective equivariant map.
    By~\Cref{ex:injective-morphism} $\alpha$ preserves and reflects closure equivalence.

    We almost have an efficient algebraic reduction: the case $\dim X = 1$ needs to be considered separately.
    For $\dim X = 1$ the action of $\GL(X)$ on $X^{\otimes 2} \otimes (X^*)^{\otimes 2} \oplus \left(X \otimes (X^*)\right)^{\oplus 2}$ is trivial,
    so we can take in the definition of the algebaric reduction $q(1) = 2$ and the first map in the sequence
    $\alpha_1(x \Id^{\otimes 2}, y \Id, z \Id) = (x \Id^{\otimes 3} + y (\Id \otimes P_{(12)}), z \Id)$, which is again an injective equivariant map (mapping invariants to invariants).
\end{proof}

\Cref{thm:complete X,thm:complete-multiconj-4,thm:complete-multiconj-31} prove $\OCIclass$-completeness of four OCI problems for four simple balanced tensor tuple sequences.
We can extend this completeness to all balanced sequences which are more complicated than these using the following lemma.

\begin{lemma}\label{lem:balanced-GL-increase}
  If $p \leq q$ and $a_k \leq b_k$ for all $k \in [p]$, then there is an efficient algebraic reduction from  
  \[\rep{\bigoplus_{k = 1}^p V^{\otimes a_k} \otimes (V^*)^{\otimes a_k}}{\GL(V)}\] to \[\rep{\bigoplus_{k = 1}^{q} V^{\otimes b_k} \otimes (V^*)^{\otimes b_k}}{\GL(V)}.\]
\end{lemma}
\begin{proof}
  The proof is similar to \cref{lem:GL-balanced-to-sum}.
  The map
  \begin{equation*}
    \alpha \colon \bigoplus_{k = 1}^p V^{\otimes a_k} \otimes (V^*)^{\otimes a_k} \to \bigoplus_{k = 1}^{q} V^{\otimes b_k} \otimes (V^*)^{\otimes b_k}
  \end{equation*}
  sending $(x_1, \dots, x_p)$ to
  \begin{equation*}
    (x_1 \otimes \Id^{\otimes (b_1 - a_1)}, \dots, x_p \otimes \Id^{\otimes (b_p - a_p)}, \Id^{\otimes b_{p + 1}}, \dots, \Id^{\otimes b_q})
  \end{equation*}
  is an injective equivariant map, and therefore preserves and reflects closure equivalence.
  Collecting these maps into a sequence, we obtain an efficient algebraic reduction.
\end{proof}

\begin{corollary}\label{cor:complete-balanced-GL}
  The OCI problem for a tensor tuple sequence $\rep{\bigoplus_{k = 1}^q V^{\otimes a_k} \otimes (V^*)^{\otimes a_k}}{\GL(V)}$ is complete for $\OCIclass$ if $\sum_{k = 1}^q a_k \geq 4$ and at least one $a_k > 1$.
\end{corollary}
\begin{proof}
  The four minimal cases with $\sum_{k = 1}^q a_k = 4$ are established in \cref{thm:complete X,thm:complete-multiconj-4,thm:complete-multiconj-31}, the rest follows by \cref{lem:balanced-GL-increase}.
\end{proof}

\subsection{Reductions and completeness for orthogonal and symplectic groups}\label{subsec:o sp}
Let $g \in V^* \otimes V^*$ be a symmetric or symplectic form on the vector space $V$.
Interpreting $g$ as a map from $V$ to $V^*$, we get an invertible linear map equivariant with respect to the group~$\Og(V)=\Og(V,g)$ or $\Sp(V)=\Sp(V,g)$, respectively.
Therefore, when talking about tensor tuple representations of~$\Og(V)$ and~$\Sp(V)$, it is enough to consider tensor tuple representations $\bigoplus_{k = 1}^p \vtuple{V}^{\otimes \ituple{a}_p}$ that do not involve dual spaces.

\begin{theorem}\label{thm:O-to-GL}
    For every $\ituple{a}_1, \dots, \ituple{a}_p \in \bbN^m$, there exists an efficient algebraic reduction from
    $\rep{\bigoplus_{k = 1}^p \vtuple{V}^{\otimes \ituple{a}_p}}{\Og(\vtuple{V})}$
    to $\rep{\bigoplus_{k = 1}^p \vtuple{V}^{\otimes \ituple{a}_p} \oplus \bigoplus_{i = 1}^m (V_i^*)^{\otimes 2}}{\GL(\vtuple{V})}$.
\end{theorem}
\begin{proof}
    Let $g_i \in V_i^* \otimes V_i^*$ denote the tensor corresponding to the symmetric bilinear form invariant under $\Og(V_i)$.
    It is clear that~$\alpha(x_1,\dots,x_p) = (x_1, \dots, x_p, g_1, \dots, g_m)$ defines an equivariant map with respect to the inclusion~$\beta \colon \Og(\vtuple{V}) \to \GL(\vtuple{V})$, hence~$\alpha$ preserves closure equivalence by \cref{lem:beta crit}.

    By \cref{lem:o sp invariants}, the invariants for $\rep{\bigoplus_{k = 1}^p \vtuple{V}^{\otimes \ituple{a}_p}}{\Og(\vtuple{V})}$ are linearly spanned by the contraction invariants
    \[
        G(x_1, \dots, x_p) = \bigotimes_{i = 1}^m \Tr_{\pi_i} \left(\bigotimes_{k = 1}^p x_k^{d_k} \otimes \bigotimes_{i = 1}^m g_i^{\otimes f_i}\right).
    \]
    Note that $G = F \circ \alpha$ where
    \[
        F(y_1, \dots, y_p, z_1, \dots, z_m) = \bigotimes_{i = 1}^m \Tr_{\pi_i} \left(\bigotimes_{k = 1}^p y_k^{d_k} \otimes \bigotimes_{i = 1}^m z_i^{\otimes f_i}\right)
    \]
    is a contraction invariant for \[\rep{\bigoplus_{k = 1}^p \vtuple{V}^{\otimes \ituple{a}_p} \oplus \bigoplus_{i = 1}^m (V_i^*)^{\otimes 2}}{\GL(\vtuple{V})}.\]
    Thus $\alpha$ reflects closure equivalence by \cref{thm:reflect}.
    The map $\alpha$ is clearly efficiently computable and hence we obtain an efficient algebraic reduction.
\end{proof}

We also have an opposite reduction from orbit closure intersection problems for general linear groups to similar problems for orthogonal groups.

\begin{theorem}\label{thm:GL-to-O}
    Let $\tworow{\ituple{a}_1}{\ituple{b}_1}, \dots, \tworow{\ituple{a}_p}{\ituple{b}_p} \in \bbN^{m} \times \bbN^{m}$ and define~$\ituple{c}_k = \ituple{a}_k + \ituple{b}_k$.
    Then there is an efficient algebraic reduction from $\rep{\bigoplus_{k = 1}^p \tspace{\vtuple{V}}{\ituple{a}_k}{\ituple{b}_k}}{\GL(\vtuple{V})}$ to $\rep{\bigoplus_{k = 1}^p \vtuple{W}^{\otimes \ituple{c}_k}}{\Og(\vtuple{W})}$.
\end{theorem}
\begin{proof}
    As mentioned at the beginning of the section, $W_i \cong W_i^*$ as representations of $\Og(\vtuple{W})$.
    Therefore is suffices to construct a reduction from $\rep{\bigoplus_{k = 1}^p \tspace{\vtuple{V}}{\ituple{a}_k}{\ituple{b}_k}}{\GL(\vtuple{V})}$ to $\rep{\bigoplus_{k = 1}^p \tspace{\vtuple{W}}{\ituple{a}_k}{\ituple{b}_k}}{\Og(\vtuple{W})}$, as the latter is isomorphic to $\rep{\bigoplus_{k = 1}^p \vtuple{W}^{\otimes c_k}}{\Og(\vtuple{W})}$.

    Let $\vtuple{V} = (V_1, \dots, V_m)$ be a tuple of vector spaces, and let $W_i = V_i \oplus V_i^*$.
    Equip each space $W_i$ with a symmetric bilinear form $g_i \in W_i^* \otimes W_i^*$ given by \[g_i((v_1, w_1), (v_2, w_2)) = w_1 \cdot v_2 + w_2 \cdot v_1\]
    and consider the corresponding groups $\Og(W_i)$.
    For $f \in \GL(V_i)$ we have $f \oplus (f^*)^{-1} \in \Og(W_i)$, which gives a homomorphism $\beta_i \colon \GL(V_i) \to \Og(W_i)$.
    Combining $\beta_i$ for all $i \in [m]$, we get a homomorphism~$\beta\colon \GL(\vtuple{V}) \to \Og(\vtuple{W})$.

    Let $\alpha_i \colon V_i \to W_i$ be the inclusions of $V_i$ into $W_i$ and let $\bar\alpha_i \colon V_i^* \to W_i^*$ be the inclusion of $V_i^*$ into~$W_i^* \cong V_i^* \oplus V_i$.
    They induce linear maps $\alpha_{\tworow{\ituple{a}}{\ituple{b}}}\colon \tspace{\vtuple{V}}{\ituple{a}}{\ituple{b}} \to \tspace{\vtuple{W}}{\ituple{a}}{\ituple{b}}$ for all $\tworow{\ituple{a}}{\ituple{b}}$,
    and we obtain a linear map~$\alpha$ from $\bigoplus_{k = 1}^p \tspace{\vtuple{V}}{\ituple{a}_k}{\ituple{b}_k}$ to $\bigoplus_{k = 1}^p \tspace{\vtuple{W}}{\ituple{a}_k}{\ituple{b}_k}$ as $\alpha = \bigoplus_{k = 1}^p \alpha_{\tworow{\ituple{a}_k}{\ituple{b}_k}}$.
    The map~$\alpha$ is equivariant with respect to~$\beta$ and hence it preserves closure equivalence, by \cref{lem:beta crit}.
    To see that it also reflects closure equivalence, consider any contraction invariant for $\rep{\bigoplus_{k = 1}^p \tspace{\vtuple{V}}{\ituple{a}_k}{\ituple{b}_k}}{\GL(\vtuple{V})}$,
    \[
    G(x_1, \dots, x_p) = \bigotimes_{i = 1}^m \Tr_{\pi_i} \left(\bigotimes_{k = 1}^p x_k^{\ot d_k}\right).
    \]
    Note that for $v_i \in V_i$ and $w_i \in V_i^*$ we have $\bar\alpha_i(w_i) \cdot \alpha_i(v_i) = w_i \cdot v_i$.
    By considering elementary tensors~$x_k$ we obtain
    \[
    \bigotimes_{i = 1}^m \Tr_{\pi_i} \left(\bigotimes_{k = 1}^p x_k^{\ot d_k}\right) = \bigotimes_{i = 1}^m \Tr_{\pi_i} \left(\bigotimes_{k = 1}^p \alpha_{\tworow{a_k}{b_k}}x_k^{\ot d_k}\right),
    \]
    that is, $G = F \circ \alpha$ where $F$ is a contraction invariant
    \[
    F(y_1, \dots, y_p) = \bigotimes_{i = 1}^m \Tr_{\pi_i} \cdot \left(\bigotimes_{k = 1}^p y_k^{\ot d_k}\right)
    \]
    for $\rep{\bigoplus_{k = 1}^p \tspace{\vtuple{W}}{\ituple{a}_k}{\ituple{b}_k}}{\Og(\vtuple{W})}$.
    Thus $\alpha$ reflects closure equivalence by \cref{thm:reflect} and we obtain an efficient algebraic reduction.
\end{proof}

\begin{corollary}
We have $\OCIclass_\Og = \OCIclass$.
\end{corollary}

We can also identify orbit closure problems for the orthogonal groups that are complete for the complexity class~$\OCIclass$.
For example, \cref{thm:complete-multiconj-4,thm:GL-to-O} combine to show that $\rep{X^{\otimes 8}}{\Og(X)}$ is complete.
With some more work we can show that seven copies suffice for completeness.
We start with the following theorem, which is of independent interest.

\begin{theorem}\label{thm:o 3 2}
There is an efficient algebraic reduction from $\rep{X^{\otimes 4} \otimes (X^*)^{\otimes 4}}{\GL(X)}$ to $\rep{(Y^{\otimes 3})^{\oplus 2}}{\Og(Y)}$.
As a consequence, the orbit closure intersection problem for the latter is complete for $\OCIclass$.
\end{theorem}
\begin{proof}
    Recall from \cref{thm:complete-multiconj-4} that the orbit closure intersection problem for $\rep{X^{\otimes 4} \otimes (X^*)^{\otimes 4}}{\GL(X)}$ is complete for~$\OCIclass$.
    Using \cref{lem:GL-balanced-to-sum} we can reduce to $\rep{X^{\otimes 6} \otimes (X^*)^{\otimes 6}}{\GL(X)}$.
    Then apply \cref{thm:GL-to-O} to reduce to $\rep{X^{\otimes 12}}{\Og(X)}$.
    To prove the theorem, it suffices to give a reduction to $\rep{Y^{\otimes 3} \;\oplus\; Y^{\otimes 2} \otimes Y^*}{\Og(Y)}$, as the latter is isomorphic to~$(Y^{\ot 3})^{\oplus 2}$ as a $\Og(Y)$-representation.

    Thus let $X$ be a vector space equipped with a symmetric bilinear form $g \in X^* \otimes X^*$, and let $\bar{g} \in X \otimes X$ be the tensor invariant under the corresponding orthogonal group $\Og(X)$ (it is the tensor of the dual bilinear form on $X^*$).
    Denote $Y = X^{\otimes 4}$.
    Let $\beta$ be the diagonal inclusion of~$\Og(X)$ into~$\Og(Y)$ where the orthogonal structure on $Y = X^{\otimes 4}$ is induced by the structure given by $g$ on $X$.
    In terms of tensors, the symmetric bilinear form on $Y$ is $h =  \sigma \cdot g^{\otimes 4} \in (X^*)^{\otimes 8} = Y^* \otimes Y^*$ where
    \[ \sigma = \begin{pmatrix}1 & 2 & 3 & 4 & 5 & 6 & 7 & 8 \\ 1 & 5 & 2 & 6 & 3 & 7 & 4 & 8\end{pmatrix} \in \Symg_8. \]

    Consider the tensors $r \in \tspace{X}{4}{2}$ and $s \in \tspace{X}{8}{4}$ defined as $r = ((23), \id) \cdot \Id[X]^{\otimes 2} \otimes \bar{g}$ and $s = r^{\otimes 2}$.
    For a tensor $x \in X^{\otimes 12}$ denote by $\hat{x}$ the same tensor considered as an element of $Y^{\otimes 3}$.
    Similarly, denote by $\hat{s}$ the tensor $s$ viewed as an element of $Y^{\otimes 2} \otimes Y^* \cong \tspace{X}{8}{4}$.
    Then we can define a map
    \[ \alpha \colon X^{\otimes 12} \to Y^{\otimes 3} \;\oplus\; Y^{\otimes 2} \otimes Y^*, \quad \alpha(x) = (\hat{x}, \hat{s}). \]
    Note that $s$ is invariant under $\Og(X)$.
    It follows that the map $\alpha$ is $\Og(X)$-equivariant with respect to~$\beta$ and hence~$\alpha$ preserves closure equivalence by \cref{lem:beta crit}.

    Consider the contraction $t = \Tr^{1,2:6}_{2,3:3} \hat{s}^{\otimes 3} \in Y^{\otimes 4} \otimes Y^*$.
    As a tensor in~$\tspace{X}{16}{4}$ it is equal to~$(\tau, \id) \cdot (\Id[X]^{\otimes 4} \otimes \bar{g}^{\otimes 6})$, where
    \[\tau = \begin{psmallmatrix} 1 & 2 & 3 & 4 & 5 & 6 & 7 & 8 & 9 & 10 & 11 & 12 & 13 & 14 & 15 & 16 \\ 1 & 5 & 9 & 13 & 3 & 7 & 11 & 15 & 2 & 4 & 6 & 8 & 10 & 12 & 14 & 16\end{psmallmatrix} \in \Symg_{16}\]
    Note in particular that the indices belonging to the first $4$ factors ($\Id[X]$ tensors) are permuted to $1$, $5$, $9$, and $13$, that is, if we divide the set $[16]$ into four quadruples corresponding to the $Y$-factors of $Y^{\otimes 4} = X^{\otimes 16}$, then $\tau$ maps the first four indices into the first indices of the four quadruples.

    \begin{figure}[ht]
    \centering
    \begin{tikzpicture}[>=latex]
        \node (r) at (-1.5,1) {\Large $r$};
        \node at (-0.9,1) {$=$};
        \draw[<-] (0,0)--(0,2);
        \draw[<-] (2,0)--(2,2);
        \node[draw,circle] (g1) at (1,1) {$\bar{g}$};
        \draw[<-] (1,0)--(g1.south);
        \draw[<-] (3,0)--(g1.south east);
    \end{tikzpicture}
    \\
    \vspace{1em}
    \begin{tikzpicture}[>=latex,scale=0.9]
        \node (s) at (-1.5,1) {\Large $s$};
        \node at (-0.9,1) {$=$};
        \draw[<-] (0,0)--(0,2);
        \draw[<-] (2,0)--(2,2);
        \draw[<-] (4,0)--(4,2);
        \draw[<-] (6,0)--(6,2);
        \node[draw,circle] (g1) at (1,1) {$\bar{g}$};
        \draw[<-] (1,0)--(g1.south);
        \draw[<-] (3,0)--(g1.south east);
        \node[draw,circle] (g2) at (5,1) {$\bar{g}$};
        \draw[<-] (5,0)--(g2.south);
        \draw[<-] (7,0)--(g2.south east);
        \draw[dashed] (-0.25,0.25)--(2.25,0.25)--(2.25,1.75)--(-0.25,1.75)--cycle;
        \draw[dashed] (3.75,0.25)--(6.25,0.25)--(6.25,1.75)--(3.75,1.75)--cycle;
        \node[fill=white] at (1,1.75) {$r$};
        \node[fill=white] at (5,1.75) {$r$};
    \end{tikzpicture}
    \\
    \vspace{1em}
    \begin{tikzpicture}[>=latex,scale=0.85]
        \node (s) at (-1.5,0) {\Large $t$};
        \node at (-0.9,0) {$=$};

        \draw[-] (0,0)--(0,2);
        \draw[-] (2,0)--(2,2);
        \draw[-] (4,0)--(4,2);
        \draw[-] (6,0)--(6,2);
        \node[draw,circle] (g1) at (1,1) {$\bar{g}$};
        \draw[-] (1,0)--(g1.south);
        \draw[-] (3,0)--(g1.south east);
        \node[draw,circle] (g2) at (5,1) {$\bar{g}$};
        \draw[-] (5,0)--(g2.south);
        \draw[-] (7,0)--(g2.south east);
        \draw[dashed] (-0.25,0.25)--(6.25,0.25)--(6.25,1.75)--(-0.25,1.75)--cycle;

        \draw[->] (0,0)--(0,-2);
        \draw[->] (1,0)--(1,-2);
        \draw[->] (2,0)--(2,-2);
        \draw[->] (3,0)--(3,-2);
        \node[draw,circle] (g3) at (0.5,-1) {$\bar{g}$};
        \draw[<-] (0.5,-2)--(g3.south);
        \draw[<-] (1.5,-2)--(g3.south east);
        \node[draw,circle] (g4) at (2.5,-1) {$\bar{g}$};
        \draw[<-] (2.5,-2)--(g4.south);
        \draw[<-] (3.5,-2)--(g4.south east);
        \draw[dashed] (-0.25,-0.25)--(3.25,-0.25)--(3.25,-1.75)--(-0.25,-1.75)--cycle;

        \draw[->] (4,0)--(4,-2);
        \draw[->] (5,0)--(5,-2);
        \draw[->] (6,0)--(6,-2);
        \draw[->] (7,0)--(7,-2);
        \node[draw,circle] (g5) at (4.5,-1) {$\bar{g}$};
        \draw[<-] (4.5,-2)--(g5.south);
        \draw[<-] (5.5,-2)--(g5.south east);
        \node[draw,circle] (g6) at (6.5,-1) {$\bar{g}$};
        \draw[<-] (6.5,-2)--(g6.south);
        \draw[<-] (7.5,-2)--(g6.south east);
        \draw[dashed] (3.75,-0.25)--(7.25,-0.25)--(7.25,-1.75)--(3.75,-1.75)--cycle;

        \node[fill=white] at (3,1.75) {$s$};
        \node[fill=white] at (1.5,-0.25) {$s$};
        \node[fill=white] at (5.5,-0.25) {$s$};
    \end{tikzpicture}
    \caption{Diagrammatic representations of the tensors $r$, $s$, and $t$ from the proof of~\cref{thm:o 3 2}}
    \end{figure}
    
    Similarly, $t^{\otimes 3} \in \tspace{X}{48}{12} = (\tau', \id) \cdot \Id[X]^{\otimes 12} \otimes \bar{g}^{\otimes 18}$ where $\tau'$ maps the first $12$ indices into the first indices of the $12$ quadruples.
    More precisely, $\tau'$ can be expressed with the help of the the inclusion $\times \colon \Symg_m \times \Symg_n \to \Symg_{mn}$ where the first permutation permutes blocks of size $n$ and the second permutation permutes each block. Indeed, $\tau' = (\id_3 \times \tau)(\rho \times \id_4)$ where
    \[\rho = \begin{pmatrix}1 & 2 & 3 & 4 & 5 & 6 & 7 & 8 & 9 & 10 & 11 & 12 \\ 1 & 5 & 9 & 2 & 6 & 10 & 3 & 7 & 11 & 4 & 8 & 12\end{pmatrix} \!\in\! \Symg_{12}\]
    Note also that the permutation $\sigma$ above also shares this property --- it maps the first two indices to the first indices of the two quadruples in~$[8]$.

    Now consider an arbitrary contraction invariant $G(x)$ for $\rep{X^{\otimes 12}}{\Og(X)}$ given by
    \[
        G(x) = \Tr_{\pi} (x^{\otimes d} \otimes g^{\otimes 6d}).
    \]
    We can obtain a scalar multiple of $G$ using only tensors over $Y$ as follows:
    \begin{align*}
    &\Tr_{\pi} \left((t^{\otimes 3} \cdot \hat{x})^{\otimes d} \otimes h^{\otimes 6d}\right)
    \\
    = &\Tr_{\pi \times \id_4} \left((\tau' \cdot x \otimes \bar{g}^{\otimes 18})^{\otimes d} \otimes (\id_{4d} \times \sigma) \cdot g^{\otimes 24d}\right) \\
    = &\Tr_{{(\pi \times \id_4)(\id_d \times \tau')}} \left((x \otimes \bar{g}^{\otimes 18})^{\otimes d} \otimes (\id_{6d} \times \sigma) \cdot g^{\otimes 24d}\right).
    \end{align*}
    Since $\tau'$ maps factors of the tensor $x$ into the first factors in each quadruple and $\pi \times \id_4$ permutes quadruples without permuting their elements, the factors of $x$ are mapped to first factors of each quadruple according to the permutation $\pi$.
    Additionally, $\sigma$ maps the factors of the first $g$ in each of $6d$ blocks to the first factors of each quadruple, so we can separate $x$ and $6d$ copies of $g$, rewriting the above expression as
    \[
        \Tr_{\pi} (x^{\otimes d} \otimes g^{\otimes 6d}) \cdot \Tr_{\pi'} (\bar{g}^{\otimes 18 d} \otimes g^{\otimes 18 d}).
    \]
    Using the facts that $\Tr^i_j \bar{g} \otimes g = \Id$ and $\Tr \Id = \dim X$, we see that this expression is $(\dim X)^c G(x)$ for some constant $c$.
    Thus we see that $G = F \circ \alpha$ for the contraction invariant
    \[ F(x, z) = \Tr_{\pi} \left(((\Tr^{1,2:6}_{2,3:3} \hat{z}^{\otimes 3})^{\otimes 3}\cdot \hat{x})^{\otimes d} \otimes h^{\otimes 6d}\right) \]
    for $\rep{Y^{\otimes 3} \oplus Y^{\otimes 2} \otimes Y^*}{\Og(Y)}$.
    By \cref{thm:reflect}, we conclude that~$\alpha$ reflects closure equivalence and we obtain an efficient algebraic reduction.
\end{proof}

\begin{corollary}
    The orbit closure intersection problem for $\rep{\left(Y^{\otimes 3}\right)^{\oplus 2} \oplus (Y^*)^{\otimes 2}}{\GL(Y)}$ is complete for~$\OCIclass$.
\end{corollary}
\begin{proof}
This follows from \cref{thm:o 3 2,thm:O-to-GL}.
\end{proof}

We now give a variant of \cref{lem:GL-sum-to-one} that can be used to replace copies by additional tensor factors for actions of the orthogonal group.

\begin{lemma}\label{lem:detupelize o}
    If $p \leq d!$, then there is an efficient algebraic reduction from $\rep{(X^{\otimes c})^{\oplus p}}{\Og(X)}$ to $\rep{X^{\otimes (c + 2d)}}{\Og(X)}$.
\end{lemma}
\begin{proof}
    As in \cref{lem:GL-sum-to-one} the permutation operators $P_{\pi}\colon X^{\otimes d} \to X^{\otimes d}$ are linearly independent $\Og(X)$-invariants.
    Then the maps $\alpha(x_1, \dots, x_p) = \sum_{k = 1}^p x_k \otimes P_{\pi_k}$ are injective $\Og(X)$-linear maps $(X^{\otimes c})^{\oplus p}$ to $\tspace{X}{c + d}{d} \cong X^{\otimes (c + 2d)}$, which combine to form an efficient algebraic reduction.
\end{proof}

\begin{theorem}
    The orbit closure intersection problem for $\rep{Y^{\otimes 7}}{\Og(Y)}$ is complete for $\OCIclass$.
\end{theorem}
\begin{proof}
    Apply \cref{lem:detupelize o} to to $\rep{\left(Y^{\otimes 3}\right)^{\oplus 2}}{\Og(Y)}$, which is complete according to \cref{thm:o 3 2}.
\end{proof}

All constructions presented in this section also apply to tensor tuple sequences for the symplectic groups $\Sp(\vtuple{V})$ if $g_i$ are taken to denote tensor corresponding to symplectic forms instead of symmetric ones.
In \cref{thm:GL-to-O} the symplectic form on $U_i = V_i \oplus V_i^*$ can be defined as $g((v_1, w_1),(v_2, w_2)) = w_1 \cdot v_2 - w_2 \cdot v_1$.
We summarize these results:

\begin{theorem}
We have $\OCIclass_\Sp = \OCIclass$.
Moreover, the orbit closure intersection problems for $\rep{(V^{\otimes 3})^{\oplus 2}}{\Sp(V)}$ and for $\rep{V^{\otimes 7}}{\Sp(V)}$ are each complete for this class.
\end{theorem}

\subsection{Reductions for symmetric groups and \texorpdfstring{$\GIclass$}{GI}-hardness}\label{subsec:gi}
In this section we consider the action of the symmetric group $\Symg_n$ on an $n$-dimensional space $\bbC^n$ and tensor spaces constructed from it.
Since the symmetric group is finite, the orbits are always closed, so closure equivalence is just equivalence under the action of the group.
The graph isomorphism problem can be easily presented as an equivalence problem for an action of the symmetric group.
We prove that the latter can be reduced to an orbit closure intersection problem for a tensor tuple sequence for the general linear group.
As a consequence, we find that $\GIclass \subseteq \OCIclass$.

\begin{lemma}
    The graph isomorphism problem $\GI$ reduces to $\OCI\rep{(\bbC^n)^{\otimes 2}}{\Symg_n}$.
\end{lemma}
\begin{proof}
    The reduction is given by mapping a pair of graphs to the corresponding pair of adjacency matrices.
    In more detail, let $(G_1, G_2)$ be an instance of the graph isomorphism problem.
    We can assume that the graphs~$G_1$ and~$G_2$ have the same number of vertices~$n$ and that the vertex sets of both graphs are~$[n]$, and we denote the edge sets by~$E_1$ and $E_2$.
    The reduction maps this data to the instance~$(n,x_1,x_2)$ for~$\OCI\rep{(\bbC^n)^{\otimes 2}}{\Symg_n}$, where~$x_1 = \sum_{(i,j) \in E_1} e_i \otimes e_j$ and~$x_2 = \sum_{(i,j) \in E_2} e_i \otimes e_j$.

    It is clear that the graphs $G_1$ and $G_2$ are isomorphic if and only if $x_1$ and $x_2$ are equivalent under the action of $\Symg_n$.
    Indeed, the graph isomorphisms from~$G_1$ to~$G_2$ are exactly the permutations~$\sigma \in \Symg_n$ such that $(i, j) \in E_1 \Leftrightarrow (\sigma(i), \sigma(j)) \in E_2$, or, equivalently, $\sigma \cdot x_1 = \sum_{(i, j) \in E_1} e_{\sigma(i)} \otimes e_{\sigma(j)} = x_2$ in terms of the group action~$\Symg_n$ on~$(\bbC^n)^{\otimes 2}$.
    Since $\Symg_n$ is a finite group, $x_1$ and $x_2$ are equivalent if and only if they are closure equivalent.
    Thus we have the desired reduction.
\end{proof}

\begin{remark}
    The same argument shows that $\OCI\rep{(\bbC^n)^{\ot2}}{\Symg_n}$ is equivalent to the graph isomorphism problem for weighted graphs with edge weights in~$\bbC \setminus \{0\}$.
\end{remark}

\begin{proposition}
    There is an efficient algebraic reduction from $\rep{(\bbC^n)^{\otimes 2}}{\Symg_n}$ to $\rep{V^{\otimes 2} \oplus V^{\otimes 3} \oplus (V^*)^{\otimes 2}}{\GL(V)}$.
\end{proposition}
\begin{proof}
    For each $V = \bbC^n$, define an injective linear map~$\alpha_n \colon V^{\otimes 2} \to V^{\otimes 3} \oplus V^{\otimes 2} \oplus (V^*)^{\otimes 2}$ by $\alpha_n(x) = (h, x, g)$, where~$g = \sum_{i = 1}^n e_i^* \otimes e_i^*$ and $h = \sum_{i = 1}^n e_i \otimes e_i \otimes e_i$.
    Let $\beta\colon \Symg_n \to \GL(V)$ be the inclusion of $\Symg_n$ into $\GL(V)$ as permutation matrices.
    Because the tensors~$g$ and~$h$ are $\Symg_n$-invariant, we have~$\alpha_n(\group{g}x) = \beta(\group{g})\alpha_n(x)$ for $\group{g} \in \Symg_n$.
    Thus $\alpha_n$ preserves closure equivalence by \cref{lem:beta crit}.

    To see that it also reflects closure equivalence, recall from \cref{lem:sym invariants,lem:gl invariants} that the invariant algebra~$\bbC[(\bbC^n)^{\otimes 2}]^{\Symg_n}$ is linearly spanned by the contraction invariants of the form
    \[
        G(x) = \Tr_{\pi} \left(x^{\otimes d} \otimes h^{\otimes f_1} \otimes g^{\otimes f_2} \right).
    \]
    For every such $G \in \bbC[(\bbC^n)^{\otimes 2}]^{\Symg_n}$ we can find an invariant~$F \in \bbC[V^{\otimes 2} \oplus V^{\otimes 3} \oplus (V^*)^{\otimes 2}]^{\GL(V)}$, namely,
    \[
        F(x, y, z) = \Tr_{\pi} \left(x^{\ot d} \ot y^{\ot f_1} \ot z^{\ot f_2} \right)
    \]
    such that $G = F \circ \alpha_n$.
    We conclude from \cref{thm:reflect} that the map~$\alpha_n$ reflects closure equivalence.
    Thus the maps~$\alpha_n$ define an algebraic reduction.
    It is efficient because the maps are also efficiently computable (including as a function of~$n$).
\end{proof}

\begin{theorem}\label{thm:gi}
    The orbit closure intersection problem $\OCI\rep{V^{\otimes 2} \oplus V^{\otimes 3} \oplus (V^*)^{\otimes 2}}{\GL(V)}$ is $\GIclass$-hard.
    Therefore, $\GIclass \subseteq \OCIclass$.
    In particular, the complete orbit closure intersection problems listed in \S\ref{main:results} are all $\GIclass$-hard.
\end{theorem}

\subsection{Reductions for unitary groups}\label{subsec:u}
In this section we observe that the tensor isomorphism problem for the unitary group can be reduced to an orbit closure intersection problem for the general linear group with closed orbits.
The main ingredient is the following result, which we state in the language of this article.
In its statement the complex reductive group~$\group{G}$ is considered as a Lie group; as such it has a maximal compact subgroup~$\group{K}$ that is unique up to conjugation.
For example, if~$\group{G}$ is the general linear group~$\GL_n$ then~$\group{K}$ can be taken to be the unitary group~$\Ug_n$.

\begin{theorem}[{\cite[Theorem~8.2]{DBLP:conf/coco/BurgisserDMWW21}}]\label{thm:general U to G reduction}
Let $X$ be a representation of a connected reductive group~$\group{G}$.
Let $\group{K}$ be a maximal compact subgroup of~$\group{G}$, and choose a $\group{K}$-invariant Hermitian inner product~$\langle\cdot,\cdot\rangle$ on~$X$.
For $x\in X$, let~$\widehat{x} \in X^*$ denote the linear functional defined by~$\widehat{x}(x') := \langle x, x'\rangle$.
Then, the following are equivalent for any two~$x,y \in X$:
\begin{enumerate}
\item $x \sim_{\group{K}} y$;
\item $(x,\widehat{x}) \sim_{\group{G}} (y,\widehat{y})$ in $X \oplus X^*$;
\item $(x,\widehat{x}) \approx_{\group{G}} (y,\widehat{y})$ in $X \oplus X^*$.
\end{enumerate}
\end{theorem}

It is proved by using the Kempf--Ness theorem to observe that~$(x,\widehat{x}) \in X \oplus X^*$ is a minimum norm point in its~$\group{G}$-orbit, which therefore in particular is closed, and using the fact that the minimum norm points form a~$\group{K}$-orbit.

We remark that~$x \mapsto (x,\widehat{x})$ is not a polynomial map due to the antilinearity of~$x \mapsto \widehat{x}$.
Indeed, since the unitary group is \emph{not} an algebraic group (it is compact), this map is not an algebraic reduction between $p$-sequences of representations as defined in \cref{def:algebraic reduction}.
Nevertheless, \cref{thm:general U to G reduction} states it has exactly the right properties to relate orbit problems for~$\rep{X}{\group{K}}$ and~$\rep{X \op X^*}{\group{G}}$.
Therefore, if the map is efficiently computable and we obtain a Karp reduction between the corresponding computational problems.

In particular, this is case for tensor tuple sequences (but it holds more generally).
We summarize this in the following theorem, where $\OE$ denotes the \emph{orbit equivalence problem}, which is defined like~$\OCI$ but for equivalence~$\sim$ in place of closure equivalence~$\approx$ (\cref{def:orbit and equivalence}); this is also called the \emph{isomorphism problem}~\cite{DBLP:conf/innovations/ChenGQTZ24}.

\begin{theorem}\label{thm:U vs GL}
For every tensor tuple sequence \[X = \bigoplus_{k = 1}^p \tspace{\vtuple{V}}{\ituple{a}_k}{\ituple{b}_k}\] there is a Karp reduction from the orbit equality problem $\OE\rep{X}{\Ug(\vtuple{V})}$ to $\OCI\rep{X \op X^*}{\GL(\vtuple{V})}$, as well as to $\OE\rep{X \op X^*}{\GL(\vtuple{V})}$.
As a consequence, we have $\TIclass_\Ug \subseteq \OCIclass \cap \TIclass$.
\end{theorem}
\begin{proof}
The existence of the Karp reduction follows from \cref{thm:general U to G reduction} and the preceding discussion.

To see the inclusion of complexity classes, recall that the unitary group orbit equality problem~$\OE\rep{V^{\ot 3}}{\Ug(V)}$ is complete for~$\TIclass_\Ug$ by~\cite[Theorem~1.7]{DBLP:conf/innovations/ChenGQTZ24}.
Thus, it suffices to show that this problem is in~$\OCIclass$ as well as in~$\TIclass$.
The former follows directly from the first part of the proof, because $\OCI\rep{V^{\ot 3} \op (V^*)^{\ot 3}}{\GL(V)}$ is in~$\OCIclass$.
For the latter, we use that~$\OE\rep{V^{\ot 3} \op (V^*)^{\ot 3}}{\GL(V)}$ reduces to~$\OE\rep{V^{\ot 3}}{\GL(V)}$ by~\cite[Theorem~1.1]{Futorny-Grochow-Sergeichuk} and~\cite[Theorem~7.3]{DBLP:journals/siamcomp/GrochowQ23}, and the latter is complete for~$\TIclass$.
\end{proof}

\section*{Acknowledgments}
\addcontentsline{toc}{section}{Acknowledgments}
The authors thank Peter Bürgisser for many discussions on the complexity of algebraic problems.

We acknowledge support by the European Research Council through an ERC Starting Grant (grant agreement no.~101040907, SYMOPTIC).
MW also acknowledges support by the NWO through grant OCENW.KLEIN.267, by the BMBF through project Quantum Methods and Benchmarks for Resource Allocation (QuBRA), and by the Deutsche Forschungsgemeinschaft (DFG, German Research Foundation) under Germany's Excellence Strategy - EXC\ 2092\ CASA - 390781972.

\IEEEtriggeratref{43}
\bibliographystyle{IEEEtran}
\bibliography{main.bib}

\end{document}